\def\fullversion{1}

\documentclass[11pt]{article}
\usepackage{amstext}
\usepackage{amsmath}
\usepackage{amsthm}
\usepackage{amssymb}
\usepackage{bm}
\usepackage{wrapfig}
\usepackage{amsfonts, amsmath}
\usepackage[utf8]{inputenc}
\usepackage{graphicx}
\usepackage{bbm, comment}
\usepackage{subfigure}
\usepackage{color}
\usepackage{float}
\usepackage{wrapfig}
\usepackage{lipsum}
\usepackage{enumitem}

\usepackage[top=1.in, bottom=1.in, left=1.in, right=1.in]{geometry}

\usepackage[numbers]{natbib}
\newcommand{\E}{\mathbb{E}}

\newtheorem*{theorem*}{Theorem}
\newtheorem{theorem}{Theorem}[section]
\newtheorem{lemma}[theorem]{Lemma}
\newtheorem{fact}[theorem]{Fact}
\newtheorem{claim}[theorem]{Claim}
\newtheorem{assumption}[theorem]{Assumption}
\newtheorem{remark}[theorem]{Remark}
\newtheorem{notation}[theorem]{Notation}
\newtheorem{corollary}[theorem]{Corollary}
\newtheorem{example}[theorem]{Example}
\newtheorem{proposition}[theorem]{Proposition}

\newtheorem{definition}[theorem]{Definition}

\newenvironment{packed_enumerate}{
\begin{enumerate}
  \setlength{\itemsep}{1pt}
  \setlength{\parskip}{0pt}
  \setlength{\parsep}{0pt}
}{\end{enumerate}}

\newcommand{\strat}[2]{\left(\begin{array}{c} 0 \rightarrow #1 \\ 1 \rightarrow #2 \end{array} \right)}
\newcommand{\trul}[0]{\strat{0}{1}}
\newcommand{\fall}[0]{\strat{1}{0}}
\newcommand{\onel}[0]{\strat{1}{1}}
\newcommand{\zerl}[0]{\strat{0}{0}}

\newcommand{\tru}[0]{\mathbf{T}}
\newcommand{\fal}[0]{\mathbf{F}}
\newcommand{\one}[0]{\mathbf{1}}
\newcommand{\zer}[0]{\mathbf{0}}
\newcommand{\qstar}[0]{\mathbf{Q^*}}
\newcommand{\alleq}[0]{\mathbf{\Psi}}

\usepackage{xcolor}

\newcommand\gs[1]{{}}
\newcommand\jz[1]{{}}
\newcommand\kl[1]{{}}
\newcommand\katrina[1]{{}}
\newcommand\yk[1]{{}}

\newcommand\yuqing[1]{{}}

\begin{document}

%\markboth{Making Truth-telling Focal}{Making Truth-telling Focal}
\title{Putting Peer Prediction Under the Micro(economic)scope and \\Making Truth-telling Focal}
\date{}
%\date{October 31, 2014}
\author{Yuqing Kong\\ University of Michigan \and Grant Schoenebeck\\ University of Michigan \and Katrina Ligett\\ California Institute of Technology and Hebrew University }
%\author{Yuqing Kong, \affil{University of Michigan}\\
% Katrina Ligett, \affil{California Institute of Technology and Hebrew University}\\
%and Grant Schoenebeck \affil{University of Michigan}}

\maketitle

%\institute{Computer Science \& Engineering, University of Michigan}
\newcommand{\fix}{\marginpar{FIX}}
\newcommand{\new}{\marginpar{NEW}}

\begin{abstract}
    
Peer-prediction~\cite{MRZ05} is a (meta-)mechanism which, given any proper scoring rule, produces a mechanism to elicit privately-held, non-verifiable information from self-interested agents.   Formally, truth-telling is a strict Nash equilibrium of the mechanism.  Unfortunately, there may be other equilibria as well (including uninformative equilibria where all players simply report the same fixed signal, regardless of their true signal) and, typically, the truth-telling equilibrium does not have the highest expected payoff.  The main result of this paper is to show that, in the symmetric binary setting, by tweaking peer-prediction, in part by carefully selecting the proper scoring rule it is based on, we can make the truth-telling  equilibrium focal---that is, truth-telling has higher expected payoff than any other equilibrium.

%\gs{I liked the old one better.}\katrina{I don't mind, if you'd like to switch back. It seemed long to me. And actually I see now that the abstract is limited to one paragraph.}
%Along the way, we classify all equilibria of the peer-prediction mechanism for the binary signal setting. We also
%introduce a new technical tool, best response payoff plots, for understanding scoring rules.

Along the way, we prove the following:  in the setting where agents receive binary signals we 1)  classify all equilibria of the peer-prediction mechanism; 2) introduce a new technical tool for understanding scoring rules, which allows us to make truth-telling pay better than any other informative equilibrium;  3)  leverage this tool to provide an optimal version of the previous result; that is, we optimize the gap between the expected payoff of truth-telling and other informative equilibria; and 4)  show that with a slight modification to the peer-prediction framework, we can, in general, make the truth-telling equilibrium focal---that is, truth-telling pays more than any other equilibrium (including the uninformative equilibria).

\end{abstract}

%\category{C.2.2}{Focal Equilibria}{}

%\terms{}

%\keywords{peer prediction, agents, crowdsoursing}

%\acmformat{Yuqing Kong, Katrina Ligett, and Grant Schoenebeck, 2015. Putting Peer Prediction Under the Micro(economic)scope and Making Truth-telling Focal.}
% At a minimum you need to supply the author names, year and a title.
% IMPORTANT:
% Full first names whenever they are known, surname last, followed by a period.
% In the case of two authors, 'and' is placed between them.
% In the case of three or more authors, the serial comma is used, that is, all author names
% except the last one but including the penultimate author's name are followed by a comma,
% and then 'and' is placed before the final author's name.
% If only first and middle initials are known, then each initial
% is followed by a period and they are separated by a space.
% The remaining information (journal title, volume, article number, date, etc.) is 'auto-generated'.

%\begin{bottomstuff}
%{}
%\end{bottomstuff}

% =============================================================================
% =============================================================================

%\thispagestyle{empty}
%\setcounter{page}{0}
%\newpage

\katrina{There are still multiply defined labels (9etable)}

\section{Introduction}

 From Facebook.com's ``What's on your mind?" to Netflix's 5-point ratings, from innumerable survey requests in one's email inbox to Ebay's reputation system, user feedback plays an increasingly central role in our online lives.  This feedback can serve a variety of important purposes, including supporting product recommendations, scholarly research, product development, pricing, and purchasing decisions.  With increasing requests for information, agents must decide where to turn their attention.   When privately held information is elicited, sometimes agents may be intrinsically motivated to both participate and report the truth.  Other times, self-interested agents may need incentives to compensate for costs associated with truth-telling and reporting: the effort required to complete the rating (which could lead to a lack of reviews), the effort required to produce an accurate rating (which might lead to inaccurate reviews), foregoing the opportunity to submit an inaccurate review that could benefit the agent in future interactions~\cite{jurcafaltings06} (which could, e.g.,  encourage negative reviews), or a potential loss of privacy~\cite{ghosh2014buying} (which could encourage either non-participation or incorrect reviews).

To overcome a lack of (representative) reviews, a system could reward users for reviews.  However, this can create perverse incentives that lead to inaccurate reviews.  If agents are merely rewarded for participation, they may not take time to answer the questions carefully, or even meaningfully.

To this end, explicit reward systems for honest ratings have been developed.  If the ratings correspond to objective information that will be revealed at a future date, this information can be leveraged (e.g., via prediction markets) to incentive honesty.  In this paper, we study situations where this is not the case: the ratings cannot be independently verified  either because no objective truth exists (the ratings are inherently subjective) or an objective truth exists, but is not observable.

In such cases, it is known that correlation between user types can be leveraged to elicit truthful reports by using side payments~\cite{dAspremontG1979,dAspremontG1982,CremerM85,CremerM88}.  Miller, Resnick, and Zeckhauser~\cite{MRZ05} propose a particular such \mbox{(meta-)mechanism} for truthful feedback elicitation, known as {\emph{peer prediction}}.    Given any proper scoring rule (a simple class of payment functions we describe further below), and a prior where each agent's signal is ``stochastically relevant" (informative about other agents' signals), the corresponding peer prediction mechanism has truth-telling as a strict Bayesian-Nash equilibrium.

\yk{I rewrite from here since there are several work that deal with equilibrium multiplicity issue since 2015 (I comment the original one)}

\yuqing{I modified the previous pragraph to this short paragraph}

There is a major problem, however: alternative, non-truthful equilibria may have higher payoff for the agents than truth-telling. This is the challenge that our work addresses.

\paragraph{Our Results}
The main result of this paper is to show that by tweaking peer prediction, in part by specially selecting the proper scoring rule it is based on, we can make the truth-telling  equilibrium focal--that is, truth-telling has higher expected payoff than any other equilibrium.

%\gs{check these section nubmers}
Along the way we prove the following:  in the setting where agents receive binary signals we 1)  classify all equilibria of the peer prediction mechanism; 2) introduce a new technical tool, the best response plot, and use it to show that we can find proper scoring rules so the truth-telling pays more, in expectation, than any other informative equilibrium;  3)  we provide an optimal version of the previous result, that is we optimize the gap between the expected payoff of truth-telling and other informative equilibrium; and 4)  we show that with a slight modification to the peer prediction framework, we can, in general, make the truth-telling equilibrium focal---that is, truth-telling pays more than any other equilibrium (including the uninformative equilibria).

%Along the way we prove the following:  in the setting where agents receive binary signals we 1)  classify all equilibria of the peer prediction mechanism (Section~\ref{sec:equilibrium}); 2) introduce a new technical tool, the best response plot, and use it to show that we can find proper scoring rules so the truth-telling pays more, in expectation, than any other informative equilibrium (Section~\ref{sec:focal-for-some});  3)  we provide an optimal version of the previous result, that is we optimize the gap between the expected payoff of truth-telling and other informative equilibrium (Section~\ref{sec:optimization}); and 4)  we show that with a slight modification to the peer prediction framework, we can, in general, make the truth-telling equilibrium focal---that is, truthtelling pays more than any other equilibrium (including the uninformative equilibria) (Section~\ref{sec:punish}).

%and  5)  we extend these results beyond the binary case to more complex finite signal spaces by, in essence, a reduction back to the binary setting.

The main technical tool we use is a best response plot, which allows us to easily reason about the payoffs of different equilibria.  We first prove that no asymmetric equilibria exist. The naive approach then would be to simply plot the payoffs of different symmetric strategies.  However, for even the simplest proper-scoring rules, these payoff curves are paraboloid, and hence difficult to analyze directly.
The best response plot differs from this naive approach in two ways:  first, instead of plotting the strategies of agents explicitly, the best response plot  aggregates the results of these actions; second, instead of plotting the payoffs of all agents, the best response plot analyzes the payoff of one distinguished agent which, given the strategies of the remaining agents, plays her best response.  This makes the plot piece-wise linear for all proper scoring rules, which makes  analysis tractable. We hope that the best response plot will be useful in future work using proper scoring rules.

\subsection{Related Work}
\label{sec:related-work}

\yuqing{I have already modified the citet and cite thing, I can check one more time later }

Since the seminal work of Miller, Resnick, and Zeckhauser introducing peer prediction~\cite{MRZ05}, a host of results in closely related models have followed (see, e.g., ~\cite{jurca2007collusion,jurcafaltings09,jurcafaltings06,goelrp09}), primarily motivated by opinion elicitation in online settings where there is no objective ground truth.

%\gs{could shorten this for NIPs}
Recent research~\cite{GaoMCA2014} indicates that individuals in lab experiments do not always truth-tell when faced with peer prediction mechanisms; this may in part be related to the issue of equilibrium multiplicity.
~\citet{GaoMCA2014} ran studies over Mechanical Turk using two treatments: in the first they compensated the participants according to peer prediction payments, and in the second they gave them a flat reward for participation. In their work, the mechanism had complete knowledge of the prior.
The participants responded truthfully more often when the payoffs were fixed than in response to the peer prediction payments.  However, it should be noted that the task the agents were asked to perform took little effort (report the received signal), and the participants were not primed with any information about the truthful equilibrium of the peer prediction mechanism (they were only told the payoffs)--an actual surveyor would have incentive to prime the participants to report truthfully.

The most closely related work is a series of papers by \citet{jurca2007collusion,jurcafaltings09}, which studies collusion between the reporting agents.  In a weak model of collusion, the agents may be able to coordinate ahead of time (before receiving their signals) to select the equilibrium with the highest payoff.  Jurca and Faltings use techniques from algorithmic mechanism design to design a mechanism where, in most situations, the only symmetric \emph{pure} Nash equilibria are truth-telling.  They explicitly state the challenge of analysing mixed-Nash equilibrium as an open question, and show challenges to doing this in their  algorithmic mechanism design framework~\cite{jurca2007collusion,jurcafaltings09}.  Our techniques, in contrast,  allow us to analyse all  Nash equilibria of the peer prediction mechanism including both  mixed-strategy and asymmetric equilibria.  Instead of eliminating equilibria, we enforce that they have a lower expected payoff than truth-telling.  Additionally, the  algorithmic mechanism design framework  used by Jurca and Faltings sacrifices ``the simplicity of specifying the payments through closed-form scoring rules"~\cite{jurca2007collusion} that was present in the peer prediction paper.  Our work recovers a good deal of that simplicity.

Jurca and Faltings further analyze other settings where colluding agents can make transfer payments, or may collude after receiving their signals.  In particular, they again use automated mechanism design to show that in the case where agents coordinate after receiving their signals that even without transfer payments, there will always be multiple equilibria; in this setting, they pose the question of whether the truth-telling equilibrium can be endowed with the highest expected payoff.  We do not deal with this setting explicitly, but in the settings we consider, we show that even in the face of multiple equilibria, we can ensure that the truth-telling equilibrium has the highest expected payoff and no other equilibrium is paid the same with truth-telling.

In a different paper~\cite{jurcafaltings06}, Jurca and Faltings show how to minimize payments in the peer prediction framework.  Their goal is to discover how much ``cost" is associated with a certain marginal improvement of truth-telling over lying.  In this paper, they also consider generalizations of peer prediction, where more than one other agent's report is used as a reference.  Our work takes this to the extreme (as did~\cite{ghosh2014buying} before us) using \emph{all} of the \emph{other} agents' reports as references.

%\katrina{I commented out Goel, Reeves, Pennock, and just cited it above, since it didn't seem all that related.}
%The work of Goel, Reeves, and Pennock~\cite{goelrp09} extends to peer prediction setting to include cases where some agents have more reliable information than others.

\yk{I add citation for dasgupta2013crowdsourced paper}

\yuqing{I add prediction drawback of BTS}

A key motivation of one branch of the related work is removing the assumption that the mechanism knows the common prior~\cite{prelec2004bayesian,lambert2008truthful,dasgupta2013crowdsourced,kamble2015truth,radanovic2015incentive,faltings2014incentives,witkowski2013learning,witkowski2012peer,jurcafb08,jurca2011incentives,riley2014minimum}. \citet{dasgupta2013crowdsourced,kamble2015truth} have a different setting than us. In their setting, agents are asked to answer several a priori similar questions while our mechanism applies to one question (thus we do not need to assume the relation between questions). ~\citet{kamble2015truth}'s mechanism applies to both homogeneous and heterogeneous population but requires a large number of a priori similar tasks. However, \citet{kamble2015truth}'s mechanism contains non-truthful equilibria that are paid higher than truth-telling.~\citet{dasgupta2013crowdsourced}'s mechanism has truth-telling as the equilibrium with the highest payoff, but contains a non-truthful equilibrium that is paid as much as truth-telling.
%Except \cite{dasgupta2013crowdsourced}, all of these works suffer from multiple equilibria, and the truth-telling equilibrium is not known to be focal.  Usually (as in peer prediction) there exist other equilibria that reward, in expectation, more than truth-telling.  However,
~\citet{prelec2004bayesian} shows that in his Bayesian Truth Serum (BTS), truth-telling maximizes each individual's expected ``Information-score" across all equilibria. However, this guarantee is not strict, and requires the number of agents to be infinite, even to just have truth-telling be an equilibrium. Moreover, it is hard to classify the equilibria or optimize mechanism in Prelec's setting. Another drawback of BTS is that it requires agents to report prediction while our mechanism only requires agents to report a single signal. \citet{radanovic2015incentive}'s mechanism solves this drawback but that mechanism is in a sensing scenario and needs to compare the information of an sensor's local neighbours with the information of global sensors while our mechanism does not require this local/global structure. Moreover, like BTS, \citet{radanovic2015incentive}'s mechanism does not have strictness guarantee and requires the number of agents to be infinite even to have truth-telling as an equilibrium. In addition,~\citet{lambert2008truthful} provide a mechanism such that no equilibrium pays more than truth-telling, but here all equilibria pay the same amount; and while truth-telling is a Bayesian Nash equilibrium, unlike in peer prediction it generally is not a strict Bayesian Nash equilibrium. \emph{Minimal Truth Serum (MTS)}~\cite{riley2014minimum} is a mechanism where agents have the option to report or not report their predictions, and also lacks analysis of non-truthful equilibria. MTS uses a typical zero-sum technique such that all equilibria are paid equally.

Equilibrium multiplicity is clearly a pervasive problem in this literature.  While our present work only applies to the classical peer prediction mechanism, it provides an important step in addressing equilibrium multiplicity, and a new toolkit for reasoning about proper scoring rules.

\section{Preliminaries, Background, and Notation}\label{prelim}

\subsection{Game Setting}
Consider a setting with $n$ agents $A$.
\ifnum\fullversion=1
If $A' \subseteq A$, we let $-A'$ denote $A \setminus A'$.
\fi
Each agent $i$ has a private signal $\sigma_i \in \Sigma$. We consider a game in which each agent $i$ reports some signal $\hat{\sigma}_i \in \Sigma$.  Let $\boldsymbol{\sigma}$ denote the vector of signals and $\boldsymbol{\hat{\sigma}}$ denote the vector of reports.  Let  $\boldsymbol{\sigma_{-i}}$ and $\boldsymbol{\hat{\sigma}_{-i}}$ denote the signals and reports excluding that of agent $i$; we regularly use the $-i$ notation to exclude an agent $i$.

We would like to encourage truth-telling, namely that agent $i$ reports $\hat{\sigma_i} = \sigma_i$.
To this end, agent $i$ will receive some payment $\nu_i(\hat{\sigma}_i, \boldsymbol{\hat{\sigma}_{-i}})$ from our mechanism.  In this paper, the game will be {\em anonymous}, in that each player's payoffs will depend only on the player's own report and the {\em fraction} of other players giving each possible report $\in \Sigma$, and not on the identities of those players.
%This payment will only $$ = f(\hat{\sigma_i}, \frac{\sum_{j \neq i} \hat{\sigma_{-i})}{n-1}$$.
%Agent $i$ will be paid for ``predicting" the reported bits of the other players, and thus it is a coordination game.

\begin{assumption}[Binary Signals]
We will refer to the case when $\Sigma = \{0, 1\}$ as the {\bf binary signal} setting, and we focus on this setting in this paper.
\end{assumption}

\begin{assumption}[Symmetric Prior]
We assume throughout that the agents' signals $\boldsymbol{\sigma}$ are drawn from some joint {\bf symmetric prior} $Q$: a priori, each agent's signal is drawn from the same distribution. We in fact only leverage a weaker assumption, that $\forall \sigma, \sigma'$, and $\forall i \neq j$ and $k \neq l$, we have $\Pr[\sigma_j = \sigma'  | \sigma_i = \sigma] = \Pr[\sigma_l = \sigma' | \sigma_k = \sigma]$.
\end{assumption}
That is, the inference your signal lets you draw about others' signals does not depend on your identity or on the identity of the other agent.

Given the prior $Q$, for $\sigma \in \Sigma$, let $q(\sigma)$ be the fraction of agents that an agent expects will have $\sigma_j = \sigma$ \emph{a priori}.  Let
\[q(\sigma'|\sigma) := \Pr[\sigma_j = \sigma'| \sigma_i = \sigma]\]
(where $j \neq i$) be the fraction of other agents that a user $i$ expects have received signal $\sigma'$ given that he has signal $\sigma$.

\begin{assumption}[Signals Positively Correlated]\label{assn:signals-correlated}
We assume throughout that
the prior
$Q$ is {\bf positively correlated}, namely that $q(\sigma | \sigma) > q(\sigma)$, for all $\sigma \in \Sigma$.
\end{assumption}
That is, once a player sees that his signal is $\sigma$, this strictly increases his belief that others will have signal $\sigma$, when compared with his prior.  Notice that even after an agent receives his signal, he may still believe that he is in the minority.  Thus, simply encouraging agent agreement is not sufficient to incentivize truthful reporting.

\begin{assumption}[Signal Asymmetric Prior]
An additional assumption we will often use is that the prior is {\bf signal asymmetric}. For binary signals, as we consider in this paper, this simply means that $q(0) \neq q(1)$.
\end{assumption}
For a richer signal space, intuitively, a signal asymmetric prior is one that changes under a relabeling of the signals, so that lying can potentially be distinguishable from truth-telling.

We say that an agent plays \emph{response} $\sigma \rightarrow \hat{\sigma}$, if the agent reports signal $\hat{\sigma}$ when he receives signal $\sigma$.  Let $X$ be the set of all responses (e.g. $X=\{0\rightarrow 0,0\rightarrow 1,1\rightarrow 0,1\rightarrow 1\}$ when $\Sigma=\{0,1\}$).  In a \emph{pure-strategy} an agent chooses a response for each $\sigma \in \Sigma$, and thus there are $|\Sigma|^{|\Sigma|}$ possible pure strategies.
%; in a pure strategy, for each possible signal $\sigma \in \Sigma$ an agent might receive, the agent must select a signal $\sigma' \in \Sigma$ to report: $S = \left\{ \sigma\rightarrow\sigma' \right\}_{\sigma \in \Sigma}.$
Let $S$ be the set of pure strategies and let $s_i \in S$ denote a pure-strategy for agent $i$.
%\yk{I just find here we have defintion for response. }
We will also consider mixed strategies $\theta_i$, where agent $i$ randomizes over pure strategies.  Here we write
\[\theta_i(\sigma' | \sigma) := \Pr[\hat{\sigma_i} = \sigma'|\sigma_i = \sigma].\]  

\yuqing{I comment the definition of support here}
%We define the \emph{support} of an agent's strategy $Supp(\theta_i) \subseteq X$ to be the set of all responses that the agent uses with non-zero probability. 

A strategy profile $\boldsymbol{\theta} = (\theta_1, \ldots, \theta_n)$ consists of a strategy for each agent.

We can think of each $\theta$ as a linear transformation from a distribution over received signals to a distribution of reported signals.
Given a set of agents $A' \subset A$, we define
$$\theta_A'(\sigma'|\sigma):= E_{i\leftarrow A}[\theta_i(\sigma'|\sigma)]$$  where $ i\leftarrow A' $ means $ i $ is chosen uniformly at random from $A'$.
When discussing  {\em symmetric strategy profiles} where all players use the same strategy, we will often abuse notation and use notation for one agent's strategy to denote the entire strategy profile.

A {\em Bayesian Nash equilibrium} consists of a strategy profile $\boldsymbol{\theta} = (\theta_1, \ldots, \theta_n)$ such that no player wishes to change his strategy, given the strategies of the other players and the information contained the prior and his signal: for all $i$ and for all alternative strategies $\theta'_i$ for $i$,
$\E[\nu_i(\boldsymbol{\theta})] \geq \E[\nu_i(\theta'_i, \boldsymbol{\theta_{-i}})],$
where the expectations are over the realizations of the randomized strategies and the prior $Q$.
We call such an equilibrium  \emph{focal} if it provides a strictly larger payoff, in expectation, to each agent, than any other  equilibrium and \emph{weakly focal} if it provides a larger payoff (maybe not strictly).

Given a symmetric prior $Q$ and strategy profile $\boldsymbol{\theta} = (\theta_1, \ldots, \theta_n)$, we define
\[\hat{q}_j(\sigma'|\sigma) := \Pr[\hat{\sigma}_j = \sigma'| \sigma_i = \sigma] = \sum_{\sigma'' \in \Sigma} q(\sigma''|\sigma) \theta_j(\sigma'|\sigma'')\] for $i \neq j$. Intuitively, $\hat{q}_j(\sigma'|\sigma)$ is the probability of player $j$ reporting $\sigma'$, given that another player $i$ sees signal $\sigma$; note that this probability does not depend on the identity of $i$, by symmetry of the prior.  Given a set of agents $A' \subset A$, we define $$\hat{q}_A'(\sigma'|\sigma):= E_{j\leftarrow A'} \hat{q}_j(\sigma'|\sigma)$$  where $ j\leftarrow A' $ means $ j $ is chosen uniformly at random from $ A' $ (again assuming that the implicit reference agent $i \not\in A'$). %We also define $$ (\hat{q}_{-A}(1|0),\hat{q}_{-A}(1|1))=(\hat{q}_{[n]\setminus A}(1|0),\hat{q}_{[n]\setminus A}(1|1)) $$
If $\boldsymbol{\theta} = (\theta, \ldots, \theta)$ is symmetric, we simplify our notation to $\hat{q}(\sigma'|\sigma)$ because the referenced set of agents does not matter.
%\[\hat{q}(\sigma'|\sigma) := \Pr[\hat{\sigma}_j = \sigma'| \sigma_i = \sigma] = \sum_{\sigma'' \in \Sigma} q(\sigma''|\sigma) \theta(\sigma'|\sigma'')\] for $i \neq j$, so that if we write $q$, $\hat{q}$ and $\theta$ as matrices, we get: $\hat{q} = \theta q$.

\vspace*{5mm}

In the binary signal setting when $\boldsymbol{\theta}$ is symmetric,
%$\hat{q}_i$ does not depend on $i$\jz{added this part of the sentence for clarity} \katrina{It's explained above}
we have:
\begin{eqnarray}
   \label{eqn:p0}  \hat{q}(1|0) &= & \theta(1|0)q(0|0) + \theta(1|1)q(1|0) \\
   \label{eqn:p1}  \hat{q}(1|1) &= & \theta(1|0)q(0|1) + \theta(1|1)q(1|1)
\end{eqnarray}

Additionally, we observe that $q(1|b) = 1 - q(0|b)$, $\theta_i(1|b) = 1 - \theta_i(0|b)~ \forall i$, and $\hat{q}(1|b) = 1 - \hat{q}(0|b)$.  Note that we will typically use $b$ instead of $\sigma$ to refer to binary signals (bits).

There are four pure strategies for playing the game in the binary signal setting: always 0, always 1, truth-telling, lying:  $$S = \left\{ \zerl, \onel, \trul, \fall \right\} = \{ \zer, \one, \tru, \fal \}.$$  We will denote mixed strategies as $\strat{\theta(1|0)}{\theta(1|1)}$.

\yuqing{This extension of setting should be checked. Maybe move it to the end as a open question? Here is the open question section you wrote before (I put it here for your reference): Several open questions remain.  The first is to generalize this beyond binary signals.  One difficulty in handling this case is that the number of non-truthtelling equilibrium grows exponentially in the signal size.   Another question is whether these results can be extended to the Bayesian Truth Serum~\cite{prelec2004bayesian} or Robust Bayesian Truth Serum~\cite{witkowski2012robust} mechanisms (which have relaxed common prior assumptions).    }

\subsection{Proper Scoring Rules}
\gs{It would be nice if these did not have to be defined on all of the reals}
A scoring rule $PS:  \Sigma \times \Delta_{\Sigma} \rightarrow \mathbb{R}$ takes in signal $\sigma \in \Sigma$  and a distribution over signals $\delta_{\Sigma} \in \Delta_{\Sigma}$ and outputs a real number.  A scoring rule is \emph{proper} if, whenever the first input is drawn from a distribution $\delta_{\Sigma}$, then the expectation of $PS$ is maximized by  $\delta_{\Sigma}$.  A scoring rule is called \emph{strictly proper} if this maximum is unique. We will assume throughout that the scoring rules we use are strictly proper. By slightly abusing notation, we can extend a scoring rule to be $PS:  \Delta_{\Sigma} \times \Delta_{\Sigma} \rightarrow \mathbb{R}$  by simply taking $PS(\delta_{\Sigma}, \delta'_{\Sigma}) = \E_{\sigma \leftarrow \delta_{\Sigma}}(\sigma,  \delta'_{\Sigma})$.

In the case of scoring rules over binary signals, a distribution can be represented by a number in the unit interval, denoting the probability placed on the signal $1$.  In the binary signal setting, then, we extend proper scoring rules to be defined on $[0, 1] \times [0, 1]$.
%and in fact sometimes further require that they be well-defined on a larger, convex domain $\mathbb{R} \times \mathbb{R}$ that contains $[0,1] \times [0,1]$.
%and in fact further require that they be well-defined on $\mathbb{R} \times \mathbb{R}$.
%%\jz{Changed "and in fact further require that they be well-defined on $\mathbb{R} \times \mathbb{R}$"}
%One can see that in this representation the scoring rule must be affine in the first entry~\cite{mccarthy1956}.

%We will use the Brier Scoring Rule for predicting binary events.\jz{Not true anymore. We want to keep the part on the Brier scoring rule somewhere to show that we have at least one example of a scoring rule that works, and because it is an example of scoring rule that can achieve any slope a/b}

\subsection{Peer Prediction}
Peer Prediction~\cite{MRZ05} with $n$ agents receiving positively correlated binary signals $\boldsymbol{b}$, with symmetric prior $Q$,  consists of the following mechanism~$\mathcal{M}(\boldsymbol{\hat{b}})$:
\begin{enumerate}
\item Each agent $i$ reports a signal $\hat{b}_i$.
\item Each agent $i$ is uniformly randomly matched with an individual $j \neq i$, and is then paid
$PS(\hat{b}_j, q(1|\hat{b}_i)),$ where $PS$ is a proper scoring rule.
\end{enumerate}
That is, agent $i$ is paid according to a proper scoring rule, based on $i$'s prediction that $\hat{b}_j = 1$, where $i$'s prediction is computed as either $q(1|0)$ or  $q(1|1)$, depending on $i$'s report to the mechanism.  This can be thought of as having agent $i$ bet on what agent $j$'s reported signal will be.

\gs{make sure I didn't botch this which I changed to binary}

Notice that if agent $j$ is truth-telling, then the Bayesian agent $i$ would also be incentivized to truth-tell (strictly incentivized, if the proper scoring rule is strict).  Agent $i$'s expected payoff (according to his own posterior distribution) for reporting his true type $b_i$ has a premium compared to reporting $\neg b_i$ of:
$$PS\left(\hat{b}_j, q\left(1|b_i\right)\right) - PS\left(\hat{b}_j, q\left(1|\neg b_i\right)\right) \geq 0$$ (strictly, for strict proper scoring rules) because we know that the expectation of $PS(\hat{b}_j, \cdot)$ is (uniquely) maximized at $q(1|b_i)$.

\yk{I move pf definition from section 5 to here}

Here we introduce a convenient way to represent peer prediction mechanism:
\begin{definition}[Payoff Function Matrix]\label{PF}
Each agent $i$ who reports $\hat{b}_i$ and is paired with agent $j$ who reports $\hat{b}_j$, will be paid $h_{\hat{b}_j,\hat{b}_i}$. Then the peer prediction mechanism can be naturally represented as a $2\times 2$ matrix: $$\left( \begin{array}{cc}
h_{1,1} & h_{1,0} \\
h_{0,1} & h_{0,0}
\end{array} \right)=\left( \begin{array}{cc}
PS(1,q(1|1)) & PS(1,q(1|0)) \\
PS(0,q(1|1)) & PS(0,q(1|0))
\end{array} \right) $$
which we call the payoff function matrix.
\end{definition}

 \begin{example}[Example of Proper Scoring Rule]\label{brier}
 The Brier Scoring Rule for predicting a binary event is defined as follows. Let $I$ be the indicator random variable for the binary event to be predicted. Let $q$ be the predicted probability of the event occurring. Then:
 $$B(I, q) = 2I\cdot q + 2(1-I)\cdot (1-q) - q^2 - (1-q)^2.$$
 Note that if the event occurs with probability $p$, then the expected payoff of reporting a guess $q$ is (abusing notation slightly):
 $$B(p,q) =2 p\cdot q + 2(1-p)\cdot (1-q) - q^2 - (1-q)^2 = 1 - 2(p - 2p\cdot q + q^2)$$
 This is (uniquely) maximized when $p = q$, and so the Brier scoring rule is a strictly proper scoring rule. Note also that $B(p,q)$ is a linear function in $p$. Hence, if $p$ is drawn from a distribution, we have: $\E_p[B(p,q)] = B(\E[p],q)$, and so this is also maximized by reporting $q = \E[p]$.

 A slight generalization of the Brier Scoring Rule is the ``Shifted Brier Scoring rule", which also takes a parameter $c \in \mathbb{R}$.  We write $B_c(p, q) = B(p - c, q - c)$, so that both of the inputs are ``shifted" before the scoring rule is evaluated. The Shifted Brier Scoring rule is also a strictly proper scoring rule.
 %While we do not do so in this paper, the output of the Brier scoring rule can also be changed be scaled by any positive number, and shifted by any number without affecting the results.
 \end{example}

\ifnum\fullversion=1
\begin{example}\label{prelim-example}
Say that a restaurant with {\emph{quality parameter $p$}} satisfies each customer with probability $p$ and fails to satisfy them with probability $1-p$. Consider a restaurant with quality parameter $p$ uniformly distributed between $2/5$ and $4/5$.

In this case, we have that $q(1) = \frac{3}{5}$, $q(1|1) = \frac{28}{45}$, and $q(1|0) =  \frac{17}{30}$.
%Also note that if we are using the Brier scoring rule, that $q^*(1) = \frac{107}{180}$.

If we use Peer Prediction based on Brier scoring rule (see example~\ref{brier}), we get the following payoff function matrix
$$\left( \begin{array}{cc}
0.715 & 0.624 \\
0.226 & 0.358
\end{array} \right)=\left( \begin{array}{cc}
B(1,q(1|1)) & B(1,q(1|0)) \\
B(0,q(1|1)) & B(0,q(1|0))
\end{array} \right) $$
This means, for example, if an agent reports 1 and is paired with an agent that reports 0, he will receive payoff 0.226.

Thus if all the other agents play truthfully, the expected payoff of the agent who receives a 1 and plays truthfully will be
$B(q(1|1), q(1|1))=q(1|1)*B(1,q(1|1))+q(0|1)*B(0,q(1|1))=q(1|1)*0.715+q(0|1)*0.226$ since with probability $q(1|1)$, he believes the agent paired with him recevies and reports 1 which implies he will be paid 0.715, with probability $q(0|1)$, he believes the agent paired with him recevies and reports 0 which implies he will be paid 0.226. Similarly,
\begin{align*}
B(q(1|1), q(1|1)) = B\left(\frac{28}{45}, \frac{28}{45}\right) \approx .530 &\mbox{ if he receives a 1 and plays truthfully;}\\
B(q(1|1), q(1|0)) =  B\left(\frac{28}{45}, \frac{17}{30}\right) \approx .524   &\mbox{ if he receives a 1 and lies;}\\
B(q(1|0), q(1|0)) = B\left(\frac{17}{30}, \frac{17}{30}\right) \approx .509 &\mbox{ if he receives a 0 and plays truthfully; and}\\
B(q(1|0), q(1|1)) =  B\left(\frac{17}{30}, \frac{28}{45}\right) \approx .503 &\mbox{ if he receives a 0 and lies.}
\end{align*}
No matter which signal he receives, the agent will be better off truth-telling than lying, given that all other agents truth-tell. So truth-telling is an equilibrium in this example.

Notice that if all agents played 0 or all played 1, this would also be a Nash equilibrium. Also note, the payoff for each agent in either of these equilibria is 1.
\end{example}

While truth-telling is always an equilibrium of the peer prediction mechanism, as we will see, it is not the only equilibrium.
%Another equilibrium is to play each bit $b$ with probability $q^*(b)$.
Two more equilibria are to always play 0 or always play 1.  In Section~\ref{sec:equilibrium}, we further investigate equilibria of the peer prediction game. Based our the analysis of these multiple equilibria, we will develop a {\bf modified peer prediction mechanism,} wherein players are paid according to the peer prediction based on a carefully-designed proper scoring rule, modulo some punishment imposed on the all playing $0$ or all playing $1$ strategy profiles. This modified mechanism will make the truth-telling equilibrium focal.

%\paragraph{Modified Peer Prediction Mechanism} In addition to pay agents according to peer prediction mechnism, we also give agents punishment when everyone gives the same answer.
%
%We will show we can pick proper $PS$ and proper punishment to make truth-telling ``the best'' equilibrium in some sense.
%\else
%Additionally, see the full version for an example to illustrate these concepts.
%\fi

\subsection{Properties of Proper Scoring Rules}

%\katrina{Need to note all the $q^*$ stuff is for binary signals only.}
\begin{definition}\label{def:qstar}
For a prior $Q$,  proper scoring rule $PS$, and a binary signal space, we define $q^*(b)$ to be
the fraction of other agents reporting $b$ that would make an agent indifferent between reporting 0 or 1, i.e.,
%the probability of the agent one is matched to reporting $b$ that would make one indifferent between reporting 0 or 1:
\begin{align*}
%q^*(1) = \{p~|~PS(p, q(1 |0)) &= PS(p,q(1|1)), 0 \leq p \leq 1\}\\
%q^*(0) = \{p~|~PS(p, q(0 |0)) &= PS(p,q(0|1)), 0 \leq p \leq 1\}.
q^*(b) := \{p~|~PS(p, q(b |1)) &= PS(p,q(b|0)), 0 \leq p \leq 1\}.
\end{align*}
Since in much of what follows, we will only use $ q^*(1)$ and not $q^*(0)$, for convenience, we often denote $ q^*(1) $ by $q^*$.
\end{definition}
%For example, if more than a $q^*(1)$ fraction of the other agents report 1, then an agent expects to be paid no less for reporting 1 than for reporting 0;  if fewer than a $q^*(1)$ fraction of the other agents report 1, then an agent expects to be paid no less for reporting 0 than for reporting 1; and, if exactly a $q^*(1)$ fraction of the other agents report $1$, then agent $i$ is indifferent as to whether he reports 0 or 1. Use of a {\em{strict}} proper scoring rule would make these inequalities strict. Note that $q^*(0) + q^*(1) = 1$.
Note that $q^*(0) + q^*(1) = 1$.

We now study existence and uniqueness of $q^*(b)$.
\begin{claim}\label{claim:qstar}
For any symmetric prior $Q$ on binary signals, and any proper scoring rule $PS$, $q^*(b)$ always exists and lies between $q(b|0)$ and  $q(b|1)$; if $PS$ is strict and the signals are positively correlated, $q^*(b)$ is unique and lies strictly between  $q(b|0)$ and $q(b|1)$.
\end{claim}

\ifnum\fullversion=1
\begin{proof}
    For proper scoring rules over binary signals (where a probability distribution can be represented by a number in the unit interval), we know that  $PS(\cdot, \ell)$ is an affine function of its first argument \gs{we should cite something here}.  Therefore $\ell^b_a(\cdot) = PS(\cdot, q(b|a))$ are affine functions for $a,b \in \{0, 1\}$.  We also know, since PS is a scoring rule, that for $a,b \in \{0, 1\}$: $\ell^b_a(q(b|a)) \geq \ell^b_{\neg a}(q(b | a))$ (if PS is strict and the signals are positively correlated, this inequality is strict, since then $q(b|a) \neq q(b|\neg a)$, and $q(b| a)$ is a strict maximizer of $PS(q(b | a), \cdot)$). Thus, there is some point $q^*(b)$ between  $q(b | 0)$ and $q(b | 1)$  where the functions intersect; if PS is strict and signals are positively correlated, this point is unique.
\end{proof}
\else

See the full version for the proof.

\fi

%For example, if the agent $i$ is matched to a report 1 with probability more than $q^*(1)$, then $i$ expects to be paid no less for reporting 1 than for reporting 0;  if this probability is less than $q^*(1)$, then $i$ expects to be paid no less for reporting 0 than for reporting 1; and, if this probability is exactly $q^*(1)$, then agent $i$ is indifferent as to whether he reports 0 or 1. Use of a {\em{strict}} proper scoring rule would make the inequalities strict.
\ifnum\fullversion=1
The value $q^*(1)$ will have useful implications for best responses.
\begin{claim}\label{claim:equil-crit}
Given symmetric prior $Q$ and fixing a proper scoring rule, suppose a strategy profile $(\theta_1,...,\theta_{n})$ is an equilibrium of the corresponding peer prediction mechanism. Then
\begin{align*}
\mbox{if }\frac{1}{n-1} \sum\limits_{k \in [n] \setminus \{i\}} \hat{q}_k(1|b) < q^*(1), &\mbox{ then } \theta_i(1|b)=0, \mbox{ and} \\
\mbox{if }\frac{1}{n-1} \sum\limits_{k \in [n] \setminus \{i\}} \hat{q}_k(1|b) > q^*(1), &\mbox{ then }\theta_i(1|b)=1.
\end{align*}
For $\boldsymbol{\theta}$ symmetric, $\frac{1}{n-1} \sum\limits_{k \in [n] \setminus \{i\}} \hat{q}_k(1|b) =\hat{q}(1|b)$.
\end{claim}
%Note that $\frac{1}{n-1} \sum\limits_{k \in [n] \setminus \{i\}} \hat{q}_k(1|b)$ is simply the expected fraction of players other than $i$ that play $1$ when $i$ has private signal $b$. Therefore, an agent $i$ excepts to get a better payoff by reporting 1 (resp. 0) when the fraction of other players reporting $1$ is more (resp. less) than $q^*(1)$.
\begin{proof}
Since any proper scoring rule is affine in the first parameter, we can write $PS(p,q)=f(q) \cdot p + g(q)$. It follows that the expected payoff for player $i$ when reporting private bit $\hat{b}_i$ is $\frac{1}{n-1} \sum\limits_{k \neq i} PS(\hat{q}_k(1|\hat{b}_i),q(1|\hat{b}_i))=PS(\frac{1}{n-1} \sum\limits_{k \neq i} \hat{q}_k(1|\hat{b}_i),q(1|\hat{b}_i))$. The rest follows by the definition of $q^*(1)$.
\end{proof}
\fi

\section{Summary of Main Result and Proof}\label{focal-section} \label{sec:focal-for all}\label{sec:zero-one elimin}

\yk{I rewrite this section and save the original file in maintheorem1.tex}

\yuqing{If possible, I hope the whole section could be checked, especially the part about punishment (I marked them)}

In this section, we introduce our modified peer prediction mechanism and sketch the main theorem of this paper, that for almost any symmetric prior, there exists a modified peer prediction mechanism such that truth-telling is the focal equilibrium. In subsequent sections, we build up pieces of the proof of this theorem. Recall, we use the term {\emph{focal}} to refer to an equilibrium with expected payoff strictly higher than that of any other Bayesian Nash equilibrium.

\subsection{Our Modified Peer Prediction Mechanism MPPM}

%\begin{figure}\centering\includegraphics[scale=0.3]{attainableprior.jpg}\caption{The triangle is an area where the prior is positively correlated($ q(1|1)>q(1|0) $) and $ q(1|1)>q(0|0)$). The case of $ q(1|1)<q(0|0) $ is analogous. The prior is attainable in $ R_1 $ and $ R_2 $, while the prior is unattainable in $ R_3 $. In the right boundary, $q(1|0)=q(1|1)$, so the private signal does not have any information. We call this boundary the set of non-informative priors. In the left boundary, $q(0|0)=q(1|1)$. We call this boundary the set of symmetric priors.  }\label{fig:attainableprior}\end{figure}

Recall that modified peer prediction mechanism is the mechanism wherein players are paid according the peer prediction based on a carefully-designed proper scoring rule, modulo some punishment imposed on the all playing $0$ or all playing $1$ strategy profiles. So our approach differentiates between two types of equilibria:
\begin{definition}[Informative strategy]\label{def:informative}
We call \textit{always reporting $\one$} and \textit{always reporting $\zer$} uninformative strategies; we call all other strategies (equilibria) {\em informative}.
\end{definition}

%Note that always reporting $\one$ and always reporting $\zer$ are uninformative strategies.

\paragraph{Designing the Optimal Peer Prediction Mechanism}
We start to describe our modified peer prediction mechanism MPPM. We use two steps to design our MPPM.  First we define the PPM:

\begin{definition}
Given any  binary, symmetric, positively correlated, and signal asymmetric prior $Q$, with $ q(1|1)>q(0|0) $ (the $ q(0|0)<q(1|1) $ case is analogous), we first design our peer prediction mechanism \textbf{$PPM(Q)$} and represent is as a payoff function matrix (See Definition~\ref{PF}). $PPM(Q)$ depends on the region $Q$ belongs to, we defer the definitions of regions $R_1,R_2,R_3$ to Definition~\ref{attainable prior}.
\begin{description}
\item[1]If $ Q\in R_1 $, then $PPM(Q)=\mathcal{M}_1(Q)$\\

\item[2]If $ Q\in R_2 $, then $PPM(Q)=\mathcal{M}_2(Q)$\\

\item[3]If $ Q\in R_3 $, then we pick a small number $\epsilon>0$ and $PPM(Q, \epsilon)=\mathcal{M}_3(Q,\epsilon)$\\

\end{description}
where

$\mathcal{M}_1(Q) = \left( \begin{array}{ccc}
\zeta(Q) & & 0\\
 & & \\
0 & & 1\\
\end{array} \right)$, $\mathcal{M}_2(Q) =\left( \begin{array}{ccc}
1 & & 0\\
 & & \\
0 & & \eta(Q)\\
\end{array} \right) $, $\mathcal{M}_3(Q,\epsilon) =\left( \begin{array}{ccc}
\zeta(Q,\epsilon) & & \delta(Q,\epsilon)\\
 & & \\
0 & & 1\\
\end{array} \right) $ and\\

$0\leq\zeta(Q)$, $\eta(Q)\leq 1$ are constants that only depend on common prior $Q$. $0\leq\zeta(Q,\epsilon)$, $\delta(Q,\epsilon)\leq 1$ are constants that only depend on common prior $Q$ and $\epsilon>0$\footnote{Actually $\zeta(Q)=\sqrt{\frac{q(0|0) q(0|1)}{q(1|0) q(1|1)}}   $, $  \eta(Q)=\frac{1}{q(1|1)}(\sqrt{\frac{(q(1|1)-q(1|0))(q(1|1)-q(0|0))}{q(0|0)q(1|0)}}-q(0|1))   $,\\
$\zeta(Q,\epsilon)=\frac{q(0|0) q(0|1)}{q(0|0) q(0|1) q(0|0)+\epsilon+q(1|0) (q(1|1)-q(1|1)
   q(0|0)+\epsilon)}$, and\\
$\delta(Q,\epsilon)=\frac{q(1|0) q(1|1) (q(0|0)+\epsilon-1)^2-q(0|0) q(0|1) q(0|0)+\epsilon^2}{q(0|0)+\epsilon (q(1|0)
      q(1|1) (q(0|0)+\epsilon-1)-q(0|0) q(0|1) q(0|0)+\epsilon)}$. }.
\end{definition}

\vspace{7pt}
Note that actually $PPM(Q)$ is a quite simple mechanism. We use region $R_1$ as example: if the prior belongs to region $R_1$, for every $i$, agent $i$ will receive 0 payment if the agent paired with agent $i$, call him agent $j$, reports different signal than him. If both agent $i$ and agent $j$ report $1$, agent $i$ will receive a payment of $0\leq\zeta(Q)\leq 1$, if both agent $i$ and agent $j$ report $0$, agent $i$ will receive payment of $1$.

Actually for regions $R_1,R_2$, the $PPM(Q)$ we define here is the optimal peer prediction mechanism in that it maximizes the advantage of truth-telling over the informative equilibria which have the second largest expected payoff over all Peer-prediction mechanisms with payoffs in $[0, 1]$. For region $R_3$, the optimal peer prediction mechanism does not exist, but the advantage of the $PPM(Q, \epsilon)$ we define approaches the optimal advantage as $\epsilon$ goes to 0.

\begin{definition}  \label{def:Delta*Q} We define $\Delta^*(Q)$ to be the supremum of the advantage of truth-telling over the informative equilibria which have the second largest expected payoff over all Peer-prediction mechanisms with payoffs in $[0, 1]$.
\end{definition}

\yuqing{This paragraph for punishment is new}

\paragraph{Add Punishment} In our $PPM(Q)$, an uninformative strategy can still obtain the highest payoff. For example, in mechanism $\mathcal{M}_1$, agents will receive maximal payment 1 by simply always reporting 0.

Our final $MPPM(Q)$ Mechanism is the same as the $PPM(Q)$ except that we add a punishment designed to hurt the all 0 or all 1 equilibria.

\begin{definition}
Our Modified Peer-Prediction Mechanism \textbf{$MPPM(Q)$} (or \textbf{$MPPM(Q, \epsilon)$} if $Q \in R_3$) has payoffs identical to  $PPM(Q)$ (or $MPM(Q, \epsilon)$) except that, in the event all the \emph{other} agents play all $\zer$ or all $\one$,
 it  will issue an agent a punishment of  $p = \frac{1-t}{2(1-\epsilon_Q)} + \frac{\Delta^*(Q)}{2 \epsilon_Q}$  where  $\epsilon_Q$ is the maximum probability that a fixed set of $n-1$ agents receive the same signal (either all $\zer$ or all $\one$); $t$ is the expected of payoff of truth-telling $\tru$ in the $PPM(Q)$, and $\Delta^*(Q)$ is as defined in Definition~\ref{def:Delta*Q}.
\end{definition}

To make truth-telling focal, we need to impose a punishment to the agents if everyone else reports the same signal.  However, such a punishment may distort the equilibria of the mechanism.  To avoid this, we punish an agent by $p$ when all the \emph{other} agents report the same signal.  Because an agent's strategy does not influence his punishment, his marginal benefit for deviation remains the same and so the equilibrium remain the same.  However, while all $\zer$ and all $\one$ remain equilibrium, in them, $MPPM(Q)$ will punish each agent  by $p$.

A difficulty arises: if the number of agents  is too small like 2 or 3, it is possible (and even probable) that all agents report their true signals, yet are still punished by the $MPPM(Q)$ mechanism.  Punishments like this might distort the payoffs among the informative equilibrium.  However, if $\epsilon_Q$ (the probability that $n-1$ agents receives the same signal) is sufficient small, this is no longer a problem.  For most reasonable priors, as the number of agents increases, $\epsilon_Q$ will go to zero.  Formally we will need that the number of agents is large enough such that $\epsilon_Q < \frac{\Delta^*(Q)}{1 - t +\Delta^*(Q)}$.
%
%we will show that when the number of agents is sufficient large,  %(like 20 or 30)
%to make the probability ($\epsilon_Q$) that everyone receives the same signal sufficient small, we can make truth-telling focal by imposing a proper punishment $p$ to the agent when everyone else reports the same signal.
%
%However, if the number of agents is too small like 2 or 3, it is totally possible that everyone tells the truth and reports the same signal. But we will show that when the number of agents is sufficient large (like 20 or 30) to make the probability ($\epsilon_Q$) that everyone receives the same signal sufficient small, we can make truth-telling focal by imposing a proper punishment $p$ to the agent when everyone else reports the same signal.
%
%Formally, let $\epsilon_Q$ be the maximum probability that a fixed set of $n-1$ agents receive the same signal (either all $\zer$ or all $\one$). Let $t$ be the expected of payoff of truth-telling $\tru$ in the PPM we defined above.

%If the number of agents is large enough such that $\epsilon_Q < \frac{\Delta^*(Q)}{1 - t +\Delta^*(Q)}$, We will design a mechanism MPPM identical to the PPM we defined above except that we will issue a punishment of $ p = \frac{1-t}{2(1-\epsilon_Q)} + \frac{\Delta^*(Q)}{2 \epsilon_Q}$ to an agent if all the other agents play all $\zer$ or all $\one$.

If the number of agents is too small such that $\epsilon_Q \leq \frac{\Delta^*(Q)}{1 - t +\Delta^*(Q)}$, we cannot show that $MPPM(Q)$ has truth-telling as a focal equilibrium.

In particular, we can see if $\epsilon_{Q} \rightarrow 0$ (say as the number of agents increases), then at some point, truth-telling will be focal.  We know that such a limit is necessary because, for example, with two agents making truth-telling focal is impossible.

Note that if the prior tells us the probability of a 1 event is concentrated far away from 0 and 1, the number of agents we need to make truth-telling focal will be very small since uninformative equilibria (all 1 and all 0) are far away from truth-telling. To give a feeling of the actual number needed to make truth-telling focal, we calculated it for a prior that is not very ``good'': the probability p of a 1 event is uniformly drawn from [0.5,0.9] (This prior means we only know 1 event is more likely to be happen and with at least 0.1, 0 event will happen). Even for this prior, we only need at most 30 agents.

%In this section, we introduce our modified peer prediction mechanism and sketch the main theorem of this paper, that for almost any symmetric prior, there exists a strictly proper scoring rule such that truth-telling is the focal equilibrium of the peer prediction mechanism using that scoring rule. In subsequent sections, we build up pieces of the proof of this theorem. Recall, we use the term {\emph{focal}} to refer to an equilibrium with expected payoff strictly higher than that of any other Bayesian Nash equilibrium.

%We undertake this in two steps, distinguishing between two types of non-truthtelling equilibria of the peer prediction mechanism.  First we show how to choose the proper-scoring rule so that truth-telling is focal among the informative equilibrium.  Then, given those results, we show how to modify the payments so that the uninformative equilibria pay less, in expectation, than truth-telling.\\

%For convenience, we denote $q^*(1) $ by $ q^* $.
%\gs{better place?}
%{
%\renewcommand{\thetheorem}{}

%\yuqing{I feel like the statement should be checked}

\begin{theorem}(Main Theorem (Informal))\label{thm:focal_main}
Let $Q$ be a binary, symmetric, positively correlated and signal asymmetric prior, and let $\epsilon_Q$ be the maximum probability that a fixed set of $n-1$ agents receive the same signal (either all $\zer$ or all $\one$).  Then
\begin{enumerate}
\item In our PPM, truth-telling has the largest expected payoff among all informative equilibria. Moreover, over the space of Peer-Prediction mechanisms, our $PPM(Q)$ maximizes the advantage truth-telling has over the informative equilibrium which have the second largest expected payoff, over all Peer-prediction mechanisms with payoffs in $[0, 1]$ for regions $R_1,R_2$ and  $PPM(Q, \epsilon)$ approaches the maximal advantage for region $R_3$ as $\epsilon$ goes to 0.
\item There exists a constant $\xi_{q(1|1),q(1|0)}$ which only depends on $q(1|1)$ and $q(1|0)$ such that, if $\epsilon_Q < \xi$, our $MPPM(Q)$ makes truth-telling focal.
\end{enumerate}
\end{theorem}
%\addtocounter{theorem}{-1}
%}

The main theorem is proved in four steps. First, we classify all of the equilibria of the peer prediction mechanism, using the best response plot as a technical tool (see Section~\ref{sec:equilibrium}). Second, using the best response payoff plot, we are able to compare the payoffs between informative equilibria (see Section~\ref{sec:focal-for-some}).  Third, we show that by carefully selecting the scoring rule on which peer-prediction is based, truthtelling can be made focal among \emph{informative} equilibria, and are even able to optimize the advantage of truth-telling over the other informative equilibria (see Section~\ref{sec:optimization}).  Finally,
we suitably punish the uninformative equilibria, so that their payoff is lower than that of truthtelling (see Section~\ref{sec:punish}).

\section{Equilibrium Characterization} \label{sec:equilibrium}
In this section, we discuss the multiple equilibria of binary peer prediction games instantiated with strict proper scoring rules.
 %corresponding to a $ (\alpha,\beta,q^*,\gamma) $ line-set. 
 We show that there are between 7 and 9 symmetric Bayesian Nash Equilibria in these mechanisms.

%Recall that $q^*$ is such that $\ell(q^*, 0) = \ell(q^*,1).$   By Claim~\ref{claim:qstar} we know that $q(1|0) < q^* < q(1|1)$.

\begin{theorem}\label{eqbinary}
Let $Q$ be a symmetric and positively correlated prior on $\{0, 1\}^n$, and let $\mathcal{M}$ be a peer-prediction mechanism run with a strictly proper scoring rule with break-even $q^*$ (Definition~\ref{def:qstar}. Then there are no asymmetric equilibria. All equilibria are symmetric and only depend on $q^*$; they are
 $$\zer, \one, \tru, \strat{q^*}{q^*}, \strat{0}{\frac{q^*}{q(1|1)}},\strat{\frac{q^*-q(1|0)}{q(0|0)}}{1}$$
 and also conditionally include
\begin{eqnarray}
\fal & & \mbox{if } q(0|1) \leq q^* \leq  q(0|0) \label{eq:always-lie} \\
\strat{1}{\frac{q^*-q(0|1)}{q(1|1)}} & & \mbox{if } q(0|1) \leq q^*   \label{eq:zero-lie-one-amb} \\
\strat{\frac{q^*}{q(0|0)}}{0} & & \mbox{if }  q^* \leq q(0|0)  \label{eq:one-lie-zero-amb}
\end{eqnarray}
\ifnum\fullversion=1  We denote the set of all equilibria by $ \alleq_Q(q^*) $. \fi
\end{theorem}

Because either the conditions of Equation~\ref{eq:zero-lie-one-amb} or of Equation \ref{eq:one-lie-zero-amb} are satisfied, there are always between 7 and 9 equilibria.   Additionally, we note that if the conditions of Equation~\ref{eq:zero-lie-one-amb} or \ref{eq:one-lie-zero-amb} are equalities, then the corresponding equilibrium is equivalent to the $\fal$ equilibrium.

\begin{figure}[htbp]
\begin{center}
\includegraphics[scale=.5]{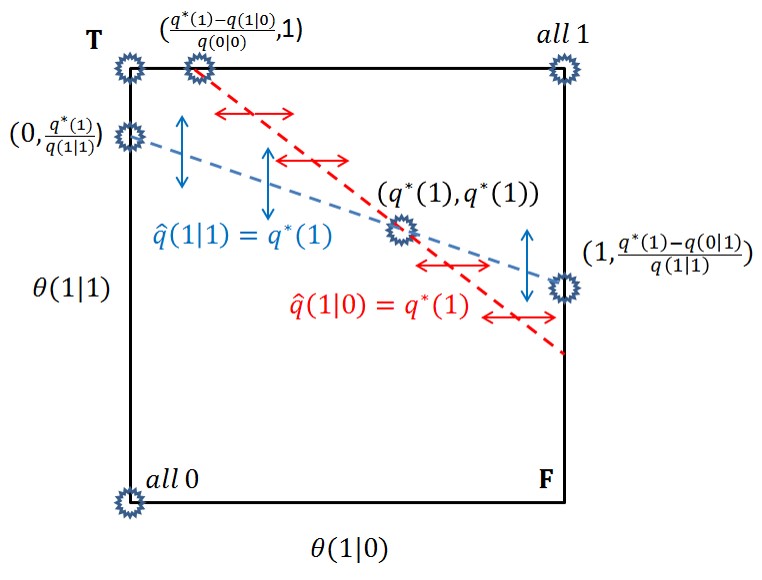}
\caption{Illustration of the multiple equilibria of peer prediction \ifnum\fullversion=1 for Example~\ref{prelim-example}\fi. \label{fig:plot-of-equilibrium}}
\vspace{-2\baselineskip}
\end{center}
\vspace{\baselineskip}
\end{figure}

Figure~\ref{fig:plot-of-equilibrium} shows the 7 equilibria \ifnum\fullversion=1 that exist in Example~\ref{prelim-example} using the Brier scoring rule (see~\ref{brier}), which are $\zer, \one, \tru,$
\begin{align*}
    \strat{q^*}{q^*} &\approx \strat{.594}{.594}, \\
  \strat{0}{\frac{q^*}{q(1|1)}}&\approx \strat{0}{.955},\\
  \strat{\frac{q^*-q(1|0)}{q(0|0)}}{1}&\approx \strat{.064}{1},\mbox{ and}\\
   \strat{1}{\frac{q^* - q(0|1)}{q(1|1)}}&\approx \strat{1}{.348}.
\end{align*}
\else of peer prediction under a specific scoring rule (see the full version in the full version). \fi
 Note that to the right of the red line where $\hat{q}(1|0) = q^*$, the best response is to increase $\theta(1|0)$; to the left of the red line, the best response is to decrease $\theta(1|0)$; and on the line an agent is indifferent. Similarly, above the blue line where $\hat{q}(1|1) = q^*$, the best response is to increase $\theta(1|1)$; below the blue line, the best response is to decrease $\theta(1|1)$; and on the line an agent is indifferent.

If  \ifnum\fullversion=1 Example~\ref{prelim-example} \else the example \fi were such that $q^*$ were closer to $1/2$, then the red line would intersect the bottom half of the square and there would be two additional equilibria, one at $\fal$ and the other at the intersection.  Exactly 8 equilibria would exist if the intersection of the red line with the edges of the square were exactly at $\fal$.

\subsection{Best Response Plots}\label{sec:bestresponseplot}
We introduce a new tool for the analysis of proper scoring rules, the best response plot (see Figure~\ref{fig:reportplot1} for an example); these plots are
the main tool we use in the proof of Theorem~\ref{eqbinary}.
Given any scoring rule, one can draw its corresponding best response plot.
Each point $(x, y)$ on a best response plot corresponds to a strategy of the game: the $x$-axis indicates the probability of reporting $1$ when one observes $0$; the $y$-axis the probability of reporting $1$ when one observes $1$.
At any point on the plot, we suppose that all players other than a single fixed player $i$ play the same fixed strategy (the strategy corresponding to that point of the plot), and we examine the best response of $i$, given that fixed strategy. Since  we will show that no asymmetric equilibria exist, intuitively, the question of whether a particular point on the best response plot corresponds to an equilibrium of the mechanism is simply the question of whether responding with that same strategy is a best response for $i$.

%On a best response plot, given fixed strategies for the players other than $i$, the $x$ axis indicates the number of people who will report 1 when agent $i$ observes 0; the $y$ axis indicates the number of people who will report 1 when agent $i$ observes 1. This will help player $i$ decide which signal he should report.  Intuitively, when the number of people who will report 1 is high, agent $ i $'s report should also be 1, no matter what he observes;  when the number of people who will report 1 is low, his report should also be 0.

Intuitively, when the number of other players who will report 1 is high, agent $i$'s report should also be 1, no matter what he observes;  when the number of other players who will report 1 is low, his report should also be 0.
There exists a ``\textit{break-even}'' point where agent $i$ is indifferent between reporting 1 and reporting 0. We will use this break-even point to divide the plot into four regions:
the horizontal axis is the break-even point indicating $i$'s indifference between reporting 1 versus 0 when he observes a signal of 0; the vertical axis is the break-even point when he observes 1. Above the horizontal axis, agent $ i $ prefers to report 1 when he observes a 1. Left of the vertical axis, agent $ i $ prefers to report 0 when he observes a 0. Thus, in the upper-left region, agent $i$ prefers to tell the truth. In other words, the best response of agent $i$ is truth-telling $\tru$ when ${\hat{q}}_{-i}$ is in the upper-left region, which is why it is labelled $ R_{\tru} $ (we label other regions similarly).

We want to select the proper scoring rule in order that truth-telling will have higher expected payoff than any other equilibrium. Different scoring rules may have different equilibria, which seems to pose a challenge to their analysis. However, {\bf{by studying the best response plot, we will show that the set of equilibria only depends on the value of the break-even point,}} and not on any other aspects of the proper scoring rule. This property will allow us to easily classify proper scoring rules by their equilibria.

\begin{figure}
\centering
\includegraphics[scale=0.15]{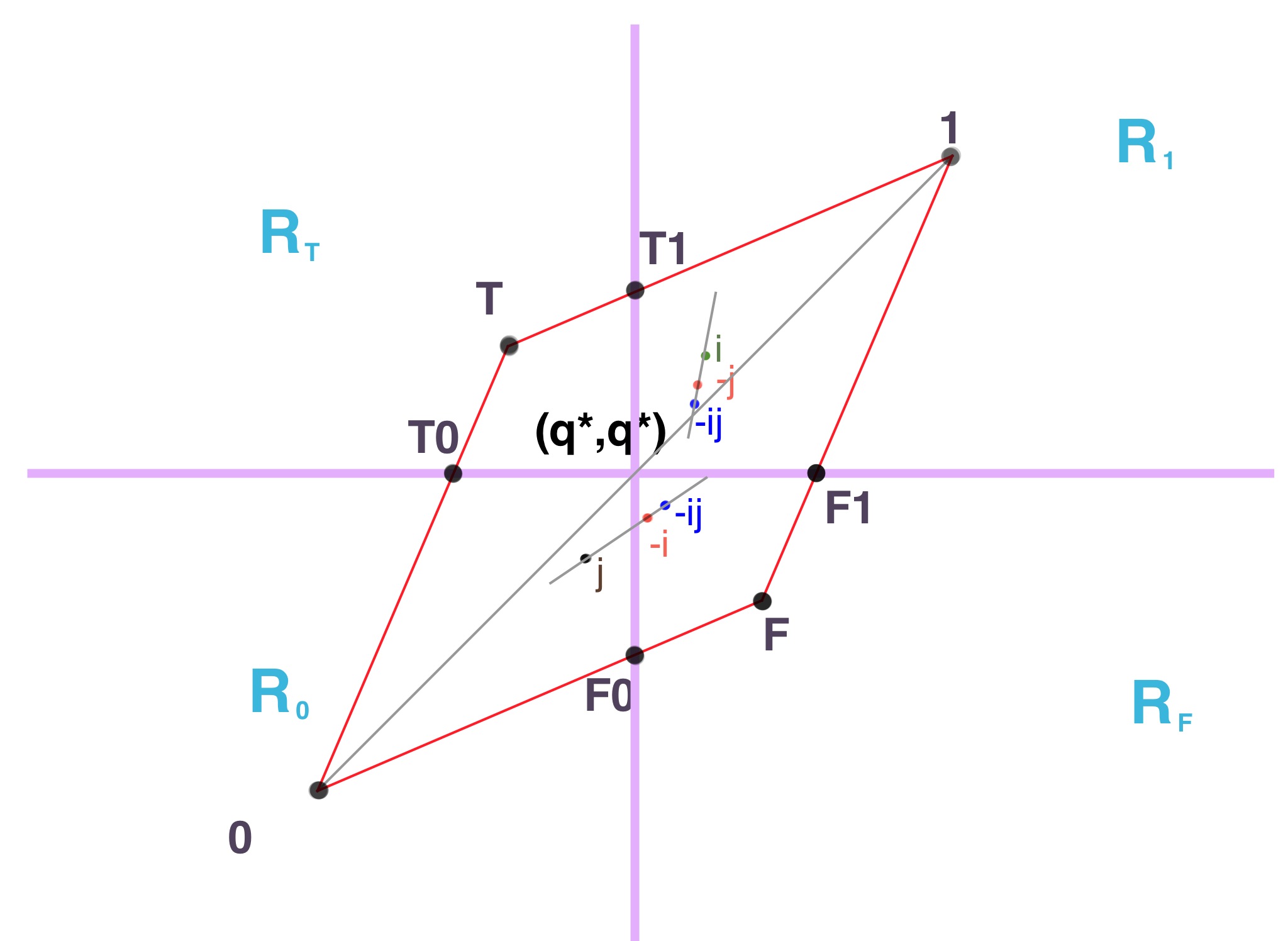}
\caption{A best response plot. Given fixed strategies for the players other than $i$, the $x$ axis indicates the number of people who will report 1 when agent $i$ observes 0; the $y$ axis indicates the number of people who will report 1 when agent $i$ observes 1. The horizontal axis is the break-even point indicating $i$'s indifference between reporting 1 versus 0 when he observes a signal of 0; the vertical axis is the break-even point when he observes 1. Above the horizontal axis, agent $ i $ prefers to report 1 when he observes a 1. Left of the vertical axis, agent $ i $ prefers to report 0 when he observes a 0. Thus, in the upper-left region, agent $i$ prefers to tell the truth. In other words, the best response of agent $i$ is truth-telling $\tru$ when ${\hat{q}}_{-i}$ is in the upper-left region, which is why it is labelled $ R_{\tru} $ (we label other regions similarly). Since $ \fal $ is in the $ R_{\fal} $ region,  $ \fal $ is an equilibrium.}
\label{fig:reportplot1}
\end{figure}

Now we will introduce the best response plot in detail. The best response plot is a plot which, given $\theta_{A'}$ (the ``average" strategy of the agents in $A'$), maps it to the point $ \hat{q}_{A'}(1|0),\hat{q}_{A'}(1|1)$. For example, $\trul$ would map to $(q(1|0),q(1|1))$, and $\onel$ would map to $ (1,1) $. Because we assume $ q(1|1)> q(1|0) $, this map is a bijection. Since it is a bijection, we will abuse notation a little by using $\hat{q}_{A'}(1|0),\hat{q}_{A'}(1|1)$ to represent $\theta_{A'}$.

The four pure strategies $ \tru $, $ \fal $, $ \textbf{0} $, $ \textbf{1} $ are mapped to four distinct points in the best response plot. We abuse notation and call these four points in the best response plot $\tru $, $ \fal $, $ \textbf{0} $, $ \textbf{1} $.

We define $R$ to be the convex hull defined by the points in $S$ (recall $S$ is the set of pure strategies). 

Recall that $X=\{0\rightarrow 0,0\rightarrow 1,1\rightarrow 0,1\rightarrow 1\}$ is the set of all possible best responses. We define the \emph{support} of an agent's strategy $Supp(\theta_i) \subseteq X$ to be the set of all responses that the agent uses with non-zero probability. For example, for the strategy that is always reporting 1 with probability $\frac{1}{2}$ and ignoring the private signal received, its support is $X$. For pure strategy like $\tru$, its support is itself $\{0\rightarrow 0,1\rightarrow 1\}$. 

%
%\begin{definition}\label{response}
%We define $ b\rightarrow \hat{b} $ as a response: receive $ b $ and report $ \hat{b} $ where $ b,\hat{b}\in \{0,1\} $
%\end{definition}
%
%\begin{definition}\label{bestresponse}
%We define $ b\rightarrow \hat{b} $ as a response: receive $ b $ and report $ \hat{b} $ where $ b,\hat{b}\in \{0,1\} $. Let $ P_0 $ and $ P_1 $ denote $ \{ 0\rightarrow 0,0\rightarrow 1\} $ and $ \{1\rightarrow 0,1\rightarrow 1\} $ respectively. Let $ P=P_0\cup P_1 $. We will abuse the notation a little bit to say $ \theta(b\rightarrow \hat{b})=\theta(\hat{b}|b) $
%\end{definition}
%
%\gs{again the below definition should be redone with responses???, also not sure if this should be defined for all of $\mathbb{R}^2$, or just $R$.}
%
For $W \subseteq X$, let $$R_W =\{p \in R~|~\forall r \in W, r \mbox{ is a best response when } {\hat{q}}_{-i} = p\}.$$ For example, when $W=\{0\rightarrow 0,1\rightarrow 1\}=\tru$, $R_{W}$ is the upper left region. When $W=\{\tru,\zer\}$, $R_{W}$ is the left part of $x$ axis (which actually is the intersection of $R_{\tru}$ and $R_{\zer}$).

Actually, we abuse notation in two further ways.  First, given $s \in S$ we define $$R_s = \cap_{r \in Supp(s)}R_r$$ to divide $R$ into four quadrants: $R_\fal, R_\tru, R_\zer, R_\one$.  We also extend these quadrants to partition the entire plane.

Given a finite set of points $P$, we define $ Conv(P) $ as the convex hull of $ P $. Given a convex set $U \subseteq \mathbb{R}^k$, let $ Int(U)$ denote the interior points of $U$ and let $\partial(U)$ denote the boundary of $U$ (so that $U$ is the disjoint union of $ Int(U)$ and $\partial(U)$).
%For any subset $U$ of $ (\hat{q}(1|0),\hat{q}(1|1)) $ plane, we define $ \mathcal{A}(U)=\{i|{\hat{q}}_{i}\in U\} $

%\smallskip

We name some additional points of our best response plot. We define $ \tru \textbf{1}=\partial(R)\cap R_{\{\tru,\textbf{1}\}} $, $ \tru \textbf{0}=\partial(R)\cap R_{\{\tru,\textbf{0}\}} $, $ \fal \textbf{1}=\partial(R)\cap R_{\{\fal,\textbf{1}\}} $, $ \fal \textbf{0}=\partial(R)\cap R_{\{\fal,\textbf{0}\}}$, if they exist. (See Figure~\ref{fig:reportplot1}.)

%\smallskip

Notice that in the best response plot, 
$ \tru\in R_{\tru}$,
%/\cup_{s\neq \tru}R_s  $; 
$ \textbf{1}\in R_{\textbf{1}}$, and 
%/\cup_{s\neq \textbf{1}}R_s  $; 
$ \textbf{0}\in R_{\textbf{0}}$,
%/\cup_{s\neq \textbf{0}}R_s  $. 
but $ \fal $ may not be in $ R_{\fal} $.  (See Figure~\ref{fig:reportplot2}.)
When $ q(1|1)\geq q(1|0) $, if $ q(0|0)<q^* $, then $ \fal\notin R_{\fal} $; if $ q(0|0)\geq q^* $, then $ \fal\in R_{\fal} $. When $ q(1|1)<q(0|0) $, if $ q(0|1)>q^* $, then $ \fal\notin R_{\fal} $;  if $ q(0|1)\leq q^* $, then $ \fal\in R_{\fal} $. It turns out that these two cases (i.e., whether $\fal$ lies in $R_{\fal}$ or not) create large differences in the structure of the equilibria.

\yk{I cut here}
\ifnum\fullversion=1

\begin{claim}\label{eqplot}
$(\theta_1,...,\theta_n)$ is an equilibrium iff
$ \forall i $, $$ {\hat{q} }_{-i}\in R_{Supp(\theta_i)} $$
\end{claim}
\begin{proof}
$(\theta_1,...,\theta_n)$ is an equilibrium iff $ \forall i $, any strategy in $ Supp(\theta_i) $ is a best response for player $ i $; this means that $ {\hat{q}_{-i}} \in \bigcap \limits_{s\in Supp(\theta_i)} R_s=  R_{Supp(\theta_i)}$ $\forall i$.
\end{proof}

Note that we mentioned before that there is a bijection between $\theta_i$ and ${\hat{q}}_i$, for convenience, we sometimes write this claim as $${\hat{q} }_{-i}\in R_{Supp({\hat{q}}_i)}$$

\subsection{Proof of Theorem~\ref{eqbinary}}

%\paragraph*{Proof Overview for Theorem~\ref{eqbinary}}
We will use the best response plot as a tool to compute all the equilibria of the peer prediction mechanism. The proof proceeds in three steps. First, we use the best response plot to find some symmetric equilibria. Second, we show that no asymmetric equilibria nor additional symmetric equilibria exist.  At this point, we have classified all equilibria in terms of ${\hat{q}}$.  The third step inverts the mapping from $\theta$ to ${\hat{q}}$ to solve for the equilibria in terms of $\theta$.

%Each strategy profile $(\theta_1,\theta_2,...,\theta_n) $ has corresponding points $ ({\hat{q}}_1,{\hat{q}}_2,...,{\hat{q}}_n)$ in this plot; note that since we assume $q(1|0) \neq q(1|1)$, the mapping from strategies to points in the best response plot is injective (see equations (1) and (2)). If $ (\theta_1,\theta_2,...,\theta_n) $ is an equilibrium, then for every agent $i$, $ \theta_i $ is any linear combination of player $i$'s best reponses. $ \theta_i $ only depends on the expectation of other agents' report, $ \hat{q}_{-i}$; note that we pick the player $j$ we match to $i$ uniformely at random, and therefore $ \hat{q}_{-i}=\frac{1}{n-1} \sum\limits_{i=1,i\neq j}^n {\hat{q}}_{i}$. So the positions of the ${\hat{q}}_i$'s immediately tell us whether $ (\theta_1,\theta_2,...,\theta_n) $ is an equilibrium\jz{We do not have a statement on the best response as a function of $q^*$ in the asymmetric case. Added one}. That is why we use best response plot rather than figure~\ref{fig:plot-of-equilibrium}. If for any equilibrium, we have $ {\hat{q}}_1={\hat{q}}_2=...={\hat{q}}_n $, then there are no asymmetric equilibrium.\\
%For convenience purposes, let $ q^* $ denote $ q^* $; we will never use $ q^*(0) $ from now on.\\
%We begin with the following claim.  \gs{This should be a claim, not a lemma}

\paragraph{Symmetric equilibria}

\begin{figure}
\centering
\includegraphics[scale=0.09]{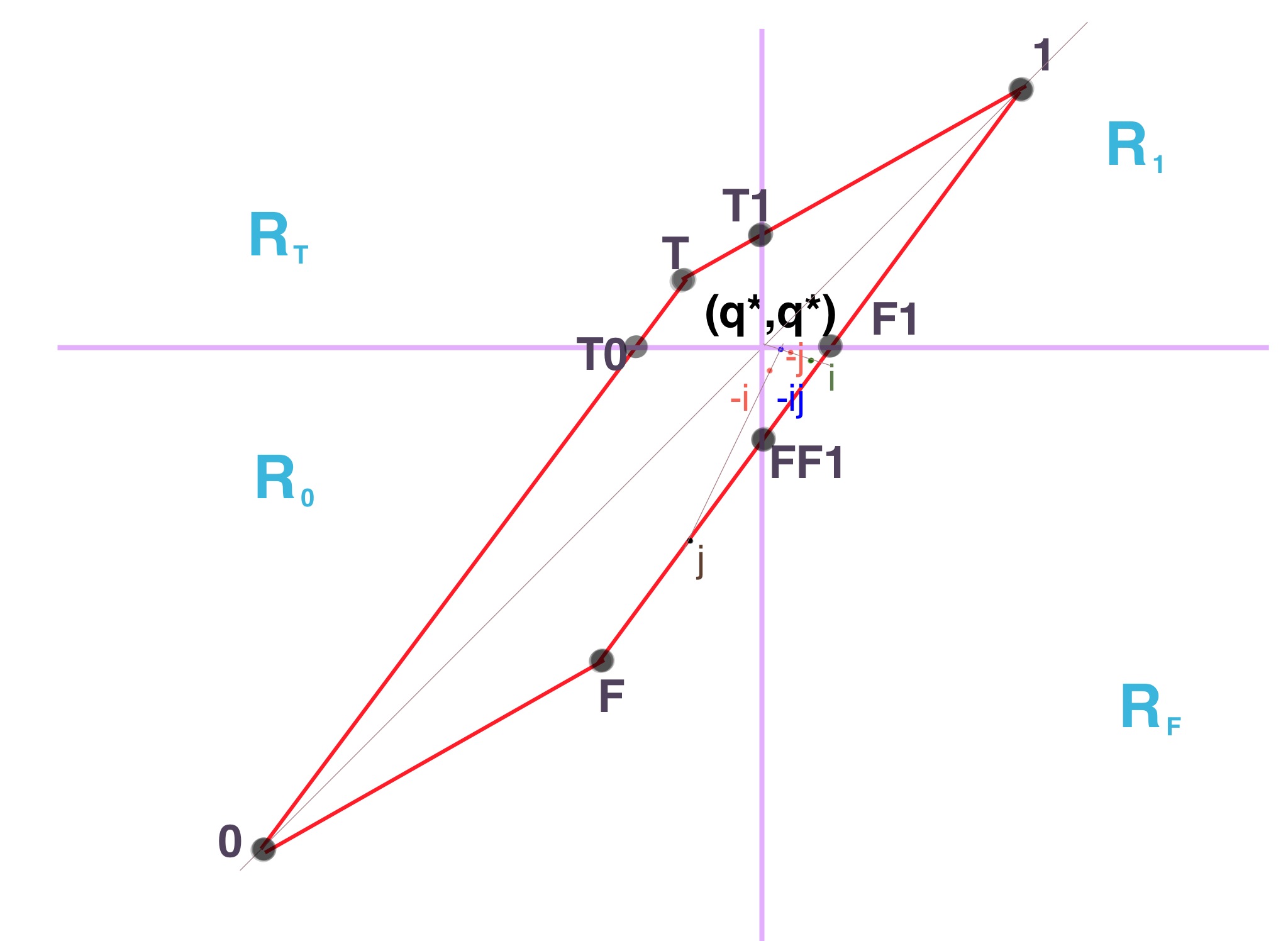}
\includegraphics[scale=0.09]{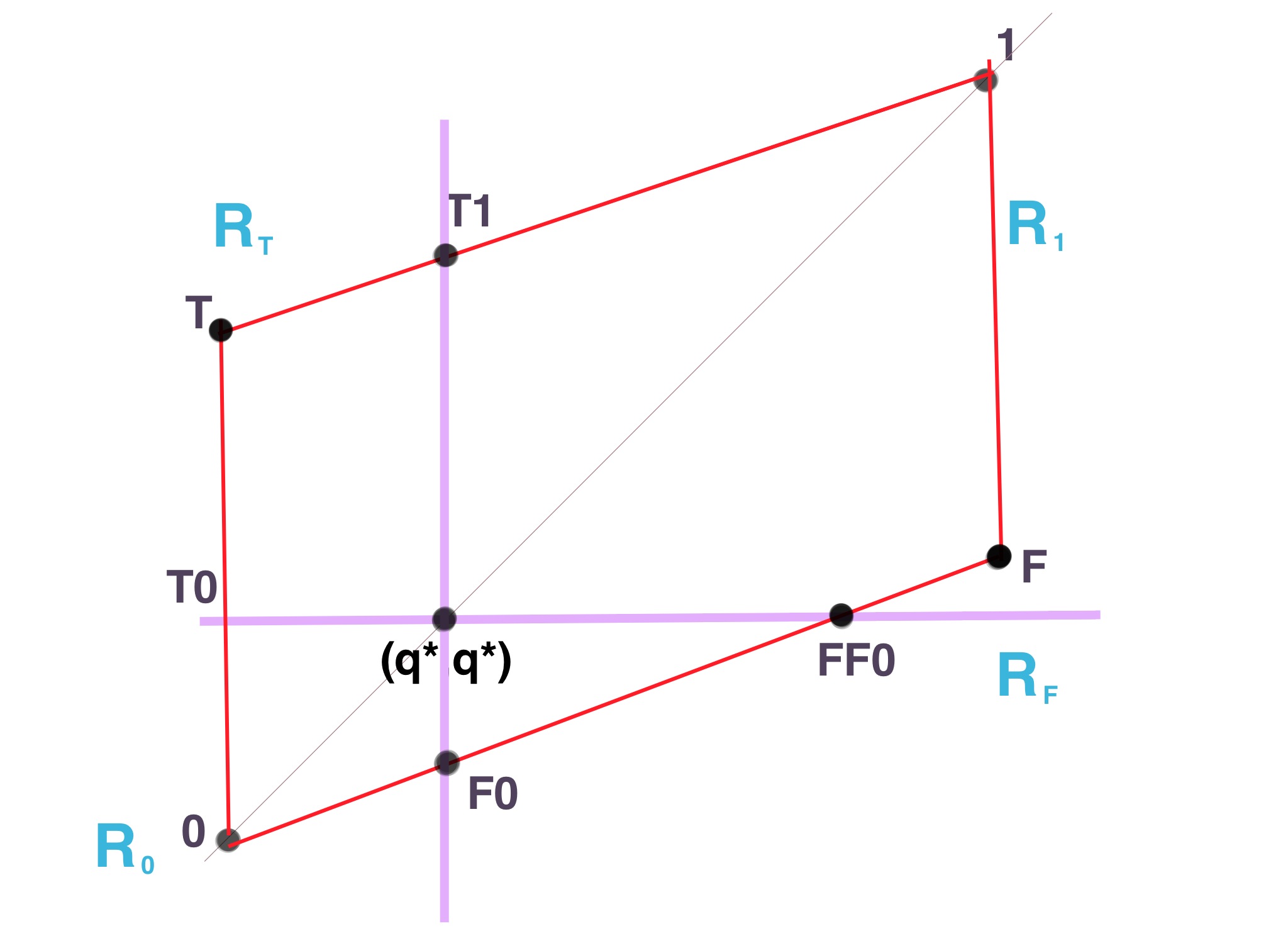}
\caption{In these best response plots, $ \fal $ is not an equilibrium.}
\label{fig:reportplot2}
\end{figure}

There are two basic cases: first $\fal \in R_{\fal}$ and second  $\fal \not\in R_{\fal}$.
%We list some symmetric equilibria here. \\
%\smallskip
If $ (\theta, \theta, ..., \theta) $ is a symmetric equilibrium, recall, we will use $ \theta $ to refer to it.
\begin{lemma}\label{claim:symmetriceq}
For each condition below, we prove existence of a set of symmetric equilibria of the peer prediction mechanism.

\begin{align*}
\mbox{If } \fal\in R_{\fal} \mbox{ then }&  ~ \tru, \fal, \zer,\one,\tru\zer,\tru\one,\fal\zer,\fal\one,\qstar  \mbox{ are symmetric equilibria.}\\
\mbox{If }  \fal\in R_{\zer}/R_{\fal} \mbox{ then }&  ~ \tru, \zer,\one,\tru\zer,\tru\one,\fal\one,\qstar  \mbox{ are symmetric equilibria.}\\
\mbox{If } \fal\in R_{\one}/R_{\fal} \mbox{ then }&  ~ \tru, \zer,\one,\tru\zer,\tru\one,\fal\zer,\qstar  \mbox{ are symmetric equilibria.}
\end{align*}\end{lemma}

\begin{proof}
All points are illustrated in Figures~\ref{fig:reportplot1} and~\ref{fig:reportplot2}.

Notice that it is always true that $ \tru\in R_{\tru}, \one\in R_{\one}, \zer\in R_{\zer} $. There are two cases: $ \fal\in R_{\fal}$ or $ \fal\notin R_{\fal} $.

We first consider $ \fal\in R_{\fal} $.
$ \tru $ is a symmetric equilibrium since if $ \forall i $, $ {\hat{q}}_i = \tru $ ($\theta_i=\tru$), then we have $ \forall i $, $  {\hat{q}}_{-i} = \tru $ and $ Supp(\tru)=\tru $. Based on Claim~\ref{eqplot}, $ \tru $ is a symmetric equilibrium. Similarly, $ \fal, \one, \zer $ are symmetric equilibria as well. $ \tru\one $ is a symmetric equilibrium since if $ \forall i $, $ {\hat{q}}_i = \tru\one $ ($\theta_i=\tru\one$), then we have $ \forall i $, $ {\hat{q}}_{-i} = \tru\one $ and $ Supp(\tru\one)=\{\tru,\one\} $. Based on Claim~\ref{eqplot}, $ \tru\one $ is a symmetric equilibrium. Similarly, $ \tru\zer,\fal\one , \fal\zer $ are symmetric equilibria as well. $ \qstar $ is a symmetric equilibrium since if $ \forall i $, $ {\hat{q}}_i = \qstar $, then we have $ \forall i $, ${\hat{q}}_{-i} = \qstar $ and $ Supp(\qstar)=X $. Based on Claim~\ref{eqplot}, $ \qstar $ is a symmetric equilirium.

Next, we consider $ \fal\notin R_{\fal} $; the results are proved similarly. But notice that in this case, $ \fal $ is not a equilibrium since $ {\hat{q}}_{-i} =\fal\notin R_{\fal} $. Also note, $ \fal\zer $ does not exist when $ \fal\in R_{\zer}/R_{\fal} $, and  $ \fal\one $ does not exist when $ \fal\in R_{\one}/R_{\fal} $.
\end{proof}

\paragraph{There are no asymmetric equilibria}

We will be aided by the following proposition, which states that once we have shown that all the ${\hat{q}}_i$ are in the same quadrant, if even one agent is in the interior, then all ${\hat{q}}_i$ must be $(q^*, q^*)$.

\begin{proposition}\label{prop:internal}  For any equilibrium, if there exists $s \in S$ such that for all $i$, ${\hat{q}}_{i} \in R_s$, and there exists an agent $j$ such that  ${\hat{q}}_{j} \in Int(R)$, then it must be that for all $i$,  ${\hat{q}}_{i} = (q^*, q^*)$.
\end{proposition}

%\gs{rewrite this proof}
\begin{proof}
Because $ {\hat{q} }_i\in Int(R) $, we know that  $  Supp({\hat{q} }_i)=X  $, meaning $ {\hat{q} }_i $ has full support. Based on Claim~\ref{eqplot}, we have $ {\hat{q} }_{-i}=\qstar $. Then the average of $ {\hat{q} }_j,$ for $j\neq i $, is an extreme point of $ R_s\cap R $, which is a convex set. Since all $ {\hat{q} }_j, j\neq i $ are in $ R_s\cap R $, we will have  $\forall  j\neq i, {\hat{q} }_j=\qstar $. Now we examine agent $ j\neq i $, because $ {\hat{q} }_j=\qstar\in Int(R) $, by the same reason we know $ {\hat{q} }_i=\qstar $ as well.
\end{proof}

We now show that the peer prediction mechanism does not have asymmetric equilibria.

\begin{lemma}\label{lemma:no-asymmetric}  Let $Q$ be a symmetric and positively correlated prior on $\{0, 1\}^n$, and let $\mathcal{M}$ be a corresponding peer-prediction mechanism.
%run with a strictly proper scoring rule with break-even $q^*$.
There are no asymmetric equilibria of the mechanism.  \end{lemma}

\begin{proof}%[Proof of Lemma \ref{lemma:no-asymmetric}]
To prove Lemma~\ref{lemma:no-asymmetric} we will consider two cases:  first $\fal \in R_{\fal}$ and second  $\fal \not\in R_{\fal}$. For each case we will prove the Lemma in two different steps. In the first step we will show that for any equilibrium, ${\hat{q}}_{i}$ is in the same quadrant of the best response plot for all agents. In the next step we will show that, in fact, all the  ${\hat{q}}_{i}$  must be the same for all agents, and moreover must be one of the equilibria detailed in the statement of Lemma~\ref{claim:symmetriceq}. 

\paragraph{Case 1, Step 1: Showing that if $\fal \in R_{\fal}$, then all ${\hat{q}}_{i}$ are in the same quadrant.}

First, we show that either $ \forall i,{\hat{q} }_i\in R_{1\rightarrow 0}$ or $ \forall i,{\hat{q} }_i\in R_{1\rightarrow 1}$. Second, we show that either $ \forall i,{\hat{q} }_i\in R_{0\rightarrow 0}$ or $ \forall i,{\hat{q} }_i\in R_{0\rightarrow 1}$. This will complete step 1.

For the sake of contradiction, we assume that it is not the case that  either $ \forall i,{\hat{q} }_i\in R_{1\rightarrow 0}$ or $ \forall i,{\hat{q} }_i\in R_{1\rightarrow 1}$. Then there exist agent $ i$ and agent $j$ such that $ {\hat{q}}_{i}(1|1)>q^* $ and $ {\hat{q}}_{j}(1|1)<q^* $, which implies that $ 1\rightarrow 1\in Supp({\hat{q}}_{i}) $ and $ 1\rightarrow 0\in Supp({\hat{q}}_{j}) $. If $ {\hat{q}}_{-\{i,j\}}(1|1)\geq q^* $, then $ {\hat{q}}_{-j}(1|1)>q^* $, which means $ {\hat{q} }_{-j}\notin R_{1 \rightarrow 0} $. But this is a contradiction, due to Claim~\ref{eqplot}.
The analysis for $ {\hat{q}}_{-\{i,j\}}(1|1)< q^* $ is similar.

The proof for either $ \forall i,{\hat{q} }_i\in R_{0\rightarrow 0}$ pr $ \forall i,{\hat{q} }_i\in R_{0\rightarrow 1}$ is similar to the previous proof when $ \fal\in R_{\fal} $.

\paragraph{Case 1, Step 2: Showing that if $\fal \in R_{\fal}$, then all ${\hat{q}}_{i}$ are the same and are equal to one of
$\{ \tru, \tru\one,$ $\tru\zer, \qstar, \fal, \fal\one, \fal\zer,\one,\zer \}$.}
We prove the Claim when $ \forall i, {\hat{q} }_i\in R_{\tru} $. The remaining cases proceed similarly.

We first consider when $ \exists i $ such that $ {\hat{q} }_i\notin \partial(R) $. By Proposition~\ref{prop:internal}, we have $ \forall i, {\hat{q} }_i=\qstar $. 

Next we consider when $ \forall i,{\hat{q} }_i\in \partial(R) $. If $ \forall i,{\hat{q} }_i=\tru $ then this a symmetric equilibrium $ \tru $. Otherwise, there exists $ i $ such that either $ Supp({\hat{q} }_i)=\{\tru, \one\} $ or $ Supp({\hat{q} }_i)=\{\tru, \zer\} $. If  $ Supp({\hat{q} }_i)=\{\tru, \one\} $, then based on Claim~\ref{eqplot}, $ {\hat{q} }_{-i}\in Conv(R_{\tru}\cap \partial(R))\cap R_{\{\tru,\one\}}=\{\tru\one\}$, which means $ {\hat{q} }_{-i}=\tru\one $. Then the average of $ {\hat{q} }_j, j\neq i $ is an extreme point of a convex set $ R_{\tru}\cap R $. Since all $ {\hat{q} }_j, j\neq i $ are in $ R_{\tru}\cap R $, we will have  $\forall  j\neq i, {\hat{q} }_j=\tru\one $. Now we examine agent $ j\neq i $, because $ {\hat{q} }_j=\tru\one $, by the same reason we know $ {\hat{q} }_i=\tru\one $ as well. The proof for $ Supp({\hat{q} }_i)=\{\tru, \zer\} $ is similar.

The results for other quadrants are shown similarly. In summary, we show $ (\theta_1,\theta_2,...,\theta_n) $ is an equilibrium iff all ${\hat{q}}_{i}$ are the same and are equal to one of $ \tru, \tru\one, \tru\zer, \qstar, \fal, \fal\one, \fal\zer,\one,\zer $.

\paragraph{Case 2, Step 1: Showing that if $\fal \not\in R_{\fal}$, then all ${\hat{q}}_{i}$ are in the same quadrant.}

Notice that under the case $ \fal\notin R_{\fal} $, when $ q(1|1)>q(0|0) $, $ \fal\in R_{\zer}  $; when $ q(1|1)\leq q(0|0) $, $ \fal\in R_{\one}  $. We show the proof for $ q(1|1)\geq q(0|0) $ and $ \fal\in R_{\zer}  $. The other case, when $ q(1|1)<q(0|0) $ and $ \fal\in R_{\one}  $, is analogous.

First, we show that  either $ \forall i,{\hat{q} }_i\in R_{1\rightarrow 0}$ or $ \forall i,{\hat{q} }_i\in R_{1\rightarrow 1}$. For the sake of contradiction, we assume that it is not the case that  either $ \forall i,{\hat{q} }_i\in R_{1\rightarrow 0}$ or $ \forall i,{\hat{q} }_i\in R_{1\rightarrow 1}$. Then there exist agent $ i$ and agent $j$ such that $ {\hat{q}}_{i}(1|1)>q^* $ and $ {\hat{q}}_{j}(1|1)<q^* $. If $ {\hat{q}}_{-\{i,j\}}(1|1)\geq q^* $, then $ {\hat{q}}_{-j}(1|1)>q^* $ which means $ {\hat{q} }_{-j}\notin R_{1 \rightarrow 0} $, but this is a contradiction based on Claim~\ref{eqplot}.
The analysis for $ {\hat{q}}_{-\{i,j\}}(1|1)< q^* $ is similar.

But to show  either $ \forall i,{\hat{q} }_i\in R_{0\rightarrow 0}$ or $ \forall i,{\hat{q} }_i\in R_{0\rightarrow 1}$ we cannot use the same approach as before. We will show it under two cases by contradiction:

First, we consider when $ \forall i,{\hat{q} }_i\in R_{1\rightarrow 1} $. (This is illustrated in Figure~\ref{fig:reportplot2} when all points are in upper half-plane):

If $ {\hat{q}}_{i}(1|0)>q^* $, $ {\hat{q}}_{j}(1|0)<q^* $ and $ {\hat{q}}_{-\{i,j\}}(1|0) \geq q^* $, combining with $ \forall i, {\hat{q}}_i\in R_{1\rightarrow 1} $, we obtain $ i\in R_{\one} $, $ j\in R_{\tru} $. Then $ \tru\in Supp({\hat{q} }_j) $ while $ {\hat{q} }_{-j}\notin R_{\tru}\cup R_{\zer} $ which is a contradiction. If $ \hat{q}_{-\{i,j\}}(1|0) < q^* $, then $ \one\in Supp({\hat{q} }_i) $ while $ {\hat{q} }_{-i}\notin R_{\one}\cup R_{\fal} $, which is a contradiction. 

Next we consider when $ \forall i,{\hat{q} }_i\in R_{1\rightarrow 0} $. (This is illustrated in Figure~\ref{fig:reportplot2} when all points are in lower half-plane):

If  $ {\hat{q}}_{i}(1|0)>q^* $, $ {\hat{q}}_{j}(1|0)<q^* $ and $ {\hat{q}}_{-\{i,j\}}(1|0) \leq q^* $, we have $ i\in R_{\fal} $, $ j\in R_{\zer} $. Then $ \fal\in Supp({\hat{q} }_i) $ while $ {\hat{q} }_{-i}\notin R_{\one}\cup R_{\fal} $, which is a contradiction. 

If $ {\hat{q}}_{i}(1|0)>q^* $, $ {\hat{q}}_{j}(1|0)<q^* $ and $ {\hat{q}}_{-\{i,j\}}(1|0) > q^* $, then $ {\hat{q} }_{-j}\notin R_{\tru}\cup R_{\zer} $. But even if $ j\in R_{\zer} $, we do not have the fact that $ \zer\in Supp({\hat{q} }_j) $, since $ \fal\in R_{\zer} $ and $ {\hat{q} }_j $ can be $ \fal $, whose support set only contains $ \fal $. We cannot obtain a contradition using the previous way. We prove it in two steps: first, we show that $ {\hat{q} }_j(1|1) $ is strictly less than $ q^* $; second, we show that when $ {\hat{q} }_j(1|1) $ is strictly less than $ q^* $, $ {\hat{q} }_i(1|0) $ cannot be strictly greater than $ q^* $.

Based on Claim~\ref{eqplot}, $ Supp({\hat{q} }_{j})\subseteq \{\fal,\one\} $,  which means $ {\hat{q} }_{j} $ is on the line incident on $ \fal $ and $ \one $. Combining with the fact that $ {\hat{q} }_{j}\in R_{\zer} $, we have $ {\hat{q}}_{j}(1|1)<q^* $, then $   {\hat{q}}_{-i}(1|1)<q^* $. Since $ {\hat{q}}_{-i}(1|1)<q^* $, based on Claim~\ref{eqplot}, $ Supp(\theta_i)\subseteq \{\fal,\zer\} $ while $ {\hat{q}}_{i}(1|0)>q^*\Rightarrow {\hat{q} }_{i}\notin R_{\zer} $ which is a contradition since both $ \fal $ and $ \zer $ are in $ R_{\zer} $.

\paragraph{Case 2, Step 2: Showing that if $\fal \not\in R_{\fal}$, the all ${\hat{q}}_{i}$ are the same and are equal to one of
$ \tru, \tru\one, \tru\zer, \qstar, \fal\one, \\\one,\zer $ if $ \fal\in R_{\zer} $ or $ \tru, \tru\one, \tru\zer, \qstar, \fal\one,\one,\zer $ if $ \fal\in R_{\one} $.}

The analyses for all points in $ R_{\tru} $ and $ R_{\one} $ are similar to the previous analysis under the case where $ \fal\in R_{\fal} $. We will show the analyses for $ R_{\fal} $ and $ R_{\zer} $ here. 

First, if all points are in the same quadrant and there is $ i $ such that $  {\hat{q}}_i $ is not on the boundary, then $ \forall i, {\hat{q}}_i= \qstar $, based on Proposition~\ref{prop:internal}. So then we only need to consider the case when all points are on the boundary.

If $ \forall i, {\hat{q}}_i\in R_{\fal}\cap \partial(R) $, then $ \exists i$ such that $Supp(\theta_i)=\{\one, \fal\} $. We have $ {\hat{q}}_{-i}=R_{\{\one,\fal\}}\cap \partial(R)=\fal\one $. Then since $ \fal\one $ is an extreme point and all points are in $R_{\fal}$, we have $ \forall j\neq i, {\hat{q}}_j= \fal\one $ and based on the same reasoning, $ {\hat{q}}_i= \fal\one $, as well. 

If $ \forall i,{\hat{q} }_i\in {\hat{q}}_i\in R_{\zer}\cap\partial(R) $, when $ \forall i,{\hat{q} }_i=\zer $, then this is a symmetric equilibrium. Otherwise:
\begin{enumerate}[label=(\alph*)]
\item When $ \exists i $ such that $ Supp({\hat{q} }_i)=\{\tru,\zer\} $:

In Claim~\ref{eqplot}, $ {\hat{q} }_{-i}\in Conv(R_{\zer}\cap \partial(R))\cap R_{\{\tru,\zer\}}=\{\tru\zer\}$ which means $ {\hat{q} }_{-i}=\tru\zer $. Then since $ \tru\zer $ is an extreme point of $ R_{\zer}\cap R $ and $ \forall j ,{\hat{q} }_j\in R_{\zer}\cap R$, we have $ \forall j\neq i ,{\hat{q} }_j=\tru\zer $. We also have $ {\hat{q} }_i=\tru\zer $ based on the same reason since $ Supp(\tru\zer)=\{\tru,\zer\} $. So we have $ \forall i, {\hat{q} }_i=\tru\zer $.  
 
 \item When $ \exists i $ such that $ Supp({\hat{q} }_i)=\{\fal,\one\} $:
 
Based on Claim~\ref{eqplot}, $ {\hat{q} }_{-i}\in Conv(R_{\zer}\cap \partial(R))\cap R_{\{\fal,\one\}}=\emptyset$ which means situation (b) is impossible, that is, there does not exist any $i$ such that $ Supp({\hat{q} }_i)=\{\fal,\one\} $. 
%\katrina{What is (b)?}\yk{I modified the enumeration labels and add some words here}

\item When $ \exists i $ such that $ Supp({\hat{q} }_i)=\{\fal,\zer\} $ or $ {\hat{q} }_i=\fal $:

Based on Claim~\ref{eqplot}, $ {\hat{q} }_{-i}\in Conv(R_{\zer}\cap \partial(R))\cap R_{\fal}$. $ Conv(R_{\zer}\cap \partial(R))\cap R_{\fal}$ is a point and let $ \fal\fal\one $ denote it. $ {\hat{q} }_{-i}=\fal\fal\one $. Then since $ \fal\fal\one $ is an extreme point, we have $ \forall j\neq i, {\hat{q} }_j= \fal\fal\one $. So $ Supp({\hat{q} }_j) $ is $ \{\fal,\one\} $ while $ {\hat{q} }_{-j}\notin R_{\fal} $ when $ Supp({\hat{q} }_i)=\{\fal,\zer\} $ or $ {\hat{q} }_i=\fal $. However, we already proved there does not exist any $i$ such that $ Supp({\hat{q} }_i)=\{\fal,\one\} $, so situation (c) cannot happen as well. \qedhere
\end{enumerate}
% \katrina{What is (c)?}\yk{I add some words here}
\end{proof}

%cut end here

In summary, we show $ (\theta_1,\theta_2,...,\theta_n) $ is an equilibrium iff all ${\hat{q}}_{i}$ are the same and are equal to one of $ \tru, \tru\one, \tru\zer, \qstar,\fal\one,\one,\zer $.

Now we consider the case when $ q(1|1)\leq q(0|0) $. Then $ \fal \in R_{\one} $, and, based on a similar analysis as before, we show that $ (\theta_1,\theta_2,...,\theta_n) $ is an equilibrium iff all ${\hat{q}}_{i}$ are the same and are equal to one of $ \tru, \tru\one, \tru\zer, \qstar,\fal\zer,\one,\zer $.

\paragraph{Translating from best response plot back to strategy profiles}
We calculate the strategies that are mapped to $ \tru\zer, \fal\one, \tru\one,\fal\zer $. 

For $ \tru\zer $, setting ${\hat{q}}(1|1) = q^*$ and $\theta(1|0) = 0$ in Equation~\ref{eqn:p1} we get $$q^* = \theta(1|1) q(1|1) \Rightarrow \theta(1|1) = \frac{q^*}{q(1|1)}.$$

For $ \fal\one $, setting ${\hat{q}}(1|1) = q^*$ and $\theta(1|0) = 1$ in Equation~\ref{eqn:p1} we get $$q^* = q(0|1) + \theta(1|1)q(1|1) \Rightarrow \theta(1|1) = \frac{q^* - q(0|1)}{q(1|1)}.$$

For $ \tru\one $, this is analogous to the case where players receiving 0 always tell the truth  and players receiving 1 sometimes lie, and the result follows by the above calculations by simply switching the labels 0 and 1.  Before we showed that $\theta(1|1) = \frac{q^*}{q(1|1)}$.  Switching the labels gives us:
$$\theta(0|0) = \frac{q^*(0)}{q(0|0)} \Rightarrow \theta(1|0) = 1 -  \frac{1 - q^*}{q(0|0)} = \frac{q^* - q(1|0)}{q(0|0)}.$$

For $ \fal\zer $, this is analogous to the case where players receiving 0 always lie and players receiving 1 sometimes lie, and the result follows by the above calculations by simply switching the labels 0 and 1.
      $$\theta(0|0) = \frac{q^*(0) - q(1|0)}{q(0|0)}\Rightarrow \theta(1|0) = 1 - \frac{1 - q^* - q(1|0)}{q(0|0)} =  \frac{q^*}{q(0|0)}.$$

% % % % % % % % % % % % % % % % % % % % %
\else

\paragraph*{Proof Overview for Theorem~\ref{eqbinary}}

We use the best response plot as a tool to compute all the equilibria of the peer prediction mechanism. The proof proceeds in three steps. First, we use the best response plot to find some symmetric equilibria.
 %If $ (\theta, \theta, ..., \theta) $ is an equilibrium, then we write $ \theta $ for that symmetric equilibrium.
 We show
 \begin{packed_enumerate}
 \item If $ \fal\in R_{\fal} $, then $ \tru, \fal, \zer,\one,\tru\zer,\tru\one,\fal\zer,\fal\one,\qstar $ are symmetric equilibria (see Figure~\ref{fig:reportplot1}).
\item If $ \fal\in R_{\zer} \setminus R_{\fal} $, then  $ \tru, \zer,\one,\tru\zer,\tru\one,\fal\one,\qstar $ are symmetric equilibria.
\item If $ \fal\in R_{\one} \setminus R_{\fal} $, then $ \tru, \zer,\one,\tru\zer,\tru\one,\fal\zer,\qstar $ are symmetric equilibria (see Figure~\ref{fig:reportplot323}).
\end{packed_enumerate}
Second, we show that no asymmetric equilibria nor additional symmetric equilibria exist.  At this point, we have classified all equilibria in terms of ${\hat{q}}$.  \katrina{Remind the reader what hat{q} is.} The third step inverts the mapping from $\theta$ to ${\hat{q}}$ to solve for the equilibria in terms of $\theta$. \katrina{Remind the reader what theta is.}

\begin{figure}
\centering
\includegraphics[scale=0.09]{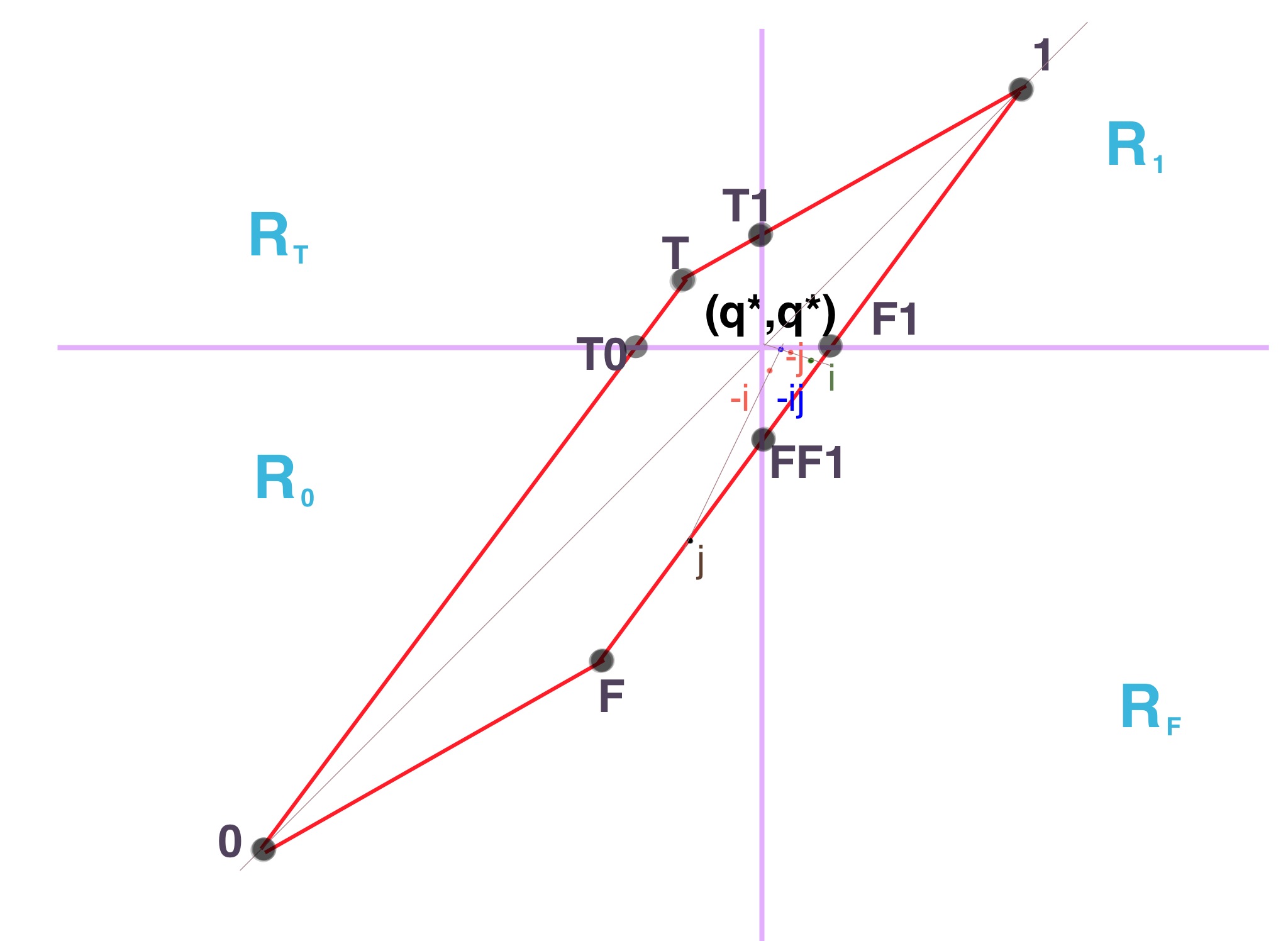}
\includegraphics[scale=0.09]{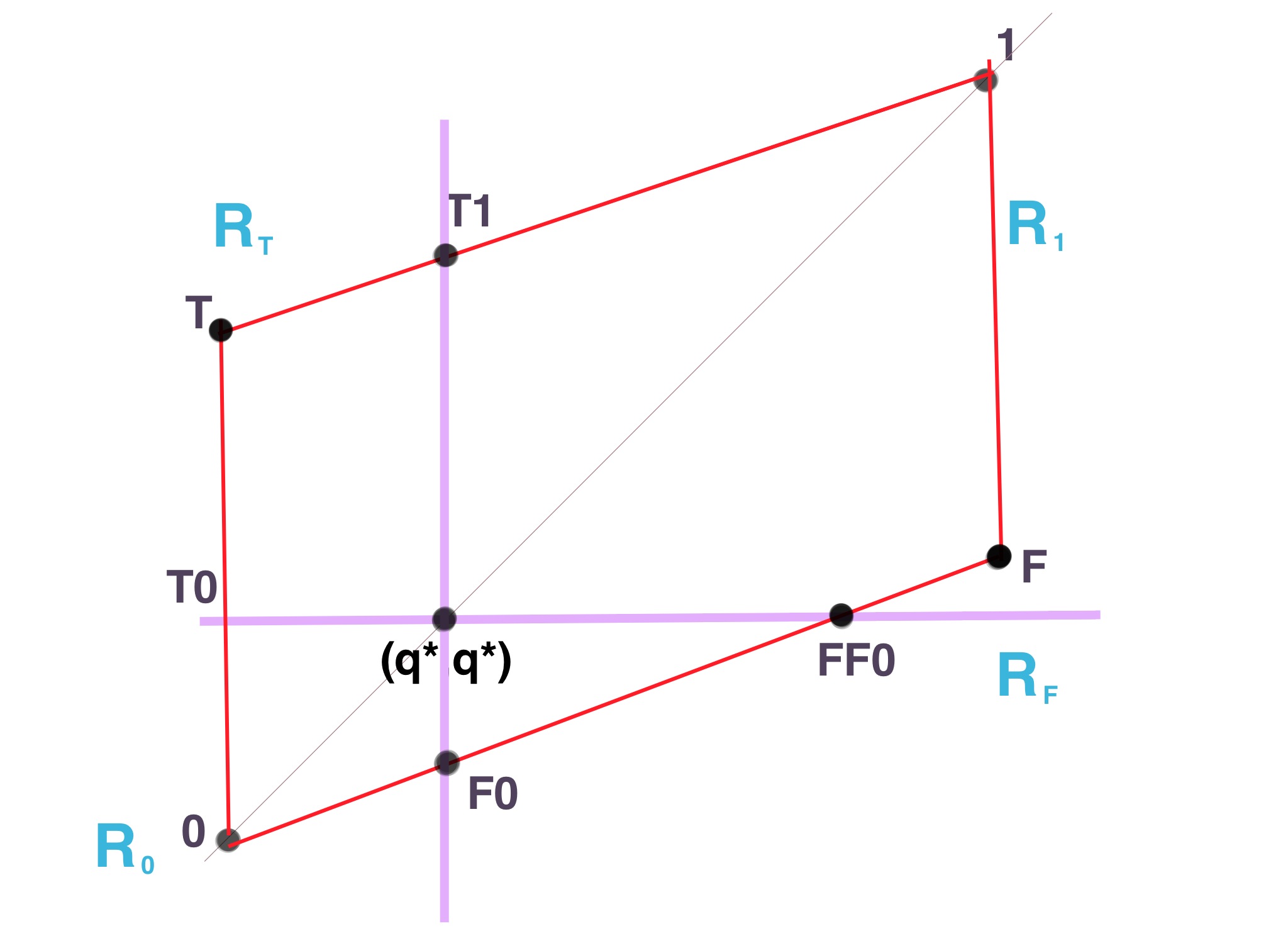}
\caption{In these best response plots, $ \fal $ is not an equilibrium.}
\label{fig:reportplot323}
\end{figure}

The proof of Theorem~\ref{eqbinary} appears in the supplementary materials.
\fi
\yk{end here}
%Claim~\ref{eqplot} is our main tool to prove the theorem. We use it to list some symmetric equilibria. To prove there is no asymmetric equilibria we will consider two different case:  first $\fal \in Q_{\fal}$ and second  $\fal \not\in Q_{\fal}$.  For each case we will prove it in two different steps.  In the first step we will show that for any equilibrium, ${\hat{q}}_{i}$ is in the same quadrant of the best reponse plot for all agents.  In the next step we will show that, in fact, all the  ${\hat{q}}_{i}$  must be the same and moreover must be one of the equilirbria in the theorem statement.

%See Proof for theorem~\ref{eqbinary} in full version.

%

%
%
%
%\end{document} 
%\input{peer-eq-analysis}

\section{Making truth-telling focal among the informative equilibria} \label{sec:focal-for-some}
In this section, we first introduce a useful concept, line sets, and then we build up a sequence of lemmas that will be useful in Section~\ref{sec:optimization}. Along the way, we show Theorem~\ref{exist}, that there exists a proper scoring rule that makes truth-telling have higher payoff than any other informative equilibrium.

\subsection{Line Sets}
We now introduce notation that highlights the role of three constants that will emerge as the key defining parameters of any proper scoring rule. Useful related observations related to scoring rules appear in Appendix~\ref{sec:appendix}.

%\begin{notation} We now introduce the shorthand \[\ell(x,b) :=PS(x,q(1|b)),\] which we will use throughout Section~\ref{sec:focal-for some}. \katrina{Is that section number right? do we want to include it?} For a given strictly proper scoring rule, there exist constants $\alpha,\beta,\gamma$ such that we we can rewrite \begin{align*} \ell(x,0)&=\alpha(x-q^*(1))+\gamma\ \ell(x,1)&=\beta(x-q^*(1))+\gamma. \end{align*} where $\alpha=f(q(1|0))$ and $\beta=f(q(1|1))$.  \end{notation}

\begin{definition}\label{def:line-set}
We define a $ (\alpha,\beta,q^*,\gamma) $ line-set as a set of linear functions $ \{\ell(x,0),\ell(x,1)\} $ where $ \ell(x,0)=\alpha(x-q^*)+\gamma $, $ \ell(x,1)=\beta(x-q^*)+\gamma $. We say a proper scoring rule and a $(\alpha,\beta,q^*,\gamma) $ line-set {\em correspond} if  \[PS(x,q(1|b))=\ell(x,b).\]
\end{definition}

\yk{I move the definition of payoff function to preliminaries}

Notice that there is a natural mapping from a $ (\alpha,\beta,q^*,\gamma) $ line-set to payoff fucntion matrix $\left( \begin{array}{cc}
\ell(1,1) & \ell(1,0) \\
\ell(0,1) & \ell(0,0)
\end{array} \right) $.

It is useful to note that arbitrary convex functions can be converted into scoring rules as follows.
\begin{fact}\label{lem:bregman}\cite{GR07}
Given any (strictly) convex, differentiable function $r: \mathbb{R} \rightarrow \mathbb{R}$, and any arbitrary function $h: \mathbb{R} \rightarrow \mathbb{R}$, the function $PS(p, q) := -r(p) + r(q) + \nabla r(q) \cdot(p - q) + h(p)$ is a (strictly) proper scoring rule. Note that the first three terms are the negation of the Bregman divergence of $r$.\footnote{A more general version in higher dimensions also holds, as does a converse of the equivalence. However, we do not use these stronger statements here.}
\end{fact}

As a consequence, for a certain class of $\alpha$, $\beta$, $q^*$, and $\gamma$, one can find a corresponding proper scoring rule.
\begin{lemma}\label{lem:build-ps}
Given any $ (\alpha,\beta,q^*,\gamma) $ line-set (or payoff function) with $\beta < \alpha $ and $q^* \in [0, 1]$,  there exists a strictly proper scoring rule that corresponds to this line-set.
\end{lemma}

\begin{proof}
Essentially, appealing to Fact~\ref{lem:bregman}, our goal is to produce a strictly convex function $r$ and an arbitrary $h(\cdot)$ such that the proper scoring rule constructed via the Bregman divergence matches the desired scoring rule at $q = q(1|0)$ and $q = q(1|1)$. This gives constraints on $r$ that $\nabla r (q(1|0)) = \beta$ and $\nabla r (q(1|1)) = \alpha$, and $r(q(1|1)) - r(q(1|0)) = \alpha (q(1|1) - q^*) - \beta (q(1|0) - q^*)$.

These two tangent lines to $r$, at $q(1|0)$ and $q(1|1)$, intersect at $q^*$. Recall that $q(1|0)  < q^* < q(1|1)$. Thus, between $q(1|0)$ and $q(1|1)$, one can construct $r(\cdot)$ as a quadratic Bezier curve given $r(q(1|0)), r(q(1|1)),$ and their intersection at $q^*$. Outside of $[q(1|0), q(1|1)]$, one can extend $r(\cdot)$ with a spline.
\end{proof}

\subsection{Making truth-telling focal among informative equilibria}

Informally, the following theorem states that under a weak condition, there is a $ (\alpha,\beta,q^*,\gamma) $ line-set such that the proper scoring rule satisfying this line-set makes truth-telling focal among all informative equilibria.

 \begin{theorem}\label{exist}
 Let $Q$ be a symmetric and positively correlated prior over binary signals. Then there exists $r>0$  such that for any strictly proper scoring rule that corresponds to a $ (\alpha,\beta,q^*,\gamma) $ line-set with $r=-\frac{\alpha}{\beta}$, the payoff of $ \tru $ is no less than that of any other informative equilibrium. Moreover, if $ q(1|0)\neq q(0|1)$,  $q(1|1) \neq q^*$, and $q(0|0) \neq q^*$,  the payoff of $ \tru $ is strictly greater than that of any other informative equilibrium.
 \end{theorem}

\paragraph*{Motivation for Best Response Payoff Plot}
We want to select a proper scoring rule in order that truth-telling will have higher expected payoff than any other equilibrium.  If we just compute the average payoffs of agents playing symmetric strategies, then even for the simplest proper-scoring rules, these payoff curves are paraboloid, and hence difficult to analyze directly.  Instead we will analyse the payoffs of the  best-responses.  Because the best responses are the same in each quadrant, the best response payoff plot is piece-wise linear for all proper scoring rules, which makes the analysis tractable.    

A key tool will be the best response payoff plot which is defined by extending the definition of the best response plot to include not only the best response strategy, but additionally the expected payoff of playing this best-response strategy.  A contour on the plot corresponds to a set of points that have identical best response payoffs.
A key observation is that for any equilibrium, agent $i$'s payoff is equal to his best response payoff (this is not true for non-equilibrium strategy profile).

Since the best response payoff of any equilibrium is equal to that equilibrium's payoff, the contour plot can be used to compare the payoffs of different equilibria.

The rest of this section will formally define the best response payoff and the contour plot, and finally give a proof outline.  The actual proof is in Section~\ref{sec:proof-of-main-tech-Lemma}.

\paragraph{Best Response Payoff}
We extend the definition of the best response plot of Section~\ref{sec:bestresponseplot}.
%\katrina{There should be a ref to the actual definition.}\yk{The whole section is the definition of best response plot.}
We only consider symmetric strategies (because we have shown that all the equilibria are symmetric).  However, we fix some player $i$ and consider his best response when all of the other players are playing some symmetric strategy $(\hat{q}_{-i}(1|0),\hat{q}_{-i}(1|1))$; we write the payoff of this best response as $BR(\cdot, \cdot)$.
%Note that the strategy $ \theta_{-i} $ used by all agents excluding $ i $ can be mapped to a point $ (\hat{q}_{-i}(1|0),\hat{q}_{-i}(1|1)) $ where $ \hat{q}_{-i}=\theta_{-i} q $, and this mapping is injective when $q(1|0) \neq q(1|1)$. For example, $ (q(1|0),q(1|1)) $ corresponds to truth-telling; $ (1,1) $ corresponds to always reporting 1. This best response is given by
%\katrina{I think one-to-one is risky terminology, since it's not clear enough whether we mean an injection or bijection.}
%\jz{I believe we mean injection here; I think this is all we need. We just want to make sure that several equilibria are not mapped to the same point}
%\yk{Maybe we should move this to preliminaries. Also I define $ \hat{q}_{-i} $ before as the average of $ \hat{q}_j $ where $ j\neq i $. But when strategy is symmetric, $ \hat{q}_{-i}=\theta_{-i} q $  }
%\jz{changed to injective}
\begin{align*}
BR(\hat{q}_{-i}(1|0),\hat{q}_{-i}(1|1))=q(0) \left( \max_{b_0\in\{0,1\}}\ell(\hat{q}_{-i}(1|0),b_0) \right)+q(1) \left( \max_{b_1\in\{0,1\}}\ell(\hat{q}_{-i}(1|1),b_1) \right).
\end{align*}

Now notice that we can express $q(1)$, the probability a priori that a player receives a signal of one, in terms of $q(1|0)$ and $q(1|1)$:
\begin{align*}
q(1)&= q(0) \cdot q(1 | 0) + q(1) \cdot q(1 | 1)\\
    &= (1 -  q(1)) \cdot q( 1 | 0) + q(1) \cdot q( 1 | 1)\\
    &= \frac{q(1 | 0)}{1 - q(1 | 1) + q(1 | 0)}.
\end{align*}
This allows us to express the best expected payoff $ BR(\hat{q}_{-i}(1|0),\hat{q}_{-i}(1|1)) $ only in terms of $q(1 | 0),$ $q(1 | 1),$ $\hat{q}_{-i}(1|0),$ and $\hat{q}_{-i}(1|1)$.

\paragraph{Contour Plot} We introduce the notion of a {\em contour plot} depicting $ BR(\hat{q}_{-i}(1|0),\hat{q}_{-i}(1|1)) $, where the $x$ and $y$ axes are, respectively, $\hat{q}_{-i}(1|0)$ and $\hat{q}_{-i}(1|1)$. A {\em contour} on a contour plot corresponds to the set of points such that
\begin{align*}
BR(\hat{q}_{-i}(1|0),\hat{q}_{-i}(1|1)) = \nu,
\end{align*}
for some constant $\nu$.  An important observation is that for any equilibrium, agent $i$'s payoff is equal to his best response payoff (this is not true in general).

\paragraph{Proof Outline}%\yk{I modified here to fit short version}
\ifnum\fullversion=1
The main tool we use to prove Lemma~\ref{exist} is the contour plot. Since the best response payoff of any equilibrium is equal to that equilibrium's payoff, the contour plot can be used to compare the payoffs of different equilibria. The plot is centered at $(q^*, q^*)$ and thus the four quadrants, $R_\one, R_\zer, R_\fal, R_\tru$,  correspond  to whether each of $\hat{q}_{-i}(1|0)$ and $\hat{q}_{-i}(1|1)$ is greater or less than $q^*$. By the definition of $q^*$, each quadrant corresponds to a different best response.

We will first show that in each quadrant, a contour is actually a line whose slope depends only on the ratio of the slopes, $\frac{\beta}{\alpha}$, from a $ (\alpha,\beta,q^*,\gamma) $-line set.  These lines make it easier for us to compare the payoffs of various equilibria: the further a line is from the center $\qstar$ of the plot, the higher its payoff is. We want to be able to compare the lines that go through the equilibrium points in some way, but unfortunately, the equilibria are not necessarily all in the same quadrant. Surprisingly, we are able to define a mapping that translates lines from one quadrant to another in a way that preserves the payoffs at equilibrium and does not depend on the values of $\alpha$ and $\beta$. When all translated equilibria are in the same quadrant,
we will show that truth telling is an extreme point of the convex hull $ \mathcal{H} $ of all translated informative equilibria. We do so using Lemma~\ref{main_Lemma}, which pairs up equilibria such that the lines going through each pair of equilibria are parallel, and such that the truthful equilibrium is on the highest payoff line amongst the informative equilibria (see Figure~\ref{fig:nine}). Once we show Lemma~\ref{main_Lemma}, we can see by taking  $\alpha$ and $\beta$ such that the slope of the contours in $ R_{\tru} $ equals the slope the line incident on the two extreme equilibria adjacent to truth-telling, we ensure that the truthful equilibrium pays at least well any other equilibrium (see Corollary~\ref{the slope of contours}). Based on Lemma~\ref{main_Lemma} and Corollary~\ref{the slope of contours}, we show Lemma~\ref{exist}.  %\gs{I changed this, I hope for the better!}
\gs{reword?}
\else
%\begin{figure}
%\centering
%\includegraphics[scale=0.19]{nine.jpg}
%\caption{Best response plot of a prior and given $ q^* $ with nine equilibria. Notice that we can pair up equilibria such that
%the lines going through each pair of equilibria are parallel. $ \tru $ is an extreme point of all translated informative equilibria and $ \tru\one $ and $ f(\fal) $ are $ \tru $'s two adjacent extreme equilibria. }
%\label{fig:nine}
%\end{figure}
The full version of the proof appears in the supplementary materials.
\begin{figure}
\centering
\includegraphics[scale=0.08]{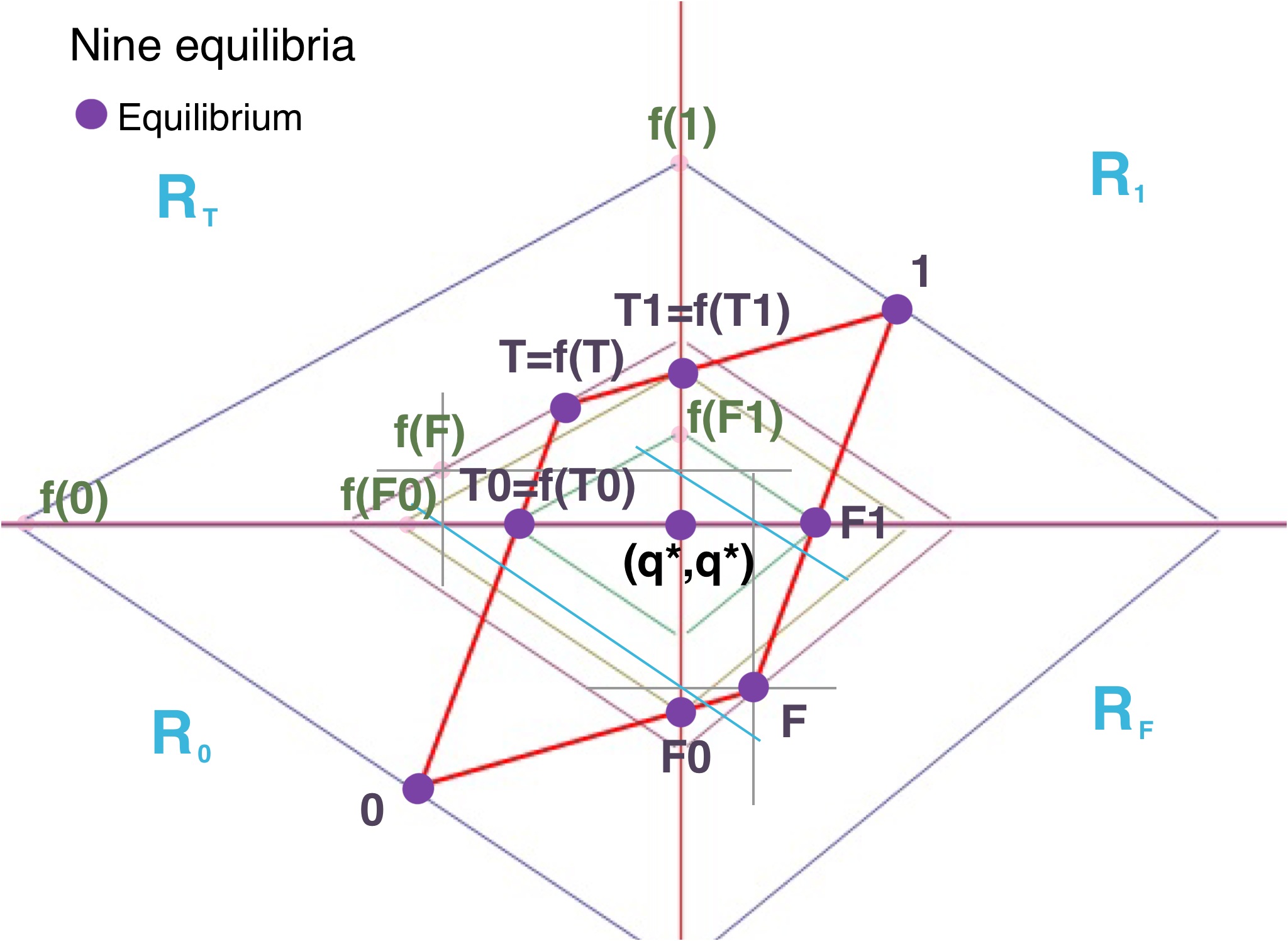}
\includegraphics[scale=0.08]{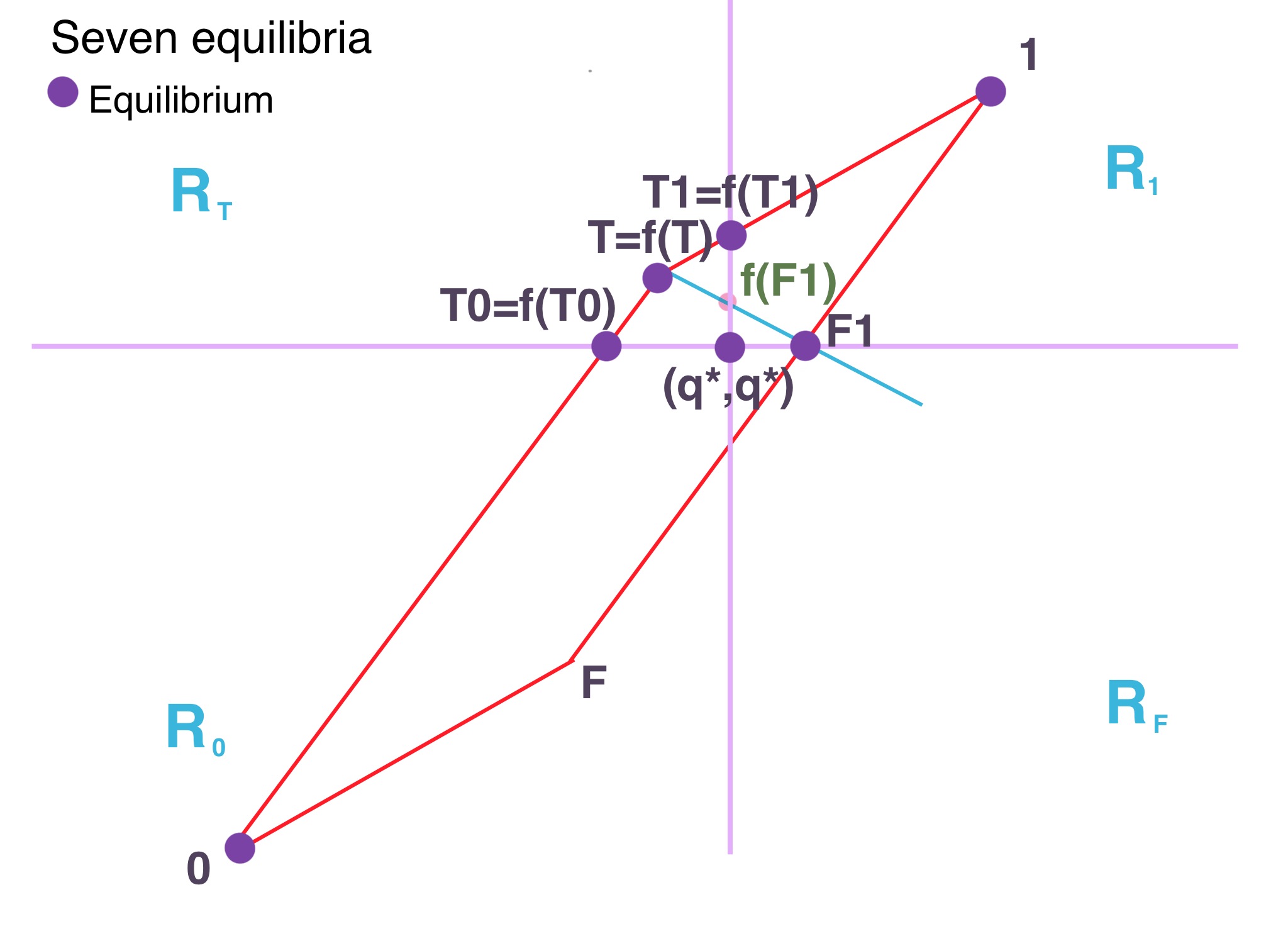}
\caption{(1)Left figure: Best response plot of a prior and given $ q^* $ with nine equilibria. The purple points with black words are equilibria and the pink points with green words are translated equlibria. Notice that we can pair up equilibria such that the lines going through each pair of equilibria are parallel. $ \tru $ is an extreme point of all translated informative equilibria and $ \tru\one $ and $ f(\fal) $ are $ \tru $'s two adjacent extreme equilibria.\newline (2)Right figure: Best response plot of a prior and given $ q^* $ with seven equilibria. Notice that $ \tru $ is an extreme point of all translated informative equilibria and $ \tru\one $ and $ \tru\zer $ are $ \tru $'s two adjacent extreme equilibria. }
\label{fig:seven}\label{fig:nine}
\end{figure}
The main tool we use to prove Theorem~\ref{exist} is the contour plot, as in Figure~\ref{fig:nine}. Since the best response payoff of any equilibrium is equal to that equilibrium's payoff, the contour plot can be used to compare the payoffs of different equilibria. The plot is centered at $(q^*, q^*)$ and thus the four quadrants, $R_\one, R_\zer, R_\fal, R_\tru$,  correspond  to whether each of $\hat{q}_{-i}(1|0)$ and $\hat{q}_{-i}(1|1)$ is greater or less than $q^*$. By the definition of $q^*$, each quadrant corresponds to a different best response.

We will first show that in each quadrant, a contour is actually a line whose slope depends only on the ratio of the slopes, $\frac{\beta}{\alpha}$, from a $ (\alpha,\beta,q^*,\gamma) $-line set.  These lines make it easier for us to compare the payoffs of various equilibria: the further a line is from the center $\qstar$ of the plot, the higher its payoff is. We want to be able to compare the lines that go through the equilibrium points in some way, but unfortunately, not all equilibria are  necessarily in the same quadrant. Surprisingly, we are able to define a mapping that translates lines from one quadrant to another in a way that preserves the payoffs at equilibrium and does not depend on the values of $\alpha$ and $\beta$. This will allow us to pair up equilibria such that the lines going through each pair of equilibria are parallel, and such that the truthful equilibrium is on the highest payoff line amongst the informative equilibria (see Figure~\ref{fig:nine}). From this we can see that truth-telling is an extremal point of the translated equilibria.  By taking  $\alpha$ and $\beta$ such that the slope of the payoff function equals the slope the line incident on the two extremal equilibrium adjacent to truth-telling, we ensure that the truthful equilibrium pays at least well any other equilibrium.  %\gs{I changed this, I hope for the better!}
\fi

%\yk{end here}

%The contour plot is the main tool to prove this Theorem. Since the best response payoff of any equilibrium equals that equilibrium's payoff, the contour plot defined above can be used to compare equilibria's payoffs.\\
%Firstly, we show that contours are lines in each quadrants. Then we show that the slope of the lines depends on $ a,b $. If all equilibria were in the same quadrant, the problem becomes linear programming. And using the coordinate of each equilibrium, we can easily determined whether there exists $  a,b $ such that $ \tru $ obtains the highest payoff. Although the equilibria are in different quadrants, we find a way to translate all equilibria to one quadrant while preserving the best response payoff. Moreover, this translation does not depend on $ a,b $. In other words, we find a quadrant such that for any equilibrium, there is a point in that quadrant which has the same best response as that equilibrium. We then show that the arrangements of those points implies $ a,b $ such that $ \tru $ has highest payoff among all informative symmetric equilibria.\\

%We now build up the technical tools that will allow us to formally prove Theorem~\ref{exist}.
%\begin{definition}[Quadrants]
%The quadrants of the contour plot are defined with respect to $ q^* $: $ R_{\one} $ is the region where $ \hat{q}(1|0)\geq q^*$ and $\hat{q}(1|1)\geq q^* $. $ R_{\tru} $, $ R_{\zer} $, and $ R_{\fal} $ are similarly defined in a counter-clockwise fashion on the contour plot.
%\end{definition}

\subsection{Proof of Theorem~\ref{exist}}\label{sec:proof-of-main-tech-Lemma}

Our goal is to find a $ (\alpha,\beta,q^*,\gamma) $-line set, making the truthful equilibrium at least as good as any other equilibrium.  Since the difference between the payoffs of any two equilibria does not depend on $ \gamma $, we can assume $ \gamma=0 $ now without loss of generality, but will in Section~\ref{sec:optimization} discuss the value of $ \gamma $ in the context of maximizing the payoff of truth-telling with respect to other equilibria. (We will see that the result of Theorem~\ref{opt_thm} is obtained by properly choosing $\alpha,\beta,\gamma$.)

 %\katrina{This should be a precise forward reference.}\yk{I explained more here for where we use the value of $\gamma$ in future.}

The claim below will help us construct a mapping that translates lines from one quadrant to another in a way that preserves the payoffs at equilibrium and does not depend on the values of $\alpha$ and $\beta$.
\begin{claim}\label{claim:translationmap}
\begin{enumerate}
\item In every quadrant, the contour of the best response payoff function is a line.
\item The slope of the contours of the best response payoff function in both $ R_{\one} $ and $ R_{\zer} $ is $ -\frac{q(1)}{q(0)}=-\frac{q(1|0)}{q(0|1)}$ and does not depend on the parameters $ \alpha, \beta $.
\item The best payoff of any point $ (x,y) $ can be decomposed as the sum of the best payoff of $ (x,q^*) $ and the best payoff of $ (q^*,y)$
\end{enumerate}
\end{claim}

\begin{proof}[Proof of Claim~\ref{claim:translationmap}]
Recall that we assume $ q(1|1) > q(1|0) $.  %, the analysis for $ q(1|1)<q(1|0) $ is similar.\\
%\gs{This table should be changed with the new notation of $R_s$.  so merge first two columns}

%\begin{tabular}{|c|l|l|}
%\hline Quadrant & Best response payoff when other players play (x,y)\\
%\hline $ R_{\one} $ & $  q(0)\ell(x,1 )+q(1)\ell(y,1 )=\alpha(q(0)(x-q^*)+q(1)(y-q^*)) $ \\
%\hline $ R_{\tru} $ & $  q(0)\ell(x,0 )+q(1)\ell(y,1 )=\beta(q(0)(x-q^*))+\alpha(q(1)(y-q^*)) $ \\
%\hline $ R_{\zer} $ & $  q(0)\ell(x,0)+q(1)\ell(y,0)=\beta(q(0)(x-q^*)+q(1)(y-q^*)) $ \\
%\hline $ R_{\fal} $ & $  q(0)\ell(x,1)+q(1)\ell(y,0)=\alpha(q(0)(x-q^*))+\beta(q(1)(y-q^*)) $ \\
%\hline
%\end{tabular}

\begin{table}[h] \label{tabel:best-response-by-quadrant}
\caption{illustrates the best response payoff for a $ (\alpha,\beta,q^*,\gamma) $-line set in each quadrant\label{qtable}}{
\resizebox{\textwidth}{!} {\begin{tabular}{|c|l|l|}
\hline Quadrant & Best response payoff when other players play (x,y)\\
\hline $ R_{\one} $ & $  q(0)\ell(x,1 )+q(1)\ell(y,1 )=\alpha(q(0)(x-q^*)+q(1)(y-q^*)) $ \\
\hline $ R_{\tru} $ & $  q(0)\ell(x,0 )+q(1)\ell(y,1 )=\beta(q(0)(x-q^*))+\alpha(q(1)(y-q^*)) $ \\
\hline $ R_{\zer} $ & $  q(0)\ell(x,0)+q(1)\ell(y,0)=\beta(q(0)(x-q^*)+q(1)(y-q^*)) $ \\
\hline $ R_{\fal} $ & $  q(0)\ell(x,1)+q(1)\ell(y,0)=\alpha(q(0)(x-q^*))+\beta(q(1)(y-q^*)) $ \\
\hline
\end{tabular}}
}
\end{table}

In Table \ref{qtable}, the best response for each quadrant is the best response of agent $ i $ when $ (\hat{q}_{-i}(1|0),\hat{q}_{-i}(1|1)) $ is in that quadrant. This comes from the fact that
\begin{align*}
BR(\hat{q}_{-i}(1|0),\hat{q}_{-i}(1|1))=q(0) \left( \max_{b_0\in{0,1}}\ell(\hat{q}_{-i}(1|0),b_0) \right)+q(1) \left( \max_{b_1\in{0,1}}\ell(\hat{q}_{-i}(1|1),b_1) \right).
\end{align*} and that the best response of a player $i$ who gets bit $b$ is $0$ if $\hat{q}_{-i}(1|b) < q^*$, and $1$ if $\hat{q}_{-i}(1|b)>q^*$. Part (1) of the claim follows immediately from Table~\ref{qtable}. We can also observe from Table~\ref{qtable} that when $ \alpha,\beta\neq 0 $, the slope of the contours in both $R_{\one}$ and $ R_{\zer} $ are $-\frac{q(0)}{q(1)}=-\frac{q(0|1)}{q(1|0)} $, so part (2) follows. Part (3) of the claim follows since in the best response plot, $ (x,y) $, $ (x,q^*) $ and $ (y,q^*) $ are in the same quadrant. If they are in the $ R_{\tru} $ region, according to Table \ref{qtable}, the best response payoff of $ (x,y) $ is
%\katrina{What does best payoff mean here?}\yk{sorry it should be best response payoff. I cite table 1 here.}
\begin{align*}
\beta q(0)(x-q^*)+&\alpha q(1) (y-q^*)\\
=&\beta q(0)(x-q^*)+\alpha q(1) (q^*-q^*)+\beta q(0)(q^*-q^*)+\alpha q(1) (y-q^*)
\end{align*}
which is a sum of the best payoff of $ (x,q^*) $ and $ (q^*,y) $. The proof is similar in other quadrants.
\end{proof}

\begin{definition}\label{Translation}
$ g:\mathbb{R}^2\mapsto \mathbb{R}^2 $ is a {\em translation map} if for any point $ (x,y) $, the mapping $ g((x,y)) $ is in $ R_{\tru} $ and if $(x,y)$ and $g((x,y))$ have the same payoff.
\end{definition}

Now we will construct a map $ f $ and prove that $ f $ is a translation map.
\begin{definition}\label{translationmap} Define $f$ as follows. For$(x, y)$ in $ R_{\tru} $, $ f((x,y)) :=(x,y) $. For $ (x,y) $ in $ R_{\one} $, let $ f((x,y)) $ be the intersection point of the contour line and $ x=q^* $. For $ (x,y) $ in $ R_{\zer} $,  let $ f((x,y)) $ be the intersection point of the contour line and $ y=q^* $. For $ (x,y) $ in $ R_{\fal} $, $ f((x,y)):=(f((q^*,y))_x,f((x,(q^*))_y) $ where $ f(((q^*,y))_x $ means the $ x $-coordinate of $ f(((q^*,y)) $,  and $ f((x,(q^*))_y $ means the $ y $-coordinate of $ f((x,(q^*)) $.
\end{definition}
\begin{claim}\label{fistm}
  $ f $ is a translation map.
\end{claim}
\begin{proof}
It is clear that $ f((x,y)) $ is in $ R_{\tru} $; now we need to prove that it preserves the best expected payoff. Based on Claim~\ref{claim:translationmap}, we know the slopes of the contours in both $ R_{\one} $ and $ R_{\zer} $ do not depend on $ \alpha $ and $ \beta $, so $ f $ does not depend on $ \alpha $ and $ \beta $. According to the definition of a contour, $ f $ preserves the best expected payoff of points in $ R_{\one} $ and $ R_{\zer} $. For points in $ R_{\fal} $, since we know the best payoff of $ (x,y) $ can be decomposed as the sum of the best payoff of $ (q^*,y) $ and $ (x,q^*) $ based on Claim~\ref{claim:translationmap} and $ f $ preserves the best expected payoff of $ (q^*,y) $ and $ (x,q^*) $(because they are in $ R_{\one} $ or $ R_{\zer} $), we can see $ f $ also preserves the best expected payoff of points in $ R_{\fal} $. \end{proof}

\yk{I cut here}
\ifnum\fullversion=1

%%

%\begin{proof}[ Proof of Theorem \ref{exist}:]

We will show the analysis for $ q(1|1)\geq q(0|0) $. the case when $ q(1|1)\leq q(0|0) $ is analogous.\\

%\begin{comment}
%\paragraph{Symmetric Equilibria in the Binary Case}
% \begin{lemma}\label{eq}
% In the contour plot, let $ p_1=\one $, $ p_2=\tru $, $ p_3=\zer $, $ p_4=\fal $ denote $ (1|1) $,  $ (q(1|0),q(1|1)) $, $ (0,0) $ and  $ (q(0|0),q(1|1)) $. In table \ref{qtable}, each quadrant is corresponding to a strategy which is the best response of agent $ i $ when $ (\hat{q}_{-i}(1|0),\hat{q}_{-i}(1|1)) $ is in that quadrant. Let $ Q_t $ denote the quadrant where $ p_t $ is the best reponse in that quadrant. A point $ p $ is an equilibrium iff $ \exists $ $ \lambda_t\in \mathbb{R}\geq 0 $ for $ t\in\{1,2,3,4\} $ such that both $ p=\sum_{t=1}^{4}\lambda_t p_t $ and $ \lambda_t>0 \Rightarrow p\in Q_t $.
% \end{lemma}
%The proof Proof of Lemma\ref{eq} is in section\yk{?}.\\
%\yk{Maybe move this to other section}\\

%\begin{lemma}
%For $ (\hat{q}_{-i}(1|0),\hat{q}_{-i}(1|1)) $ in $ R_{\one} $, the best response of player $i$ is to always report $1$. For $ (\hat{q}_{-i}(1|0),\hat{q}_{-i}(1|1)) $ in $ R_{\tru} $, it is to tell the truth. For $ (\hat{q}_{-i}(1|0),\hat{q}_{-i}(1|1)) $ in $ R_{\zer} $, the best response of player $i$ is to always report $0$. For $ (\hat{q}_{-i}(1|0),\hat{q}_{-i}(1|1)) $ in $ R_{\fal} $, the best response of player $i$ is to lie. The best response payoffs in every case are given in Table \ref{qtable}.
%\end{lemma}

%\begin{proof}
%This directly comes from the best response of a player $i$ who gets bit $b$ being $0$ if $\hat{q}_{-i}(1|b) < q^*$ and $1$ if $\hat{q}_{-i}(1|b)>q^*$.
%\end{proof}
%\end{comment}

Recall that in Section~\ref{sec:equilibrium}, we defined $ \tru\one $ as an equilibrium in which agents play truthfully when they receive bit $1$, and randomize when they receive bit $0$. This equilibrium is the convex combination of $ \tru $ and $ \one $ that has $\hat{q}(1|0)=q^*$, and is therefore at the boundary of quadrants $ R_{\one} $ and $ R_{\zer} $. We also defined $ \tru\zer $, $ \fal\one $ and $ \fal\zer $ in a similar fashion.

Having constructed the translation map, we can translate all equilibria to the same quadrant $ R_{\tru} $, to ease comparisons among them. If we wish for truth-telling to be focal, we need to tune the slope of the contours in $ R_{\tru} $, which will turn out to be a linear programming problem.

\begin{definition}\label{slope}
Notice that in the $ R_{\tru} $ region, the best response payoff is $$  q(0)\ell(x,0 )+q(1)\ell(y,1 )=\beta(q(0)(x-q^*))+\alpha(q(1)(y-q^*)), $$ so the slope of the contours in $ R_{\tru} $ is $$ -\frac{\beta q(0|1)}{\alpha q(1|0)}. $$ We denote this slope by $ k $.
\end{definition}

\yk{Is it a better place to introduce the slope?}

Before we tune the slope $ k $, we will first show, in the next lemma, that truth telling is an extreme point of the convex hull $ \mathcal{H} $ of all translated informative equilibria.

\begin{lemma}\label{main_Lemma}
Given positively correlated common prior Q, given $ q(1|0)<q^*<q(1|1) $,  consider the convex hull $ \mathcal{H} $ of a set of points $ \{f(\theta)\}_{\theta\in \alleq_{Q}(q^*)/\{\zer,\one\} } $ which are derived using the translation map to translate all informative equilibria to the $ R_{\tru} $ region.
\begin{description}
\item[Case 1 (when $ q(1|0)\neq q(0|1))$] In this case, we have three different cases based on the number of equilibria. \\

\begin{description}
\item[If \textbf{$ |\alleq_{Q}(q^*)|=9 $}:] $ f(\tru) $ is an extreme point, $ f(\tru) $ and $ f(\tru\one) $ share the same facet, and $ f(\tru) $ and $ f(\fal) $ share the same facet;
\smallskip
\item[If \textbf{$ |\alleq_{Q}(q^*)|=8 $}:]      $ f(\tru) $ is on the line segment whose endpoints are $ f(\tru\one) $ and $ f(\fal) $, so $ f(\tru) $ is not an extreme point of $ \mathcal{H} $
\smallskip

\item[If \textbf{$ |\alleq_{Q}(q^*)|=7 $}:]  $ f(\tru) $ is an extreme point, $ f(\tru) $ and $ f(\tru\one) $ share the same facet, and $ f(\tru) $ and $ f(\tru\zer) $ share the same facet;
\end{description}
\item[Case 2. (when $ q(1|0)= q(0|1) $)]
$ f(\fal)=f(\tru) $, and $ f(\tru) $ and $ f(\tru\one) $ share the same facet.

\end{description}
\end{lemma}
\smallskip
\smallskip

We prove Lemma~\ref{main_Lemma} below, but first, we discuss its consequences. Given Lemma~\ref{main_Lemma}, we will tune the slope $ k $ of contours in $ R_{\tru} $ to make truth-telling focal. According to Lemma~\ref{main_Lemma}, if this slope $ k $ is strictly between the lines which are incident on truth-telling and its adjacent extreme equilibrium, we can ensure that the truthful equilibrium pays at least well any other informative equilibrium. When $ \alleq_Q(q^*)=8 $, this is not possible since $ \tru $ is not an extreme point in this case; however, one can avoid this case by picking $ q^*\neq q(0|0) $.

The above results are stated in detail in the below corollary.

\begin{corollary}\label{the slope of contours}
Given positively correlated Q, we show how to make truth-telling focal among the informative equilibria. We state results for the case when $ q(1|1)\geq q(0|0) $; the case when $ q(1|1)\leq q(0|0) $ is analogous.
\begin{description}
\item[ Case 1 $ q(1|0)\neq q(0|1) $:] Truth-telling can be focal among the informative equilibria and will be so if $ \beta<0<\alpha $ and either of these two restrictions on $ q^*,k $ holds:\\
\begin{description}
\item[(1)] $ \max\{q(0|0),q(1|0)\}<q^*<q(1|1) $, $$  \frac{f(\tru\one)_y-f(\tru)_y}{f(\tru\one)_x-f(\tru)_x}<k<\frac{f(\tru\zer)_y-f(\tru)_y}{f(\tru\zer)_x-f(\tru)_x} $$
\item[(2)] $ q(1|0)<q^*<q(0|0) $, $$  \frac{f(\tru\one)_y-f(\tru)_y}{f(\tru\one)_x-f(\tru)_x}<k<\frac{f(\fal)_y-f(\tru)_y}{f(\fal)_x-f(\tru)_x} $$
\end{description}
\item[ Case 2 $ q(1|0)= q(0|1) $:] There is no way to make truth-telling focal among the informative equilibria, since $ f(\fal)=f(\tru) $. However, truth-telling can have  (non-strictly) greatest payoff among them if $ \beta<0<\alpha $, $ q(0|1)=q(1|0)<q^*<q(1|1)=q(0|0) $, $$  k=\frac{f(\tru\one)_y-f(\tru)_y}{f(\tru\one)_x-f(\tru)_x}$$
\end{description}
\end{corollary}

%\katrina{Which lemma should the following refer to?}\yk{It is the proof for the main theorem}
Once we have Corollary~\ref{the slope of contours}, the results of Theorem~\ref{exist} follow:

\begin{proof}[Proof of Theorem~\ref{exist}]
Based on Corollary~\ref{the slope of contours},  given positively correlated prior $ Q $, we can select $ q^* $ and  $r=-\frac{\alpha}{\beta}>0$ to make truth-telling focal if $ q(1|0)\neq q(0|1) $.
\end{proof}

\begin{proof}[Proof of Lemma~\ref{main_Lemma}]
\begin{figure}
\centering
\includegraphics[scale=0.19]{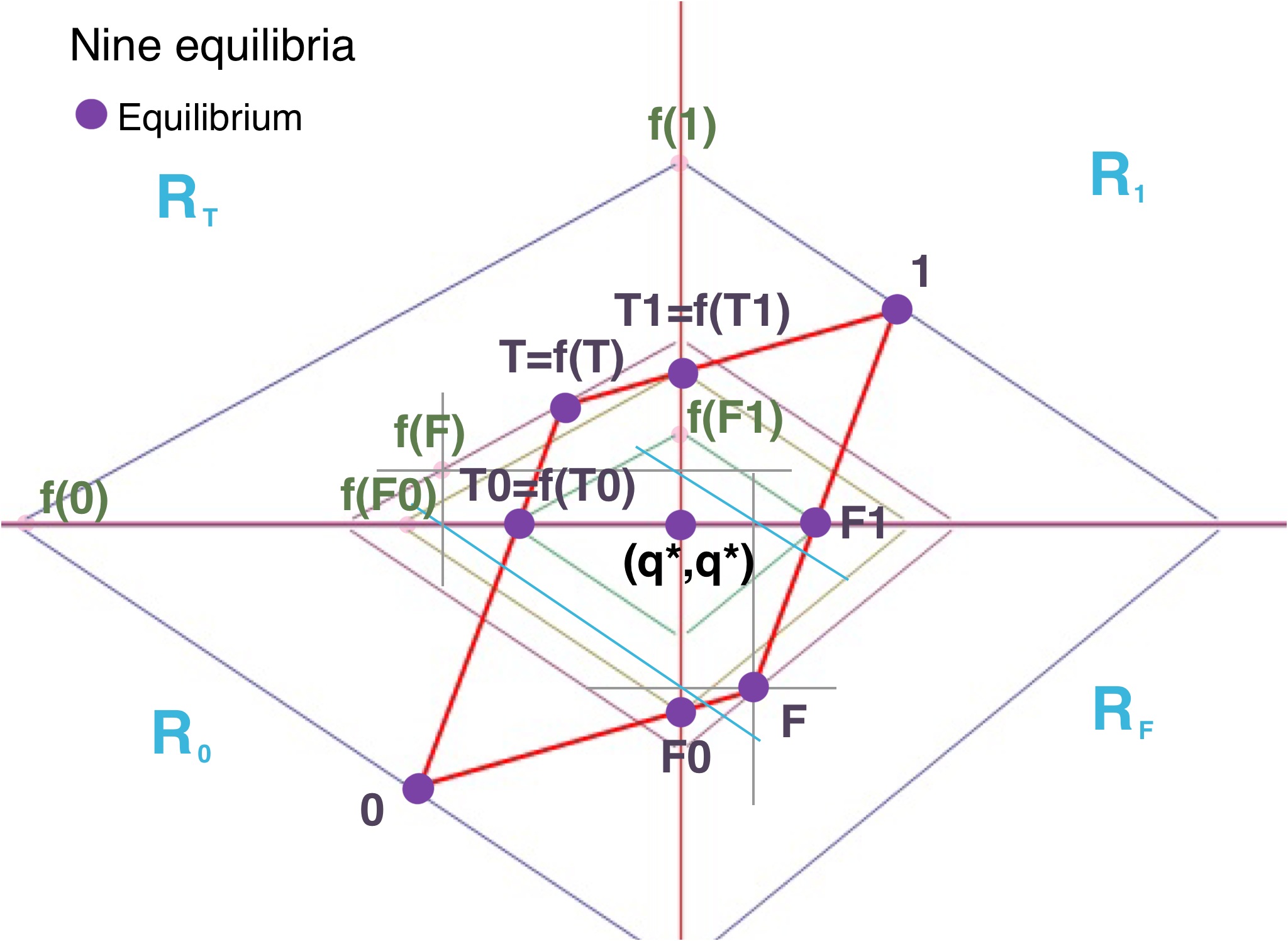}
\caption{Best response plot of a prior and given $ q^* $ with nine equilibria. The purple points with black words are equilibria and the pink points with green words are translated equilibria. Notice that we can pair up equilibria such that
the lines going through each pair of equilibria are parallel. $ \tru $ is an extreme point of all translated informative equilibria and $ \tru\one $ and $ f(\fal) $ are $ \tru $'s two adjacent extreme equilibria. }
\label{fig:nine}
\end{figure}

%\paragraph{Case: $ q(1|1)=q(1|0) $\\}
%The payoffs of all informative symmetric equilibria are the same. \yk{Maybe we have already discussed the case previously}
Given positively correlated Q and $ q(1|0)<q^*<q(1|1) $, we discuss the problem under the following possible cases:
$\begin{cases}
 q^*\leq q(0|0)\Rightarrow \alleq_Q(q^*)=8\ or\ 9\\
 q^* > q(0|0)\Rightarrow \alleq_Q(q^*)=7\\
\end{cases}$
\paragraph{Case 1: when there exist eight or nine equilibria ($q^*\leq q(0|0)$)\\}
There are nine equilibria iff the pure strategies are in different quadrants. Since we have $ q(1|1) > q(1|0) $, there are nine equilibria iff $ q(0|1)\leq q(1|0)<q^*< q(0|0)\leq q(1|1) $.
\yk{I cut coordinates table here}
\ifnum\fullversion=1
Table \ref{9etable} contains the coordinates of all nine equilibria.\\

% \begin{tabular}{|c|l|l|}
% \hline Equilibrium & $(x,y)$ & $f((x,y))$ \\
% \hline $ \one $ &$ (1,1) $  & $ (q^*,\frac{1-q^* q(0)}{q(1)}) $ \\
% \hline $ \zer $ & $ (0,0) $ & $ (-q^* \frac{q(1)}{q(0)},q^*) $ \\
% \hline $ \tru $ &  $ (q(1|0),q(1|1)) $ &  $ (q(1|0),q(1|1)) $ \\
% \hline $\fal$ & $ (q(0|0),q(0|1)) $ & $ (q(1|0)-\frac{q(1|0)q^*}{q(0|1)}+q^*,\frac{q(0|1)(q(0|0)-q^*)}{q(1|0)}+q^*) $ \\
% \hline $ \tru\one $ & $ (q^*, \frac{(q(1|1)-q(1|0))(1-q^*)}{q(0|0)}+q^*) $ & $ (q^*, \frac{(q(1|1)-q(1|0))(1-q^*)}{q(0|0)}+q^*) $ \\
% \hline $\fal\zer $ &  $ (q^*,\frac{(q(1|1)-q(1|0))q^*}{q(1|0)-1}+q^* ) $ & $ (-\frac{(q(1|1)-q(1|0))q(1|0)q^*}{q(0|0)q(0|1)}+q^*,q^*) $ \\
% \hline $ \tru\zer $ & $ (\frac{(q(1|0))q^*}{q(1|1)},q^* )$ & $ (\frac{(q(1|0))q^*}{q(1|1)}, q^* )$ \\
% \hline $\fal\one$ & $ (\frac{(q(1|1)-q(1|0))(1-q^*)}{q(1|1)}+q^*,q^*) $ & $ (q^*,\frac{(q(1|1)-q(1|0))(1-q^*)q(0|1)}{q(1|1)q(1|0)}+q^*) $ \\
% \hline $ \qstar $ & $ (q^*,q^*) $ & $ (q^*,q^*) $ \\
% \hline
% \end{tabular}

 \begin{table}[h!] \label{9etable}
 \caption{Coordinates of all equilibria and all translated equilibria when there are nine equilibria\label{9etable}}{
 \resizebox{\textwidth}{!} {\begin{tabular}{|c|l|l|}
 \hline Equilibrium & $(x,y)$ & $f((x,y))$ \\
 \hline $ \one $ &$ (1,1) $  & $ (q^*,\frac{1-q^* q(0)}{q(1)}) $ \\
 \hline $ \zer $ & $ (0,0) $ & $ (-q^* \frac{q(1)}{q(0)},q^*) $ \\
 \hline $ \tru $ &  $ (q(1|0),q(1|1)) $ &  $ (q(1|0),q(1|1)) $ \\
 \hline $\fal$ & $ (q(0|0),q(0|1)) $ & $ (q(1|0)-\frac{q(1|0)q^*}{q(0|1)}+q^*,\frac{q(0|1)(q(0|0)-q^*)}{q(1|0)}+q^*) $ \\
 \hline $ \tru\one $ & $ (q^*, \frac{(q(1|1)-q(1|0))(1-q^*)}{q(0|0)}+q^*) $ & $ (q^*, \frac{(q(1|1)-q(1|0))(1-q^*)}{q(0|0)}+q^*) $ \\
 \hline $\fal\zer $ &  $ (q^*,\frac{(q(1|1)-q(1|0))q^*}{q(1|0)-1}+q^* ) $ & $ (-\frac{(q(1|1)-q(1|0))q(1|0)q^*}{q(0|0)q(0|1)}+q^*,q^*) $ \\
 \hline $ \tru\zer $ & $ (\frac{(q(1|0))q^*}{q(1|1)},q^* )$ & $ (\frac{(q(1|0))q^*}{q(1|1)}, q^* )$ \\
 \hline $\fal\one$ & $ (\frac{(q(1|1)-q(1|0))(1-q^*)}{q(1|1)}+q^*,q^*) $ & $ (q^*,\frac{(q(1|1)-q(1|0))(1-q^*)q(0|1)}{q(1|1)q(1|0)}+q^*) $ \\
 \hline $ \qstar $ & $ (q^*,q^*) $ & $ (q^*,q^*) $ \\
 \hline
 \end{tabular}}
 }
\end{table}

 %\smallskip
 \fi
 \yk{end here}
The four claims below show that truth-telling is an extreme point of the convex hull of all translated informative equilibria.
%\katrina{Should each of the claims below contain the assumption that ($q^*\leq q(0|0)$)?}\yk{I think so since it is under case 1, I add this assumption to each claim. }
\begin{claim}\label{parallel}
When $q^*\leq q(0|0)$, the lines that are incident to the following four pairs of points $\{\zer,\one\},$ $\{\tru,\fal\},$ $\{\tru\one,\fal\zer\},$ $\{\tru\zer,\fal\one\}  $ are all parallel to each other. Moreover the slope of those parallels is always greater than 0.
\end{claim}
\begin{proof}
The claim follows immediately from computing the slope of each line, which are all equal to $$ \frac{q(0|1) (1 - q^*)}{q(1|0) q^*} .$$ This slope is  always greater than 0.
\end{proof}

\begin{claim}\label{order}
When $q^*\leq q(0|0)$, going from $ \qstar $ toward $ (q^*,+\infty) $, the order of the four parallels is as follows: $\{\tru\zer,\fal\one\},\{\tru\one,\fal\zer\},\{\tru,\fal\}, \{\zer,\one\}$.
\end{claim}

\begin{proof}
To prove the claim, we compute the values of the intersection points of the following four pairs $$\{\tru\zer,\fal\one\},\{\tru\one,\fal\zer\},\{\tru,\fal\}, \{\zer,\one\}$$ with the $ y $-axis and compare them.\\
The $ y $-coordinates of the intersection points are given by:\\
\begin{align*}
 \{\tru\zer,\fal\one\}:&\frac{(q(1|1)-q(1|0))(1-q^*)q(0|1)}{q(1|1)q(1|0)}+q^*,\\
 \{\tru\one,\fal\zer\}:  &\frac{(q(1|1)-q(1|0))(1-q^*)}{q(0|0)}+q^*,\\
 \{\tru,\fal\}: &\frac{(q(0|1)+q(1|0))(1-q^*)q^*-q(1|0)q(0|1)}{q(1|0)q^*}+q^*,\\
 \{\zer,\one\}:
 &\frac{(1-q^*)}{q(1)}+q^*\\
\end{align*}
Then we compare them:

\begin{align*}
\{\tru,\fal\}- \{\tru\one,\fal\zer\}=
&\frac{(q(0|1)+q(1|0))(1-q^*)q^*-q(1|0)q(0|1)}{q(1|0)q^*}\\
-&\frac{(q(1|1)-q(1|0))(1-q^*)}{q(0|0)}\\
=&\frac{q(0|0)[(q(0|1)+q(1|0))(1-q^*)q^*-q(1|0)q(0|1)]- q(1|0)(q(1|1)-q(1|0))(1-q^*)q^*}{q(0|0)q(1|0)q^*}\\
=&\frac{q(0|1)(1-q^*)q^*-q(0|0)q(1|0)q(0|1)}{q(0|0)q(1|0)q^*}\\
=&\frac{q(0|1) (q^*-q(1|0)) (q(0|0) - q^*)}{q(0|0) q(1|0) q^* }>0
\end{align*}
according to the fact that $$ q(0|1)\leq q(1|0)<q^*<q(0|0)\leq q(1|1) .$$\\

We see that
\begin{align*}
 \{\tru\one,\fal\zer\}- \{\tru\zer,\fal\one\}=
&\frac{(q(1|1)-q(1|0))(1-q^*)}{q(0|0)}\\
-&\frac{(q(1|1)-q(1|0))(1-q^*)q(0|1)}{q(1|1)q(1|0)}\\
=&(q(1|1)-q(1|0))(1-q^*)q(0|1)(\frac{1}{q(0|0)}-\frac{q(0|1)}{q(1|1)q(1|0)})\geq 0
\end{align*}
according to the fact that $$ q(1|1)q(1|0)-q(0|0)q(0|1)=(q(1|0)-q(0|1))(q(1|1)-q(1|0))\geq 0.$$

Finally,
\begin{align*}
 \{\zer,\one\}- \{\tru,\fal\}=
&\frac{(1-q^*)}{q(1)}-\frac{(q(0|1)+q(1|0))(1-q^*)q^*-q(1|0)q(0|1)}{q(1|0)q^*}\\
=&\frac{(q(1|0)+q(0|1))(1-q^*)q^*-[(q(0|1)+q(1|0))(1-q^*)q^*-q(1|0)q(0|1)]}{q(1|0)q^*}  \\
=&\frac{q(1|0)q(0|1)}{q(1|0)q^*}\\
=&\frac{q(0|1)}{q^*}>0.\qedhere
\end{align*}

\end{proof}

\begin{claim}\label{truthlie}
When $q^*\leq q(0|0)$, $ f(\fal)_x < \tru_x $, $ f(\fal)_y < \tru_y $ if $ q(1|0)\neq q(0|1) $ and if $ q(1|0)= q(0|1) $, $ f(\fal)=\tru $
\end{claim}

\begin{proof}
 $ f(\fal)_x \leq \tru_x $ follows immediately from the x-coordinates of $ \tru $ and $ f(\fal) $. Since the slope of the line incident to $ \tru $ and $ \fal $ is always greater than 0, $ f(\fal)_y \leq \tru_y $.
\end{proof}

\begin{claim}\label{eightequilibrium}
If $ q^*=q(0|0) $ and $ q(1|0)\neq q(0|1) $, then
$$ \frac{f(\fal)_y-\tru_y}{f(\fal)_x-\tru_x}=\frac{\tru\one_y-\tru_y}{\tru\one_x-\tru_x} $$
\end{claim}
%
%\paragraph{Intuition of the claim ~\ref{parallel}:}
%If we set $ PS(0,0)=PS(1,1) $, we find that for other three pairs, any two points in the same pair will get the same payoff, which gives us intuition that the four lines defined above are all parallel to each other.\\
%\vspace{1pt}\\

\begin{proof}
For Claim~\ref{eightequilibrium}, if $ q^*=q(0|0) $, $ \fal=\fal\zer $, then $ \tru $, $ f(\fal) $, $ f(\fal\zer) $ and $ \tru\one $ will lie in one line since
$$ \{\tru,\fal\}- \{\tru\one,\fal\zer\}=\frac{q(0|1) (q^*-q(1|0)) (q(0|0) - q^*)}{q(0|0) q(1|0) q^* }=0.$$
Then the payoff $ \tru $ is highest among all informative equilibria iff the slope of the contour line is
$$ \frac{q(0|1) (1 - q^*)}{q(1|0) q^*},$$
which is the slope of the parallels in Claim~\ref{parallel}, or one of $ \tru\zer $ and $ \tru\one $ will be greater than the payoff of $ \tru $. But the payoff of $ \tru $ must be equal to the payoff of $ \tru\zer $, which means the payoff of $ \tru $ cannot be strictly greater than any other informative equilibria.
\end{proof}

%\yk{I cut here}
%\ifnum\fullversion=1
%
%
%
%%But we can pick $q^*>q(0|0)$ or $q^*<q(0|0)$ to avoid this case.\jz{I believe this needs to be explained too}
%\else
%We show these claims by calculations(See supplemental materials).
%\fi
%
%\yk{end here}

%Given the above claims, we can see that $ \tru $ is an extreme point of all translated informative equilibria and $ \tru\one $ and $ f(\fal) $ are $ \tru $'s two adjacent extreme equilibria.

%Here, since we want the payoff of $ \tru $ to exceed the payoff of $ \qstar $, we see that $ \alpha $ must be greater than 0 . This means that the further a strategy is from $ \qstar $, the higher its payoff.

Given the above claims, $ \tru $ is an extreme point of the convex polygon of all translated informative equilibria. $ \tru\one $ and $ f(\fal) $ are $ \tru $'s two adjacent extreme points. We also know that the best response payoff is a linear function with parameters $\alpha,\beta$ in region $R_{\tru}$, according to Table \ref{qtable}. If we think of the best response payoff as an objective function in a linear program and of all translated informative equilibria as feasible points, we can see it is possible to properly choose $\alpha,\beta$ such that the best response payoff of $\tru$ is the global optimum among the best response payoffs of all informative equilibria, since $\tru$ is an extreme point. To make the best response payoff of $ \tru $ exceed the best response payoff of $ \qstar $, we can see that $ \alpha $ must be greater than 0.

%\katrina{I don't understand.}\yk{I rewrite the above paragraph, and comment the original paragraph.}

\begin{enumerate}
\item If $ q(1|0)=q(0|1) $, the payoff of $ \fal $ will always be the same as the payoff of $ \tru $ since $ f(\fal)=\tru $ when $ q(1|0)=q(0|1) $. If the slope of the contour line is $$ \frac{\tru_y-\tru\one_y}{\tru_x-\tru\one_x},$$ then the payoff of $ \tru $($ \fal $) is greater than the payoffs of other informative equilibria since the contour line incident to $ \tru $ will be further from $ \qstar $ than contours incident to other informative  equilibria.\\

\item If $ q(1|0)\neq q(0|1) $, if we set $\alpha,\beta$ such that the payoff of $ \tru\one $ is equal to the payoff of $ \fal $, then the slope of the contour will be strictly greater than the slope of the line incident to $ \tru $ and $ \tru\one $ and strictly less than the slope of the line incident to $ \tru $ and $ f(\fal) $. The payoff of $ \tru $ will be strictly greater than the payoffs of other informative equilibria since $ \tru $ is an extreme point and the contour line incident to $ \tru $ is not incident to other points, based on Claim~\ref{order}.
\end{enumerate}

This completes Case 1.

\begin{figure}
\centering
\includegraphics[scale=0.19]{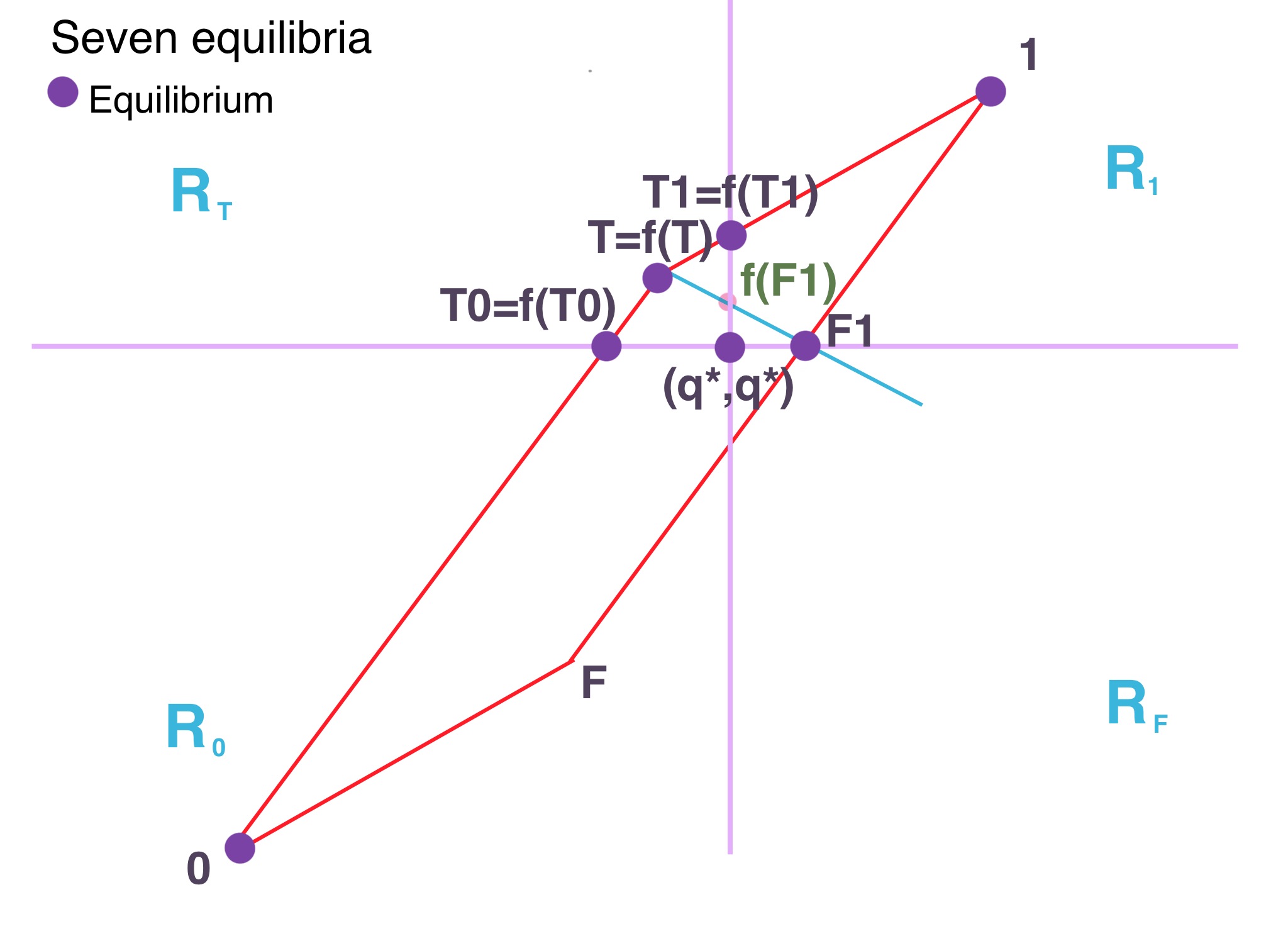}
\caption{Best response plot of a prior and given $ q^* $ with seven equilibria. Notice that $ \tru $ is an extreme point of all translated informative equilibria and $ \tru\one $ and $ \tru\zer $ are $ \tru $'s two adjacent extreme equilibria.}
\label{fig:seven}
\end{figure}

\paragraph{Case 2: There exist seven equilibria\\}
There are fewer than eight equilibria iff pure strategies are not in different quadrants. $ \tru $ must be in $ R_{\tru} $ since $ q(1|0)<q^*<q(1|1) $. $ \one $ and $ \zer $ must be in $ R_{\one} $ and $ R_{\zer} $ respectively. So there are less than eight equilibria iff $ \fal $ is in $ R_{\zer} $ and not in $ R_{\fal} $, which happens when $q^*> q(0|0)$. %If  $ \fal $ is in $ R_{\zer} $ and not in $ R_{\fal} $($ q^*> q(0|0) $), we have seven equilibria.\\
In this case, neither $ \fal $ nor $ \fal\zer $ is an equilibrium,
$$ f(\fal\one)= (q^*,\frac{(q(1|1)-q(1|0))(1-q^*)q(0|1)}{q(1|1)q(1|0)}+q^*), $$
and
$$  f(\fal\one)_y\leq \tru\one_y. $$
Now we can see $ \tru $ is an extreme point of all translated informative equilibria and $ \tru\one $ and $ \tru\zer $ are $ \tru $'s two adjacent extreme equilibria.
If we set the payoff of $ \tru\one $ equal to the payoff of $ \tru\zer $, then the slope of the contour will be strictly greater than the slope of the line incident to $ \tru $ and $ \tru\one $ and strictly less than the slope of the line incident to $ \tru $ and $ \tru\one $. The payoff of $ \tru $ will be strictly greater than the payoffs of other informative  equilibria since $ \tru $ is an extreme point and the contour line incident to $ \tru $ is not incident to other points, since $  f(\fal\one)_y\leq \tru\one_y $.
\end{proof}

\else
Having constructed the translation map, we can translate all equilibria to the same quadrant $ R_{\tru} $. We show that truth-telling is an extreme point of the convex hull $ \mathcal{H} $ of all translated informative equilibria and tune the slope of the contours in $ R_{\tru} $ to make truth-telling have larger expected payoff than any other informative equilibrium. (See the full version of the proof in the supplemental materials).
\fi
\yk{end here}

\section{Optimizing the Gap} \label{sec:optimization}

\ifnum\fullversion=1

In this section, we continue to focus only on the informative equilibria. We seek to design a payoff function, derived from a proper scoring rule, that maximizes the advantage (gap) of the payoff of the truth-telling equilibria over the payoffs of any other informative equilibrium.

Even though we classify all equilibria in Theorem~\ref{eqbinary}, it is still challenging to obtain the optimization result in Theorem~\ref{thm:focal_main}. Firstly, we must determine what kind of proper scoring rules make truth-telling have an advantage over all informative equilibria. Secondly, we must determine,  when truth-telling is better than other informative equilibria, which informative equilibrium is the second best equilibrium. Thirdly, we  must determine which proper scoring rule maximizes the advantage truth-telling has over the informative equilibrium which have the second largest expected payoff, over all Peer-prediction mechanisms with payoffs in $[0, 1]$.

Without the help of our best response payoff plots, the above steps are extremely complicated even in binary setting. However, with the help of our best response payoff plots, the above steps become approachable. We will show an important observation first: \textit{we can tune proper scoring rules by tuning the break-even point $q^*$ and the slope of contours.} We divide our proof into two steps:
\vspace{-5pt}
\begin{enumerate}
\item Fix the break-even point $q^*$ and tune the slope of contours $k$ in region $R_{\tru}$ to $k^{sup}(q^*)$ such that the advantage of truth-telling is optimized, over all Peer-prediction mechanisms with payoffs in $[0, 1]$.

\item With $k=k^{sup}(q^*)$, tune $q^*$ such that the advantage of truth-telling is optimized globally.
\end{enumerate}

\subsection{Attainable Priors}
In this section, we introduce the concept of an attainable prior, and use it to divide the space of possible priors into three regions. For the first two regions, we show that we can obtain mechanisms with optimal gap and for the third region, we show that there is no mechanism that has optimal gap but we can obtain mechanisms with gap arbitrarily close to the optimal value, which appears as Theorem~\ref{thm:optimization}.

Recall that $ PF $ (see Definition~\ref{PF}) is the payoff matrix of a line-set (which corresponds to a proper scoring rule). Now we consider a set of payoff functions where the maximal payoff is $ 1 $ and the minimal payoff is $ 0 $.

%For any two real numbers $ l<u $, $ PF\in[l,u] $ means that $ \min_{i,j=0,1}PF(i,j)=l,\max_{i,j=0,1}PF(i,j)=u $
\begin{definition} \label{def:u}
Define $\mathcal{U}$ to be the set of payoff functions with payoffs between $0$ and $1$. That is,
let $ \mathcal{U}=\{PF|\min_{i,j=0,1}PF(i,j)=0,\max_{i,j=0,1}PF(i,j)=1\}.$
For positively correlated Q, let $ \alleq_{Q}(PF) $ be the set of equilibria under payoff function $ PF $. Let $$ \Delta_Q(PF)=[\nu^{PF}(\tru)- \max_{\theta\in \alleq_{Q}(PF)/\{\tru,\zer,\one\}} \nu^{PF}(\theta)] $$ be the gap between truth-telling's payoff and the maximal payoff of all other informative equilibria.
\end{definition}
For any given symmetric and positively correlated $ Q $, we would like to
\begin{enumerate}
\item Solve $$ \Delta^*(Q)=\sup_{PF\in\mathcal{U}}\Delta_{Q}(PF)  $$
\item Either find $ PF^*\in \mathcal{U} $ such that $ \Delta_Q(PF^*)=\Delta^*(Q) $ or; if such $ PF^* $ doesn't exist, then $ \forall \epsilon>0 $, find $ PF^*(\epsilon) $ such that $ \Delta_Q(PF^*(\epsilon))>\Delta^*(Q)-\epsilon $.
\end{enumerate}

\begin{definition}[Attainable Prior]
We call a prior $ Q $ {\em attainable} if $ \exists PF^* $ such that $ \Delta_Q(PF^*)=\Delta^*(Q) $; otherwise, we call the prior {\em unattainable}.
\end{definition}

We divide the positively correlated priors into three regions which correspond to three payoff functions which (almost) attain the optimal gap $ \Delta^*(Q) $. These three regions are shown in Figure~\ref{fig:attainableprior}.
\begin{figure}
\centering
%\floatbox[{\capbeside\thisfloatsetup{capbesideposition={right,top},capbesidewidth=4cm}}]{figure}[\FBwidth]
{\includegraphics[scale=0.6]{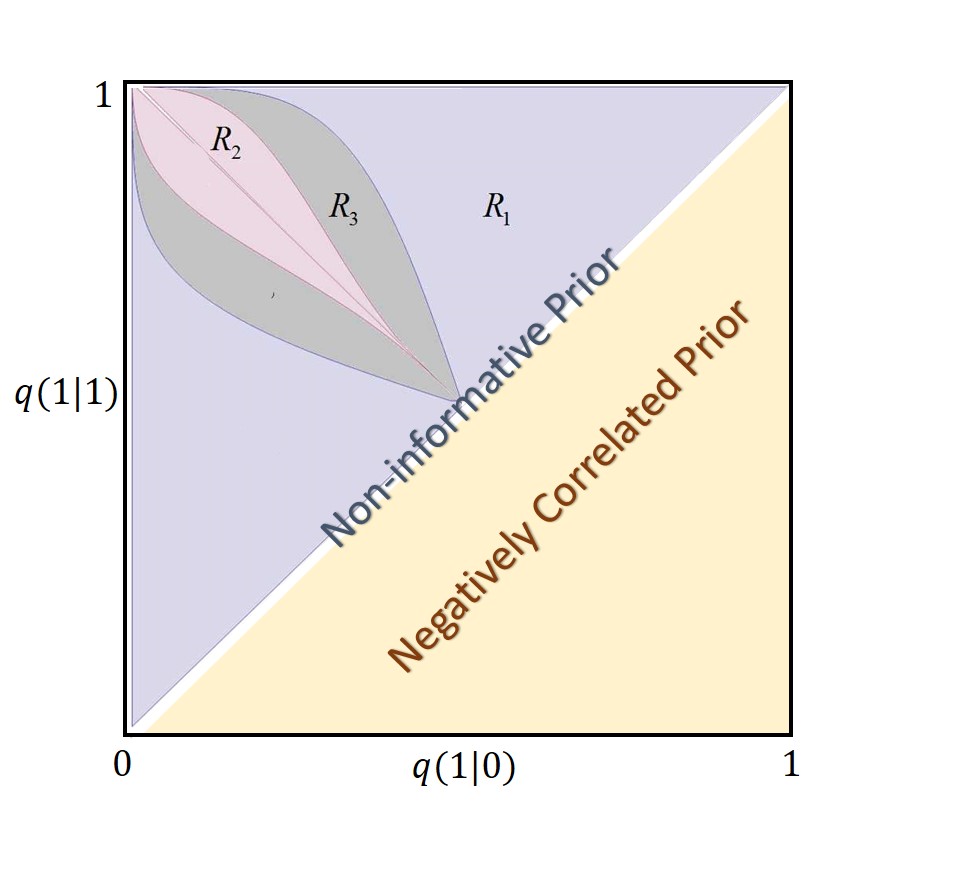}}
{\caption{The regions $R_1,R_2,R_3$ are good ``priors'' that we can make truth-telling focal. The priors in $R_1$ and $R_2$ are attainable while the priors in $R_3$ are unattainable. The white diagonals are ``bad'' priors we cannot make truth-telling focal. In the top-right to bottom-left diagonal, $q(1|0)=q(1|1)$, so the private signal does not have any information. We call this diagonal the set of non-informative priors. In the top-left to bottom-right diagonal, $q(0|0)=q(1|1)$ (actually we can see $q(0|0)=q(1|1)$ iff $q(0)=q(1)$ via some calculations). This diagonal is the set of signal symmetric priors. The yellow region is the set of the negative correlated priors. }
\label{fig:attainableprior}}
\end{figure}
\ifnum\fullversion=1
\begin{definition}\label{attainable prior}

\begin{align*}
R_1=&\{q(0|0)\leq q(1|0)\}\cup\\
&\{q(0|0)>q(1|0) \wedge q^*_{a,o}>q(0|0)\wedge \Delta^*_a(Q)\geq \Delta^*_b(Q) \}\\
R_2=& \{q(0|0)>q(1|0) \wedge q^*_{a,o}>q(0|0)\wedge \Delta^*_a(Q)< \Delta^*_b(Q) \}\cup\\
&\{ q(0|0)>q(1|0) \wedge q^*_{a,o}\leq q(0|0)\wedge \xi(k^{sup}_a(q(0|0)),q(0|0))\leq \Delta^*_b(Q) \}\\
R_3 =& \{ q(0|0)>q(1|0) \wedge q^*_{a,o}\leq q(0|0)\wedge \xi(k^{sup}_a(q(0|0)),q(0|0))> \Delta^*_b(Q) \}\\
\end{align*}
%$$ R_1=\{q(0|0)\leq q(1|0)\}\cup \{q(0|0)>q(1|0) \& q_{a,o}>q(0|0)\& \Delta^*_a(Q)\geq \Delta^*_b(Q) \} $$
%$$ R_2 = \{q(0|0)>q(1|0) \& q_{a,o}>q(0|0)\& \Delta^*_a(Q)< \Delta^*_b(Q) \}\cup\{ q(0|0)>q(1|0) \& q_{a,o}\leq q(0|0)\& \iota(q(0|0))\leq \Delta^*_b(Q) \} $$
%$$ R_3 = \{ q(0|0)>q(1|0) \& q_{a,o}\leq q(0|0)\& \iota_a(q(0|0))> \Delta^*_b(Q) \} $$
where
$ q^*_{a,o}=\frac{q(1|0)q(1|1)-\sqrt{q(1|0)q(0|1)q(0|0)q(1|1)}}{q(1|1)-q(0|0)} $,\\
$ \Delta^*_a(Q)=\frac{q(0|1)(q(0|0)q(1|0)-\sqrt{q(0|0)q(1|0)q(0|1)q(1|1)})(\sqrt{q(0|0)q(1|0)q(0|1)q(1|1)}-q(0|1)q(1|1))}{(q(0|1)+q(1|0)q(1|1)(q(1|1)-q(0|0))(\sqrt{q(0|0)q(1|0)q(0|1)q(1|1)}+q(1|1)-q(0|0)-q(1|0)q(1|1))} $,\\
$ \Delta^*_b(Q)=\frac{(q(1|1)-q(0|0))(q(0|0)q(1|0)(q(1|1)-q(0|1))-\sqrt{q(0|0)q(1|0)(q(1|1)-q(1|0))(q(1|1)-q(0|0))})}{(q(1|0)+q(0|1))q(1|0)q(1|1)q(0|0)} $, and \\
$ \xi(k^{sup}_a(q(0|0)),q(0|0))=\min\{\iota_a(q(0|0)),\kappa_a(q(0|0))\} $,\\
where $ \iota_a(q(0|0))=\frac{q(0|1)}{q(1|0)+q(0|1)} \frac{q(1|0)(q(1|1)-q(0|0))( q(0|0)-q(1|0))}{ q(0|0)((q(1|1)-q(0|0))q(0|0)-q(1|0)q(1|1)) } $,\\
$ \kappa_a(q(0|0))=\frac{q(0|1)}{q(1|0)+q(0|1)}\frac{(q(1|1)-q(0|0))(q(0|0)-q(1|0))}{q(1|1) (1-q(0|0))} $.
\end{definition}
Note that $ R_1 $ and $ R_2 $ contain \textit{attainable priors}, while $ R_3 $ contain \textit{unattainable priors}.
\smallskip
\else
$ R_1 $ and $ R_2 $ contain \textit{attainable priors} while $ R_3 $ contain \textit{unattainable priors}.
We define $R_1$, $R_2$, and $R_3$ in the full version.

\fi

%The first step of the proof is codified in the statement of the following optimization theorem:

\subsection{Optimization Theorem: Statement and Proof Structure}
\begin{theorem}\label{opt_thm}\label{thm:optimization}
For any  binary, symmetric, positively correlated, and signal asymmetric prior $Q$, with $ q(1|1)>q(0|0) $ (the $ q(0|0)<q(1|1) $ case is analogous),
\begin{description}
\item[1]If $ Q\in R_1 $, then $\Delta_{Q}(\mathcal{M}_1(Q))=\Delta^*(Q)$\\

\item[2]If $ Q\in R_2 $, then $\Delta_{Q}(\mathcal{M}_2(Q))=\Delta^*(Q)$\\

\item[3]If $ Q\in R_3 $, then $ \lim_{\epsilon\rightarrow 0^{+}}\Delta_{Q}(\mathcal{M}_3(Q,\epsilon))=\Delta^*(Q) $\\

\end{description}
where

$\mathcal{M}_1(Q) = \left( \begin{array}{ccc}
\zeta(Q) & & 0\\
 & & \\
0 & & 1\\
\end{array} \right)$, $\mathcal{M}_2(Q) =\left( \begin{array}{ccc}
1 & & 0\\
 & & \\
0 & & \eta(Q)\\
\end{array} \right) $, $\mathcal{M}_3(Q,\epsilon) =\left( \begin{array}{ccc}
\zeta(Q,\epsilon) & & \delta(Q,\epsilon)\\
 & & \\
0 & & 1\\
\end{array} \right) $ and\\
\ifnum\fullversion=1

$  \zeta(Q)=\sqrt{\frac{q(0|0) q(0|1)}{q(1|0) q(1|1)}}   $, $  \eta(Q)=\frac{1}{q(1|1)}(\sqrt{\frac{(q(1|1)-q(1|0))(q(1|1)-q(0|0))}{q(0|0)q(1|0)}}-q(0|1))   $,\\
$\zeta(Q,\epsilon)=\frac{q(0|0) q(0|1)}{q(0|0) q(0|1) q(0|0)+\epsilon+q(1|0) (q(1|1)-q(1|1)
   q(0|0)+\epsilon)}$, and\\
$\delta(Q,\epsilon)=\frac{q(1|0) q(1|1) (q(0|0)+\epsilon-1)^2-q(0|0) q(0|1) q(0|0)+\epsilon^2}{q(0|0)+\epsilon (q(1|0)
      q(1|1) (q(0|0)+\epsilon-1)-q(0|0) q(0|1) q(0|0)+\epsilon)}$.

\else
We define $\zeta(Q)$,  $\eta(Q)$, $\zeta(Q), \epsilon)$, and $\delta(Q,\epsilon)$ in the full version.
\fi

\end{theorem}

We note that the actual form of the optimal payoff functions is quite simple, especially in $R_1$ and $R_2$, where they only have one parameter to tune.

\begin{definition}\label{def:normalization}
We define $ \mathcal{N}(\cdot) $ as a normalization function and $$ \mathcal{N}(PF)=\frac{PF-\min_{i,j}PF(i,j)}{\max_{i,j}PF(i,j)-\min_{i,j}PF(i,j)}. $$
\end{definition}

Recall from Definition~\ref{def:u} that $\mathcal{U}$ is the payoff functions with payments between 0 and 1.
\begin{lemma}\label{remark:normalization}
$$ \sup_{PF\in\mathcal{U}}\Delta_{Q}(PF)=\sup_{PF}\Delta_{Q}(\mathcal{N}(PF))$$
\end{lemma}
\begin{proof}
For all $PF$,  $ \mathcal{N}(PF)\in \mathcal{U}$; so $\sup_{PF\in\mathcal{U}}\Delta_{Q}(PF)\geq \sup_{PF}\Delta_{Q}(\mathcal{N}(PF))$. For all $PF\in \mathcal{U}$, we have  $\Delta_{Q}(\mathcal{N}(PF)) = \Delta_{Q}(PF)$, so $\sup_{PF\in\mathcal{U}}\Delta_{Q}(PF)\leq \sup_{PF}\Delta_{Q}(\mathcal{N}(PF))$. Then the result follows.
\end{proof}

This allows us to translate a constrained optimization problem into an unconstrained optimization problem.

\begin{definition} \label{def:G}
We define $ G(\cdot) $ as a function mapping $ \{\alpha,\beta,q^*,\gamma\} $ to a payoff function $ PF $ such that $ \forall i,j=0,1, PF(i,j)=\ell(i,j) $ and $ \ell(x,1)=\alpha (x-q^*)+\gamma $ and $ \ell(x,0)=\beta (x-q^*)+\gamma $.
\end{definition}

\gs{language below should be tightened/fixed.}
 \begin{remark}\label{remark:surjection}
If $\alpha > \beta$ then $G(\alpha,\beta,q^*,\gamma)$ is a proper scoring function, and, restricted to this case, $G$ is surjective onto all $PF$.
\end{remark}

\begin{remark}\label{remark:k}
Notice that by Definition~\ref{slope}, in the $ R_{\tru} $ region, the slope of the contours is $$k=-\frac{\beta q(0|1)}{\alpha q(1|0)},$$ so $ \beta $ can be represented as $ -k \alpha\frac{ q(1|0)}{q(0|1)} $.
\end{remark}

By Lemma~\ref{remark:normalization}, and Remarks~\ref{remark:surjection} and ~\ref{remark:k} it is enough to optimize over the space of $q^*$, $\gamma$, and $\alpha, k > 0$.

\gs{we could change the $\xi$ to $\Delta$, which might make more sense.}

\paragraph{Proof Overview for Theorem~\ref{opt_thm}} We will show three lemmas to prove Theorem~\ref{opt_thm}. In the first lemma (Lemma~\ref{cor:optgoal}), we show that the gap $ \Delta_Q(\mathcal{N}(PF)) $ we want to optimize depends only  on the value of $ q^* $ and the slope $ k $ of contours in the $ R_{\tru} $ region. Let $$ \xi(k,q^*)=\Delta_Q(\mathcal{N}(G(\alpha,-k \alpha\frac{q(1|0)}{q(0|1)},q^*,\gamma))). $$ The goal then is to solve $\sup_{k,q^*}\xi(k,q^*)$; notice that
\begin{align*}
\sup_{k,q^*}\xi(k,q^*)&=\sup_{q^*}\sup_{k}\xi(k,q^*).\\
\end{align*}
Recall that by Theorem~\ref{eqbinary}, for a fixed prior, the set of equilibria are determined by $q^*$.  Additionally, by Definition~\ref{slope}, the contours of the best response plot are determined by $k$.  So Lemma~\ref{cor:optgoal} formally states that the contours of the best response payoff plot encode all the information that is relevant to optimizing the payoff gap.

In the second lemma (Lemma~\ref{kvalue}), to show how to optimize $k$ when $q^*$ is fixed, we define $ k^{sup}(q^*)$ so that  for any $ q(1|0)<q^*<q(1|1) $,  $$ \xi(k^{sup}(q^*),q^*)=\sup_k \xi(k,q^*).$$

In the proof there are two main cases.

In the first case, $q(0|0) < q^*$: This corresponds to the case of 7 equilibrium.  Here, by Lemma~\ref{main_Lemma} and Corollary~\ref{the slope of contours}, the points in the best response plot that limit the range of $k$ are $\tru\zer$ and $\tru\one$.  For the same reason, for any $k$ that makes $\tru$ focal the equilibrium with payoff closest to that of $\tru$ must be either $\tru\zer$ or $\tru\one$.  It will turn out that the gap is maximized by making the payoffs of $\tru\zer$  and  $\tru\one$ equal.

In the second case, $ q^* \leq q(0|0)$: This corresponds to a setting where there are 8 or 9 equilibrium. This is similar to the first case except that now $\fal$ and $\tru\one$ are the constraining equilibria. It will turn out that the gap is maximized by making the payoffs of $\fal$  and  $\tru\one$ equal. (In the first case, $\fal$ was not an equilibrium, but in this case it is).

In the third lemma (Lemma~\ref{qstarvalue}), we will solve $$ \sup_{q^*}\xi(k^{sup}(q^*),q^*).$$

Once we know how to optimize $k$ for any fixed $q^*$, the third lemma shows how to optimize $q^*$. To prove Lemma~\ref{qstarvalue}, we map out the payment for different $q^*$ while making sure that
\begin{enumerate}
\item the payoffs of $\tru\zer$  and $\tru\one$  are equal; or \label{casea}
\item  the payoffs of $\fal$ and $\tru\one$ are equal. \label{caseb}
\end{enumerate}

In Case~\ref{casea}, as $q^*$ moves from $q(1|0)$ to $q(0|0)$, the gap first increases then decreases.  Similarly, in Case~\ref{caseb} as $q^*$ moves from $q(1|0)$ to $q(1|1)$, the gap first increases then decreases.

At this point there are two settings.  In the first setting, the gap in Case~\ref{casea} is maximized when $q^* >  q(0|0)$.  In some sense, this is the ``good" setting because if  $q^* >  q(0|0)$ there are 7 equilibria, and so Case~\ref{casea} gives the correct analysis.  Here we should take the maximum of Case~\ref{casea} and Case~\ref{caseb}.

In the second setting, the gap in Case~\ref{casea} is maximized when we have that  $q^* <  q(0|0)$. Here we cannot use the maximum from Case~\ref{casea} because when $q^* \leq q(0|0)$ there are more than 7 equilibria (in particular, we must worry about the payoff of $\fal$).  But the analysis of Case~\ref{casea} assumes that there are only 7 equilibria (and, in particular, that $\fal$ is not an equilibrium).  In this setting, the best permissable Case~\ref{casea} solution does not exist. Recall that to have 7 equilibria and obtain the gap in Case~\ref{casea}, we need $q^* > q(0|0)$, but because the gap for different $q^*$ are first increasing then decreasing, the gap will increase as we approach the boundary $q(0|0)$ from the right side.

In this setting we should choose between the value of Case~\ref{casea} at $q(0|0)$ or the maximal value of Case~\ref{caseb} , whichever is greater. If the maximal value is the maximal value of Case~\ref{caseb}, then we can obtain a mechanism that has this opitmal gap, and this mechanism has 9 equilibria. If the maximal value is at $q^*=q(0|0)$, then we cannot obtain a mechanism with optimal gap since we need $q^*>q(0|0)$ to have 7 equilibria. But we can obtain a gap that arbitrarily approaches the maximal value if we set $q^*$ arbitrarily close to $q(0|0)$.

Finally, when we plug in the optimal $q^*$ and $k$, we get the payment schemes as stated in the theorem.  The regions correspond to the different cases/settings above.  In $R_1$ we use the Case~\ref{casea} maximum; in $R_2$, we use the Case~\ref{caseb} maximum, and in $R_3$ we use the Case~\ref{casea} analysis, as $q^*$ limits to $q(0|0)$ from the right.

This will prove Theorem~\ref{opt_thm}.

\vspace{10pt}

We now formally state the three lemmas.  We will defer their proofs, and instead first show how they combine to prove Theorem~\ref{thm:optimization}.

%\begin{definition} \label{def:NG}
%$NG(k, q^*)=  G(\alpha, -k \alpha\frac{q(1|0)}{q(0|1)} , \hat{\alpha}(k,q^*), \gamma{\alpha}(k,q^*))$
%where
%\begin{align}
%\hat{\alpha}(k,q^*) = 1/\max\{(1+\frac{ k q(1|0)}{q(0|1)}) q^*,\frac{ k q(1|0)}{q(0|1)} \} \\
%\hat{\gamma}(k,q^*) = ??
%\end{definition}

\begin{lemma}\label{lemma:kqstar}
%For all $\gamma$; $\alpha, k > 0$; and  $\max\{q(0|0), q(1|0)\} < q^* < q(1|1)$: \gs{check this},
%$$ NG(k,q^*)= \mathcal{N}(G(\alpha,\beta,q^*,\gamma))$$
%\end{lemma}
If $ \alpha, k >0 $, $ \mathcal{N}(G(\alpha,-k \alpha\frac{q(1|0)}{q(0|1)},q^*,\gamma)) $ only depends on $ k$ and on $q^* $.  Thus $$ NG(k,q^*) \triangleq \mathcal{N}(G(\alpha,\beta,q^*,\gamma))$$ is well defined.
\end{lemma}

\begin{corollary}\label{cor:optgoal}
For any $\alpha, k > 0$, $ PF=G(\alpha,-k \alpha\frac{q(1|0)}{q(0|1)},q^*,\gamma) $, $$ \Delta_Q(\mathcal{N}(PF))=\nu^{\mathcal{N}(PF)}(\tru)- \max_{\theta\in \alleq_Q(\mathcal{N}(PF))/\{\tru,\zer,\one\}} \nu^{\mathcal{N}(PF)}(\theta) $$ only depends on $ q^* $ and on $ k $,  which is the slope of the contours in $ R_{\tru} $.
Thus $$ \xi(k,q^*) \triangleq \Delta_Q(\mathcal{N}(PF))$$ is well defined.
\end{corollary}

\begin{remark}
By Lemma~\ref{main_Lemma} and Corollary~\ref{the slope of contours}, we can assume $\alpha,k>0$ in the optimal line set since this is necessary for truth telling to be focal.
\end{remark}
\begin{definition}
Let $$ k^{sup}_a(q^*)=\frac{f(\tru\one)_y-f(\tru\zer)_y}{f(\tru\one)_x-f(\tru\zer)_x}, $$ $$ k^{sup}_b(q^*)=\frac{f(\tru\one)_y-f(\fal_y)}{f(\tru\one)_x-f(\fal_x)} $$ where $ f $ is the translation map (see Definition~\ref{translationmap}).
Let $$ k^{sup}(q^*)=\left\{
\begin{array}{rcl}
k^{sup}_a(q^*)       &      & {q^*>q(0|0)}\\
       &      &       \\
k^{sup}_b(q^*)    &      & {q^*\leq q(0|0)}
\end{array} \right. $$
\end{definition}

\yk{I begin from here}

\begin{lemma}\label{kvalue}
 $$ \xi(k^{sup}(q^*),q^*)=\sup_{k>0} \xi(k,q^*) $$
\end{lemma}

\begin{lemma}\label{qstarvalue}
\begin{enumerate}
\item
There exist $q^*_{a,o}$, $q^*_{b,o}$ such that
$$\xi(k^{sup}_a(q^*_{a,o}),q^*_{a,o}) =  \sup_{q^*}\xi(k^{sup}_a(q^*),q^*) \triangleq \Delta^*_a(Q)$$
$$\xi(k^{sup}_b(q^*_{b,o}),q^*_{b,o}) =  \sup_{q^*}\xi(k^{sup}_b(q^*),q^*) \triangleq \Delta^*_b(Q)$$
and
$ q(1|0)<q^*_{a,o}<q(1|1) $, $ q(1|0)<q^*_{b,o}<q(0|0) $.
\item
If $q^*_{a,o} > q(0|0)$ then
 $sup_{q^*}\xi(k^{sup}(q^*),q^*)$ is either attained at either $q^*_{a,o}$ or  $q^*_{b,o}$.
If $q^*_{a,o} \leq q(0|0)$ then
 $sup_{q^*}\xi(k^{sup}(q^*),q^*)$ is either attained at either $q^*_{b,o}$ or the right limit of $q(0|0)$.
 \item $$NG(k^{sup}_a(q^*_{a,o}),q^*_{a,o})= \mathcal{M}_1(Q),$$
 $$NG(k^{sup}_b(q^*_{b,o}),q^*_{b,o})=\mathcal{M}_2(Q),$$
 and when $ q^*_{a,o} \leq q(0|0) $, $ \epsilon>0 $, $$NG(k^{sup}_a(q(0|0)+\epsilon),q(0|0)+\epsilon)=\mathcal{M}_3(Q,\epsilon).$$
 %\yk{we need $ q^*_{a,o} \leq q(0|0) $ to calculate M3}
 %\gs{is this good enough, does this treat the third case okay?}
\end{enumerate}

%\gs{try to remove below here}
%
%
%$ \xi(k^{sup}_a(q^*),q^*) $ is first increasing and then decreasing in the $ [q(1|0), q(1|1)] $ region; $ \xi(k^{sup}_b(q^*),q^*) $ is first increasing and then decreasing in the $ [q(1|0), q(0|0)] $ region. If we denote $ \sup_{q^*}\xi(k^{sup}_a(q^*),q^*) $ by $ \Delta^*_a(Q) $; denote $ \sup_{q^*}\xi(k^{sup}_b(q^*),q^*) $ by $ \Delta^*_b(Q)$, then there exists $ q(1|0)<q^*_{a,o}<q(1|1) $, $ q(1|0)<q^*_{b,o}<q(0|0) $ such that $ \forall I\in\{a,b\} $ $$ \xi(k^{sup}_I(q^*_{I,o}),q^*_{I,o})=\Delta_I^*(Q)$$ and $ NG(k^{sup}_I(q^*_{I,o}),q^*_{I,o}) $ is a normalized payoff function and let $ PF^* $ denote it. We have $$ PF^*(0,1)=PF^*(1,0)=0.$$
%Moreover, $$NG(k^{sup}_a(q^*_{a,o}),q^*_{a,o})= \mathcal{M}_1(Q),$$
%$$NG(k^{sup}_b(q^*_{b,o}),q^*_{b,o})=\mathcal{M}_2(Q),$$
%and when $ q^*_{a,o} \leq q(0|0) $, $ \epsilon>0 $, $$NG(k^{sup}_a(q(0|0)+\epsilon),q(0|0)+\epsilon)=\mathcal{M}_3(Q,\epsilon).$$
\end{lemma}

We defer the proofs of these lemmas, and first show how Theorem~\ref{opt_thm} follows from them.

\begin{proof}[Proof of Theorem~\ref{opt_thm}]
\gs{this does not cover the first case of $R_1$}
\yk{I add it}
%Based on Lemma~\ref{kvalue} and Lemma~\ref{qstarvalue}, we have
%\begin{align*}
%\sup_{k,q^*}\xi(k,q^*)
%&=\max\{\sup_{q^*\leq q(0|0)}\sup_{k}\xi(k,q^*),\sup_{q^*> q(0|0)}\sup_{k}\xi(k,q^*)\}\\
%&=\max\{\sup_{q^*\leq q(0|0)}\xi(k^{sup}_b(q^*),q^*),\sup_{q^*> q(0|0)}\xi(k^{sup}_a(q^*),q^*)\}\\
%\end{align*}
%When $q(0|0)>q(1|0)$, there are nine equilibria if $q^*<q(0|0)$ and seven equilibria if $q^*>q(0|0)$. In this situation, $k^{sup}(q^*)$ will change from $k^{sup}_a(q^*)$ to $k^{sup}_b(q^*)$ at $q(0|0)$. We will consider this situation in two cases:
We will consider three cases:\\

\begin{description}
\item[Case 1 (when $q(0|0)\leq q(1|0)$)]
When $q(0|0)\leq q(1|0)$, if the mechanism has truth-telling as a strict Nash equilibrium, it must have seven equilibria. In this situation, $k^{sup}(q^*)$ is always $k^{sup}_a(q^*)$. So $\sup_{k,q^*}\xi(k,q^*)$
is attained at $(k^{sup}_a(q^*_{a,o}), q^*_{a,o})$. We also know by Lemma~\ref{qstarvalue} that
$NG(k^{sup}_a(q^*_{a,o}),q^*_{a,o})= \mathcal{M}_1(Q)$.  Additionally, we are in $R_1$ (by the definition of $R_1$).
\item[Case 2 (when $q(0|0)> q(1|0)\wedge q^*_{a,o}>q(0|0) $)]
In this case, by Lemmas~\ref{kvalue} and \ref{qstarvalue} we know that
$\sup_{k,q^*}\xi(k,q^*)$
is attained at $(k^{sup}_a(q^*_{a,o}), q^*_{a,o})$ or $(k^{sup}_b(q^*_{b,o}),q^*_{b,o})$.
\begin{enumerate}
\item If $\Delta^*_a(Q) > \Delta^*_b(Q)$, then we know that
 $\sup_{k,q^*}\xi(k,q^*)$ is attained at $(k^{sup}_a(q^*_{a,o}), q^*_{a,o})$.  We also know by Lemma~\ref{qstarvalue} that
$NG(k^{sup}_a(q^*_{a,o}),q^*_{a,o})= \mathcal{M}_1(Q)$.  Additionally, we are in $R_1$ (by the definition of $R_1$).

\item If $\Delta^*_a(Q) \leq \Delta^*_b(Q)$, then we know that
 $\sup_{k,q^*}\xi(k,q^*)$ is attained at $(k^{sup}_b(q^*_{b,o}), q^*_{b,o})$.  We also know by Lemma~\ref{qstarvalue} that
$NG(k^{sup}_b(q^*_{b,o}),q^*_{b,o})= \mathcal{M}_2(Q)$.  Additionally, we are in $R_2$ (by the definition of $R_2$).
\end{enumerate}

\item[Case 3 (when $q(0|0)> q(1|0)\wedge q^*_{a,o}\leq q(0|0) $)]
In this case, by Lemmas~\ref{kvalue} and \ref{qstarvalue} we know that
$\sup_{k,q^*}\xi(k,q^*)$
is attained by $\lim_{q^* \rightarrow q(0|0)^+} (k^{sup}_a(q^*), q^*)$ % the right limit of $q^*$ to $q(0|0)$ $(k^{sup}_a(q^*, q^*) $
 or $(k^{sup}_b(q^*_{b,o}), q^*_{b,o})$.

\begin{enumerate}
\item If $\xi(k^{sup}_a(q(0|0)),q(0|0)) \leq \Delta^*_b(Q)$, then we know that
 $\sup_{k,q^*}\xi(k,q^*)$ is attained at $(k^{sup}_b(q^*_{b,o}), q^*_{b,o})$.  We also know by Lemma~\ref{qstarvalue} that
$NG(k^{sup}_b(q^*_{b,o}),q^*_{b,o})= \mathcal{M}_2(Q)$.  Additionally, we are in $R_2$ (by the definition of $R_2$).

%\katrina{fix the missing close parentheses}\yk{I fix it}

\item If $\xi(k^{sup}_a(q(0|0)),q(0|0)) > \Delta^*_b(Q)$, then we know that
 $\sup_{k,q^*}\xi(k,q^*)$ is given by $$\lim_{q^* \rightarrow q(0|0)^+}\xi(k^{sup}_a(q^*), q^*)$$ 
We also know by Lemma~\ref{qstarvalue} that
$NG(k^{sup}_a(q(0|0)) + \epsilon, q(0|0) + \epsilon)= \mathcal{M}_3(Q, \epsilon)$.  Additionally, we are in $R_3$ (by the definition of $R_3$).
\end{enumerate} 
\end{description}

\end{proof}

\subsection{Proof for Lemma~\ref{lemma:kqstar} and Corollary~\ref{cor:optgoal}}\label{sec:optgoal}

%The proof of the Lemma is a straightforward calculation, and the Corollary easy follows.  \gs{better description}

\begin{proof}[Proof of Lemma~\ref{lemma:kqstar}]
%Let $ PF $ be $ G(\alpha,-k \alpha\frac{q(1|0)}{q(0|1)},q^*,\gamma) $.  Letting $\beta = -k \alpha\frac{q(1|0)}{q(0|1)}$ we have
%\begin{align*}
% PF(1,1)&=\alpha(1-q^*)+\gamma \\
% PF(1,0)&=\beta(1-q^*)+\gamma  \\
% PF(0,1)&=\alpha(-q^*)+\gamma \\
% PF(0,0)&=\beta(-q^*)+\gamma
%\end{align*}
Let $ PF $ be $ G(\alpha,-k \alpha\frac{q(1|0)}{q(0|1)},q^*,\gamma) $.  Then this payoff function only depends on $\{\ell(\cdot, 0), \ell(\cdot, 1)\}$.  However, we have
\begin{align}
 PF(x,0)&=-k \alpha\frac{q(1|0)}{q(0|1)}(x-q^*)+\gamma   \label{equation:PF} \\
 PF(x,1)&=\alpha(x-q^*)+\gamma  \nonumber 
\end{align}
It follows that $ \frac{PF(\cdot,\cdot) -\gamma}{\alpha} $ only depends on $ k$ and $q^* $.  However, 

\begin{align*}
\mathcal{N}(G(\alpha,-k \alpha\frac{q(1|0)}{q(0|1)},q^*,\gamma))(\cdot,\cdot)=&\frac{PF(\cdot,\cdot)-\min_{i,j}PF(i,j)}{\max_{i,j}PF(i,j)-\min_{i,j}PF(i,j)}\\
=&\frac{\frac{PF(\cdot,\cdot)-\gamma}{\alpha}-\min_{i,j}\frac{PF-\gamma}{\alpha}(i,j)}{\max_{i,j}\frac{PF-\gamma}{\alpha}(i,j)-\min_{i,j}\frac{PF-\gamma}{\alpha}(i,j)}\\
\end{align*}

The last equality follows from the fact that we merely shifted and scaled all inputs.  Since $ \frac{PF(\cdot,\cdot)-\gamma}{\alpha} $ only depends on $ k,q^* $, the Lemma follows.

%
% We denote $ \mathcal{N}(G(\alpha,-k \alpha\frac{q(1|0)}{q(0|1)},q^*,\gamma)) $ by $$  NG(k,q^*) $$
% Moreover, $$NG(k,q^*)=\left(\begin{array}{ccc}
%\hat{\alpha}(k,q^*)(1-q^*)+\hat{\gamma}(k,q^*) & & -\hat{\alpha}(k,q^*)\frac{k q(1|0)}{q(0|1)} (1-q^*)+\hat{\gamma}(k,q^*)\\
%\hat{\alpha}(k,q^*)(-q^*)+\hat{\gamma}(k,q^*) & &-\hat{\alpha}(k,q^*)\frac{k q(1|0)}{q(0|1)} (-q^*)+\hat{\gamma}(k,q^*) \\
%\end{array}\right)$$ where $$\hat{\alpha}(k,q^*)=\frac{1}{\max_{i,j}\frac{PF(i,j)-\gamma}{\alpha}-\min_{i,j}\frac{PF(i,j)-\gamma}{\alpha}}$$
%$$ \hat{\gamma}(k,q^*)=-\frac{\min_{i,j}\frac{PF(i,j)-\gamma}{\alpha}}{\hat{\alpha}(k,q^*)} $$
\end{proof}

\begin{proof}[Proof of Corollary~\ref{cor:optgoal}]
From Theorem~\ref{eqbinary} we know that the set of equilibria, $\alleq_Q(q^*)$, only depends on $q^*$.  From Lemma~\ref{lemma:kqstar}, we have that $\mathcal{N}(PF)$ only depends on $k$ and $q^*$.  The Corollary follows immediately because $\Delta_Q(\mathcal{N}(PF))$ only depends on the equilibria and the payoffs of $\mathcal{N}(PF)$.

%Since we only consider $ PF $ that makes truth-telling focal, then because $ \alpha, k>0 $, based on Remark~\ref{normpayoffkqstar}, we have \gs{what is psi anyhow??}$$ \alleq_Q(\mathcal{N}(PF))=\alleq_Q(q^*) $$
% and
% $$ \Delta_Q(\mathcal{N}(PF))=\Delta_Q(NG(k,q^*)) $$ and the lemma follows
\end{proof}

\subsection{Proof for Lemma~\ref{kvalue}}\label{sec:kvalue}

\begin{proof}[Proof of Lemma~\ref{kvalue}]
We show this Lemma in two cases. In each case, we will optimize the gap between the best equilibrium $\tru$ and the second best equilibrium; there are two possible second best equilibria. We will prove that once we set the parameters such that the payoffs of these two possible second best equilibria are same, the gap is optimized.    \\
(1) If $ q(0|0)<q^*<q(1|1) $:\\
Based on Lemma~\ref{main_Lemma} and Corollary~\ref{the slope of contours}, there will be seven equilibria,
and $$ \max_{\theta\in \alleq_Q(PF)/\{\tru,\zer,\one\}} \nu^{PF}(\theta)=\max\{\nu^{PF}(\tru \one),\nu^{PF}(\tru \zer)\}. $$  So
\begin{align*}
 &\sup_{PF} [\nu^{\mathcal{N}(PF)}(\tru)- \max_{\theta\in \alleq/\{\tru,\zer,\one\}} \nu^{\mathcal{N}(PF)}(\theta)]\\
 =
 &\sup_{PF} [\min\{\nu^{\mathcal{N}(PF)}(\tru)-\nu^{\mathcal{N}(PF)}(\tru \one),\nu^{\mathcal{N}(PF)}(\tru)-\nu^{\mathcal{N}(PF)}(\tru \zer)\}]\\
\end{align*}
Since $ q^* $ is given, let $$ \phi^{q^*}(k) = \nu^{\mathcal{N}(PF)}(\tru)-\nu^{\mathcal{N}(PF)}(\tru \one)$$
$$ \psi^{q^*}(k) = \nu^{\mathcal{N}(PF)}(\tru)-\nu^{\mathcal{N}(PF)}(\tru \zer)$$
We have $$ \xi(k,q^*)=\min\{\phi^{q^*}(k),\psi^{q^*}(k)\} $$ We will show both $ \phi^{q^*}(k) $ is increasing and $ \psi^{q^*}(k) $ is decreasing  whenever $ \frac{q(0|1)}{q(0|0)}<k<\frac{q(1|1)}{q(1|0)}$.

This is the important range of $k$'s because
$ \tru $ being focal implies that
 $$ \frac{q(0|1)}{q(0|0)}=\frac{f(\tru\one)_y-f(\tru)_y}{f(\tru\one)_x-f(\tru)_x}<k<\frac{f(\tru\zer)_y-f(\tru)_y}{f(\tru\zer)_x-f(\tru)_x}=\frac{q(1|1)}{q(1|0)}. $$
The translation map $f$ is defined in Table~\ref{9etable}. 

It follows that  $ \xi(k,q^*) $ obtains its maximum when
 $ \phi^{q^*}(k)=\psi^{q^*}(k) $, which means $$ \phi^{q^*}(k^{sup}_a(q^*))=\psi^{q^*}(k^{sup}_a(q^*)) $$  $$ \Rightarrow \nu^{\mathcal{N}(PF)}(\tru\one)=\nu^{\mathcal{N}(PF)}(\tru\zer). $$ 
 So the payoff contour lines will go through point $f(\tru\one)$ and point $f(\tru\zer)$, which implies that the slope of the contour lines $k^{sup}_a(q^*)$ is
 
 $$ \Rightarrow k^{sup}_a(q^*)=\frac{f(\tru\one)_y-f(\tru\zer)_y}{f(\tru\one)_x-f(\tru\zer)_x}. $$

It is now left to show $ \phi^{q^*}(k) $ is increasing and $ \psi^{q^*}(k) $ is decreasing.\\

%about k

To show that $ \phi^{q^*}(k) $ is increasing, we show

\begin{align} \label{eqn:phiqstar}
\phi^{q^*}(k)&= \frac{q(1|0) (q^* - q(1|0))}{q(0|1)+q(1|0)} \min\left\{ \frac{q(0|1)}{q^* q(1|0)} (1-\frac{\frac{q(0|1)}{q(0|0)q(1|0)}}{k+\frac{q(0|1)}{q(1|0)}}) , \frac{q(0|1)}{q(1|0)}-\frac{1}{k}\frac{q(0|1)^2}{q(0|0)q(1|0)} \right\}
\end{align}

From this we can see $ \phi^{q^*}(k) $ is increasing since for any two increasing functions $ f_1,f_2 $, $ \min\{f_1,f_2\} $ is still a increasing function.

\begin{align} %\label{eqn:phiqstarstart}
\phi^{q^*}(k)&=\nu^{NG(k, q^*)}(\tru)-\nu^{NG(k, q^*)}(\tru \one)
%&=\hat{\alpha}(k,q^*) (-q(0)\frac{ k q(1|0)}{q(0|1)} (f(\tru)_x-f(\tru\one)_x)+q(1) (f(\tru)_y-f(\tru\one)_y))\
\end{align}

To arrive at Equation~\ref{eqn:phiqstar} we must compute $\nu^{NG(k, q^*)}(\tru)-\nu^{NG(k, q^*)}(\tru \one)$.  

It will help to better understand $NG$.
%Recall $NG(k,q^*) =  \mathcal{N}(G(\alpha,-k \alpha\frac{q(1|0)}{q(0|1)},q^*,\gamma)) $ by $$  NG(k,q^*) $$
By Lemma~\ref{lemma:kqstar} we know that there exist functions $\hat{\alpha}(k,q^*)$ and $\hat{\gamma}(k,q^*)$ such that:
\begin{align} \label{equation:NG}
NG(k,q^*)=\left(\begin{array}{ccc} \hat{\alpha}(k,q^*)(1-q^*)+\hat{\gamma}(k,q^*) & & -\hat{\alpha}(k,q^*)\frac{k q(1|0)}{q(0|1)} (1-q^*)+\hat{\gamma}(k,q^*)\\ \hat{\alpha}(k,q^*)(-q^*)+\hat{\gamma}(k,q^*) & &-\hat{\alpha}(k,q^*)\frac{k q(1|0)}{q(0|1)} (-q^*)+\hat{\gamma}(k,q^*) \\ \end{array}\right).
\end{align}

So $NG(k,q^*)(x,0)=-\hat{\alpha}(k,q^*)\frac{k q(1|0)}{q(0|1)} (x-q^*)+\hat{\gamma}(k,q^*)$ and $NG(k,q^*)(x,1)=\hat{\alpha}(k,q^*)(x-q^*)$

%Since the two linear functions that $ \mathcal{N}(PF) $ represents are $ \ell^{\mathcal{N}(PF)}(y,1)=\hat{\alpha}(k,q^*)(y-q^*)+\hat{\gamma}(k,q^*) $ and $ \ell^{\mathcal{N}(PF)}(x,0)=-\hat{\alpha}(k,q^*)\frac{k q(1|0)}{q(0|1)} (x-q^*)+\hat{\gamma}(k,q^*) $ and the expected payoff of points $ (x,y) $ in the $ R_{\tru} $ region is \gs{why below?} $$\nu^{\mathcal{N}(PF)}((x,y))= q(0)\ell^{\mathcal{N}(PF)}(x,0)+q(1)\ell^{\mathcal{N}(PF)}(y,1) $$ we have

From Table~\ref{tabel:best-response-by-quadrant}, we know the best payoff of point $ (x,y) $ in the $R_{\tru}$ quadrant is
$$ q(0)\ell(x,0 )+q(1)\ell(y,1 )=q(0)(\beta(x-q^*)+\gamma)+q(1)(\alpha(y-q^*)+\gamma).$$
The payoff function is $NG(k,q^*)$, so $\alpha=\hat{\alpha}(k,q^*),\beta=-\hat{\alpha}(k,q^*)\frac{k q(1|0)}{q(0|1)}$. To calculate the payoff of $\tru$, we replace $x$ with $f(\tru)_x$ and $y$ with $f(\tru)_y$ since $f(\tru)$ is the point in $R_{\tru}$ that represents $\tru$. Similarly, we calculate the payoff of $\tru\one$. Now we have

\katrina{There are two equations with the label below}\yk{I fix it}
\begin{align} \label{eqn:phiqstarstart}
\phi^{q^*}(k)&=\nu^{NG(k, q^*)}(\tru)-\nu^{NG(k, q^*)}(\tru \one)\\
&=\hat{\alpha}(k,q^*) (-q(0)\frac{ k q(1|0)}{q(0|1)} (f(\tru)_x-f(\tru\one)_x)+q(1) (f(\tru)_y-f(\tru\one)_y)). \nonumber
\end{align}

We would like to compute $\hat{\alpha}(k,q^*)$ more explicitly.  
Fix an arbitrary $\alpha, \gamma$, and let $PF = G(\alpha,-k \alpha\frac{q(1|0)}{q(0|1)},q^*,\gamma)$ %\gs{fix this}
Then we know that
\begin{align} \label{equation:alpha}
\hat{\alpha}(k,q^*)=\frac{1}{\max_{i,j}\frac{PF(i,j)-\gamma}{\alpha}-\min_{i,j}\frac{PF(i,j)-\gamma}{\alpha}}
\end{align}

% $$ \hat{\gamma}(k,q^*)=-\frac{\min_{i,j}\frac{PF(i,j)-\gamma}{\alpha}}{\hat{\alpha}(k,q^*)} $$

\begin{claim}\label{claim_l00}
When $PF$ has seven equilibria, 
$$ \max_{i,j=0,1}PF(i,j)=PF(0,0)=\ell(0,0) $$
\end{claim}

\begin{proof}

Since \begin{align} \label{Inequality:Seven Equilibria:k}
\frac{q(0|1)}{q(0|0)}< k <\frac{q(1|1)}{q(1|0)}
\end{align}   
we can see
\begin{align*}
 PF(1,1)-PF(0,0)&=\ell(1,1)-\ell(0,0)\\
 &= \alpha(1-q^*)-\beta(0-q^*)\\
 &= \alpha(1-q^*)-\frac{k q(1|0)}{q(0|1)} q^*\\
 &=\alpha(1-\frac{(q(0|1)+k q(1|0)) q^*}{q(0|1)})\\
 &<\alpha(1-\frac{(q(0|1)+\frac{q(0|1)}{q(0|0)} q(1|0)) q^*}{q(0|1)})\\
 &=\alpha (1-\frac{q^*}{q(0|0)})\\
 &<0
\end{align*}
\end{proof}
The first inequality follows from the bound on $k$ in (\ref{Inequality:Seven Equilibria:k}).
%\gs{each step of the following derivation should be explained}
So

\begin{align}
\frac{1}{\hat{\alpha}(k,q^*)}&=\max_{i,j}\frac{PF(i,j)-\gamma}{\alpha}-\min_{i,j}\frac{PF(i,j)-\gamma}{\alpha} \nonumber\\
&=\frac{PF(0,0)-\gamma}{\alpha}-\min\{ \frac{PF(0,1)-\gamma}{\alpha},\frac{PF(1,0)-\gamma}{\alpha} \} \label{equation:alpha:explicit}\\
&=-\frac{kq(1|0)}{q(0|1)}(-q^*)-\min\{-q^*,-\frac{kq(1|0)}{q(0|1)}(1-q^*)\} \nonumber\\
&=\max\{(1+\frac{ k q(1|0)}{q(0|1)}) q^*,\frac{ k q(1|0)}{q(0|1)} \} \nonumber
\end{align}

The first equality is from Equation (\ref{equation:alpha}). The second equality is from Claim~\ref{claim_l00} and the fact that $PF$ is minimized at either $(1,0)$ or $(0,1)$. The third line follows from equation (\ref{equation:PF}). 

%\gs{STOPPED HERE}

Finishing up we have that:
\begin{align}
\phi^{q^*}(k)&=\nu^{\mathcal{N}(PF)}(\tru)-\nu^{\mathcal{N}(PF)}(\tru \one) \nonumber \\
%&=\frac{\nu^{PF}(\tru)-\nu^{PF}(\tru \one)}{\max\{(\alpha+\frac{\alpha k q(1|0)}{q(0|1)}) q^*,\frac{\alpha k q(1|0)}{q(0|1)} \}}\\
&=\hat{\alpha}(k,q^*) (-q(0)\frac{ k q(1|0)}{q(0|1)} (f(\tru)_x-f(\tru\one)_x)+q(1) (f(\tru)_y-f(\tru\one)_y))  \label{equation:phi} \\
%&=(\frac{q(0|1)}{q(0|1)+q(1|0)} \frac{-\alpha k q(1|0)}{q(0|1)} (\tru_x-\tru\one_x) + \frac{q(1|0)}{q(0|1)+q(1|0)} \alpha (\tru_y-\tru\one_y))/(\max\{(\alpha+\frac{\alpha k q(1|0)}{q(0|1)}) q^*,\frac{\alpha k q(1|0)}{q(0|1)} \})\\
%&=\frac{q(1|0)\alpha (\tru\one_x - \tru_x)}{q(0|1)+q(1|0)}(k-\frac{\tru_y-\tru\one_y}{\tru_x - \tru\one_x})/(\max\{(\alpha+\frac{\alpha k q(1|0)}{q(0|1)}) q^*,\frac{\alpha k q(1|0)}{q(0|1)} \})\\
%&=\frac{q(1|0) (\tru\one_x - \tru_x)}{q(0|1)+q(1|0)}(k-\frac{q(0|1)}{q(0|0)})/(\max\{(1+\frac{ k q(1|0)}{q(0|1)}) q^*,\frac{ k q(1|0)}{q(0|1)} \})\\
&=\frac{q(1|0) (q^* - q(1|0))}{q(0|1)+q(1|0)}(k-\frac{q(0|1)}{q(0|0)})/(\max\{(1+\frac{ k q(1|0)}{q(0|1)}) q^*,\frac{ k q(1|0)}{q(0|1)} \}) \nonumber \\
&=\frac{q(1|0) (q^* - q(1|0))}{q(0|1)+q(1|0)} \min\{ \frac{q(0|1)}{q^* q(1|0)} (1-\frac{\frac{q(0|1)}{q(0|0)q(1|0)}}{k+\frac{q(0|1)}{q(1|0)}}) , \frac{q(0|1)}{q(1|0)}-\frac{1}{k}\frac{q(0|1)^2}{q(0|0)q(1|0)} \nonumber \}
\end{align}
The second line is from (\ref{eqn:phiqstarstart}). After we plug $\hat{\alpha}(k,q^*)$ value from (\ref{equation:alpha:explicit}) and translation map value from Table~\ref{9etable}, we obtain the third line. The fourth line is manipulation. 

This verifies Equation~\ref{eqn:phiqstar}.

Similarly, we replace $ \tru\one $ with $ \tru\zer $
\begin{align*}
&\psi^{q^*}(k)=\frac{q(1|0) (f(\tru)_x - f(\tru\zer)_x)}{q(0|1)+q(1|0)}(\frac{q(1|1)}{q(1|0)}-k)/(\max\{(1+\frac{ k q(1|0)}{q(0|1)}) q^*,\frac{ k q(1|0)}{q(0|1)} \})\\
\end{align*}
We can see $ \psi^{q^*}(k) $ is decreasing since when k is increasing, $\frac{q(1|1)}{q(1|0)}-k$ is decreasing and $\max\{(1+\frac{ k q(1|0)}{q(0|1)}) q^*,\frac{ k q(1|0)}{q(0|1)} \}$ is increasing.\\

(2) If $ q(1|0)<q^*< q(0|0) $:\\

Based on Lemma~\ref{main_Lemma}, there will be nine or eight equilibria, and $$ \frac{q(0|1)}{q(0|0)}=\frac{f(\tru\one)_y-f(\tru)_y}{f(\tru\one)_x-f(\tru)_x}<k<\frac{f(\fal)_y-f(\tru)_y}{f(\fal)_x-f(\tru)_x}=\frac{(q(0|1)) (1-q^*)}{q(1|0)q^*} $$ if $ \tru $ is focal.
In the previous case, the second best equilibrium is $\tru\one$ or $\tru\zer$ while in this case, the second best equilibrium is $\tru\one$ or $\fal$. Now we will use the same method in the previous case to prove that once we set the payoff of $\tru\one$ equal to the payoff of $\fal$, the gap between the best equilibrium $\tru$ and the second best equilibrium will be optimized. 

The gap we will optimize is the minimum of $\phi^{q^*}(k)$ and $\psi^{q^*}(k)$, where 
$$ \phi^{q^*}(k) = \nu^{\mathcal{N}(PF)}(\tru)-\nu^{\mathcal{N}(PF)}(\tru \one)$$
$$ \psi^{q^*}(k) = \nu^{\mathcal{N}(PF)}(\tru)-\nu^{\mathcal{N}(PF)}(\fal)$$
To calculate them, we need to give an explicit form of $\hat{\alpha}(k,q^*)$. Similarly, we have the following claim:

\begin{claim}
When $PF$ has nine equilibria,
$$\max_{i,j=0,1}PF(i,j)=PF(1,1)$$
\end{claim}
\begin{proof}
Since $k<\frac{(q(0|1)) (1-q^*)}{q(1|0)q^*}$,
\begin{align*}
 PF(1,1)-PF(0,0)&=\ell(1,1)-\ell(0,0)\\
 &=\alpha(1-\frac{(q(0|1)+k q(1|0)) q^*}{q(0|1)})\\
 &>\alpha(1-\frac{(q(0|1)+\frac{(q(0|1)) (1-q^*)}{q(1|0)q^*}  q(1|0)) q^*}{q(0|1)})\\
 &=0 \qedhere
\end{align*}
\end{proof}

To prove $\phi^{q^*}(k)$ is increasing and $\psi^{q^*}(k)$ is decreasing, we will write them explicitly. We omit the proof here since the calculations are similar to the previous case, replacing $f(\tru\zer)$ with $f(\fal)$. 

Once we have proved $\phi^{q^*}(k)$ is increasing and $\psi^{q^*}(k)$ is decreasing, the gap $\min\{\phi^{q^*}(k),\psi^{q^*}(k)\}$ will be optimized when

$$ \phi^{q^*}(k^{sup}_b(q^*))=\psi^{q^*}(k^{sup}_b(q^*)) $$  $$ \Rightarrow \nu^{\mathcal{N}(PF)}(\tru\one)=\nu^{\mathcal{N}(PF)}(\fal) $$ $$ \Rightarrow k^{sup}_b(q^*)=\frac{f(\tru\one)_y-f(\fal)_y}{f(\tru\one)_x-f(\fal)_x} \qedhere$$ \end{proof}

\subsection{Proof for Lemma~\ref{qstarvalue}}\label{sec:qstarvalue}

\begin{figure}
\centering
\includegraphics[scale=0.45]{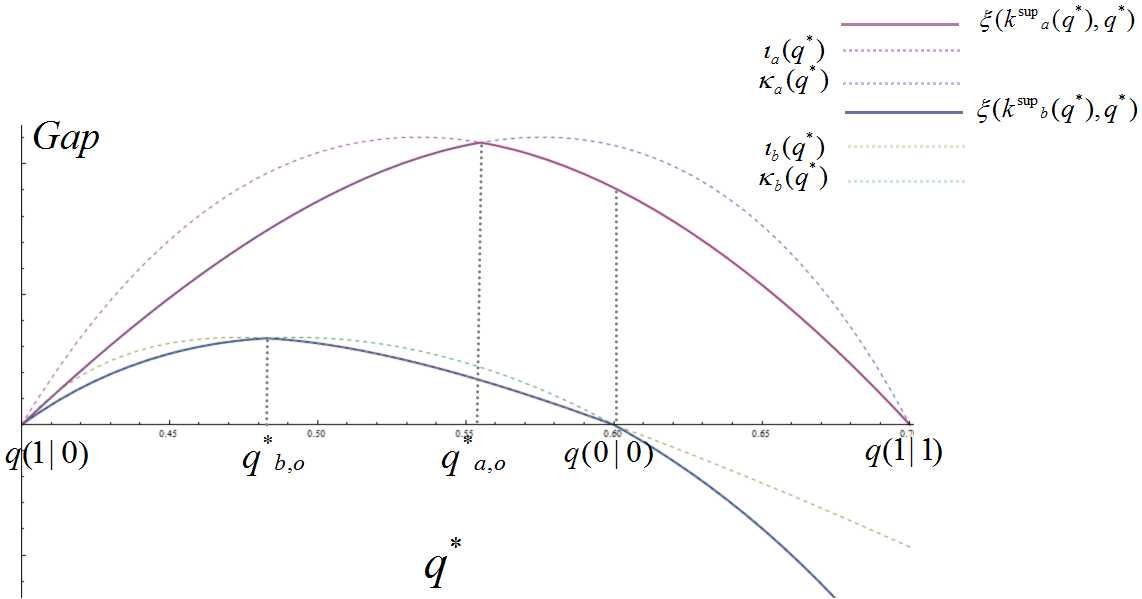}
\caption{The x axis is $ q^* $, the y axis is the gap we would like to optimize. This figure is an illustration of the results in Lemma~\ref{qstarvalue}. }
\label{fig:Unattainablecase}
\end{figure}

\begin{figure}
\centering
\includegraphics[scale=0.5]{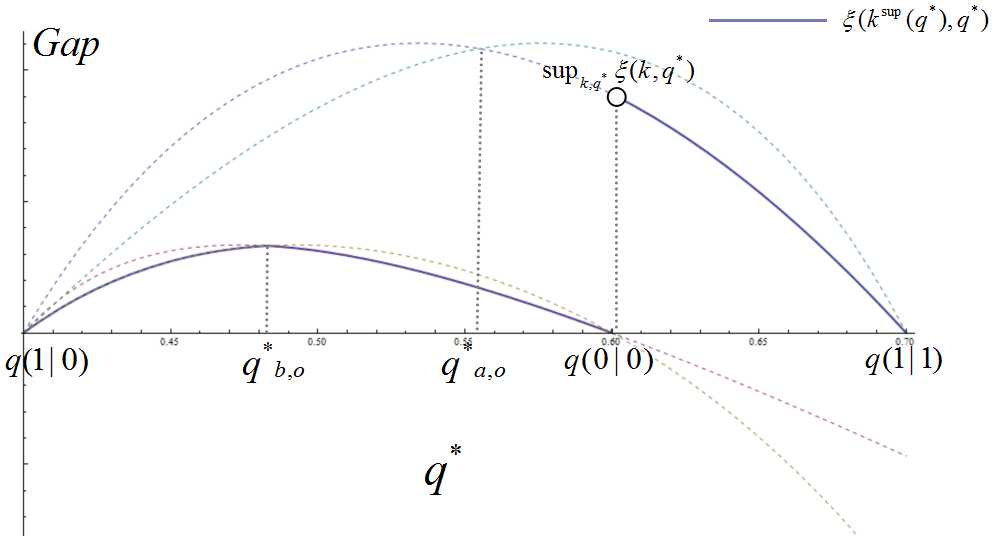}
\caption{Here is the case: $ q(0|0)>q^*_{a,o} $ and $ \sup_{k,q^*}\xi(k,q^*)=\lim_{\epsilon\rightarrow 0^{+}}\xi(k^{sup}_a(q(0|0)+\epsilon),q(0|0)+\epsilon) $, since we cannot obtain this supremum, this prior here is unattainable.  }
\label{fig:Unattainablecase2}
\end{figure}

\begin{proof}[Proof of Lemma~\ref{qstarvalue}]
The framework of the proof is as follows:
\begin{enumerate}
\item Show that $ \xi((k^{sup}_I(q^*),q^*)=\min\{\kappa_I(q^*),\iota_I(q^*) \} $ for some $\kappa_I(q^*)$, $\iota_I(q^*)$ to be defined later.
\item  Show that $\max_{q^*}\kappa_I(q^*)=\max_{q^*}\iota_I(q^*)$
\item Show that for $ \kappa_I(q^*)$ from above, there exists a point $q^*_{stat,\kappa}$ such that for $q^*< q^*_{stat,\kappa}, \kappa_I(q^*)$ is increasing, and for $q^*> q^*_{stat,\kappa}, \kappa_I(q^*)$ is decreasing.  Similarly for  $ \iota_I(q^*)$ and $q^*_{stat,\iota}$
\item Show that 2) and 3) imply that %if $q^*_{stat,\kappa} \neq q^*_{stat,\iota}$, then
there exists a $q_0^*$ between (or equal to) $q^*_{stat,\kappa}$ and $q^*_{stat,\iota}$ such that $\kappa_I(q^*_o)=\iota_I(q^*_o)$. 
\item Show that $ \xi(k^{sup}_I(q^*_o),q^*_o)=\max_{q^*}\xi(k^{sup}_I(q^*),q^*)$
\item Show that $NG(k^{sup}_a(q^*_{a,o}),q^*_{a,o})= \mathcal{M}_1(Q);NG(k^{sup}_b(q^*_{b,o}),q^*_{b,o})= \mathcal{M}_2(Q);NG(k^{sup}_a(q(0|0)) + \epsilon, q(0|0) + \epsilon)= \mathcal{M}_3(Q, \epsilon)$ by calculating the coordinates of $q^*_{I,o}$ and proving that $
NG(k^{sup}_I(q^*_{I,o}),q^*_{I,o})(0,1)=$ $NG(k^{sup}_I(q^*_{I,o}),q^*_{I,o})(1,0),I=a,b $.

\end{enumerate}
\yk{part 1 to 6 begin here}
\paragraph{Part (1) Show that $ \xi((k^{sup}_I(q^*),q^*)=\min\{\kappa_I(q^*),\iota_I(q^*) \} $ for some $\kappa_I(q^*)$, $\iota_I(q^*)$ to be defined later:\\} 
We show part (1) in two cases: $I=a,I=b$\\
$When  I = a $:\\
Let $$ \kappa_a(q^*)=\frac{q(0|1)}{q(1|0)+q(0|1)}\frac{(q(1|1)-q^*)(q^*-q(1|0))}{q(1|1) (1-q^*)} $$ $$ \iota_a(q^*)=\frac{q(0|1)}{q(1|0)+q(0|1)} \frac{q(1|0)(q(1|1)-q^*)(q^*-q(1|0))}{ q^*((q(1|1)-q(0|0))q^*-q(1|0)q(1|1)) } $$
\smallskip\\
Based on Lemma~\ref{kvalue}, $ k^{sup}_a(q^*)=\frac{f(\tru\one)_y-f(\tru\zer)_y}{f(\tru\one)_x-f(\tru\zer)_x}=\frac{q(1|1) (1-q^*)}{q(0|0) q^*} $
\begin{align*}
\xi(k^{sup}_a(q^*),q^*)&=\phi^{q^*}(k^{sup}_a(q^*))\\
%&=\frac{q(1|0) (f(\tru\one)_x - f(\tru)_x)}{q(0|1)+q(1|0)}(k^{sup}_a(q^*)-\frac{q(0|1)}{q(0|0)})/(\max\{(1+\frac{ k^{sup}_a(q^*) q(1|0)}{q(0|1)}) q^*,\frac{ k^{sup}_a(q^*) q(1|0)}{q(0|1)} \})\\
&=\frac{q(1|0) (q^* - q(1|0))}{q(0|1)+q(1|0)}(k^{sup}_a(q^*)-\frac{q(0|1)}{q(0|0)})/(\max\{(1+\frac{ k^{sup}_a(q^*) q(1|0)}{q(0|1)}) q^*,\frac{ k^{sup}_a(q^*) q(1|0)}{q(0|1)} \})\\
&=\frac{1}{q(1|0)+q(0|1)}\min\{ \frac{q(1|0)q(0|1)(q(1|1)-q^*)(q^*-q(1|0))}{ q^*((q(1|1)-q(0|0))q^*-q(1|0)q(1|1)) },\frac{q(0|1)(q(1|1)-q^*)(q^*-q(1|0))}{q(1|1) (1-q^*)} \}
\end{align*}
The first line follows from (\ref{equation:phi}). After we plug $k^{sup}_a(q^*)$ in it, we obtain the second line. 

Then we show $$ \xi(k^{sup}_a(q^*),q^*)=\min\{ \kappa_a(q^*),\iota_a(q^*) \} $$
\newline
$When  I=b $:\\
Let $$ \kappa_b(q^*)=\frac{q(0|1)(q(1|1)-q(0|0))}{q(1|0)+q(0|1)}\frac{(q^*-q(1|0))(q(0|0)-q^*)}{q(0|0)q(1|0)(q^*-q(0|1))} $$
$$ \iota_b(q^*)=\frac{q(0|1)(q(1|1)-q(0|0))}{q(1|0)+q(0|1)} \frac{(q^*-q(1|0))(q(0|0)-q^*)}{((q(0|0)q(1|0)+q(0|1)(q^*-1))(q^*-1)}$$ 
\smallskip\\

Based on Lemma~\ref{kvalue}, $ k^{sup}_b(q^*)=\frac{f(\tru\one)_y-f(\fal)_y}{f(\tru\one)_x-f(\fal)_x} $.
Like we did when $ I=a $, after all substitutions, we will get
\begin{align*}
\xi(k^{sup}_b(q^*),q^*)=\min\{\kappa_b(q^*),\iota_b(q^*) \} \\
\end{align*}

\paragraph{Part (2) Show that $\max_{q^*}\kappa_I(q^*)=\max_{q^*}\iota_I(q^*)$:\\ }
We will show Part (2) in two cases:\\
$When  I=a: $\\
Now we will show  $$ \max_{q^*}\kappa_a(q^*)=\max_{q^*}\iota_a(q^*) $$ when $ q(1|0)<q^*<q(1|1) $. 

We first show $$ \kappa_a(q(1|0) + q(1|1) - \frac{q(1|0) q(1|1)}{q^*})=\iota_a(q^*) $$ Once we show it, it's clear that $ \max_{q^*}\kappa_a(q^*)=\max_{q^*}\iota_a(q^*) $ since $q(1|0) + q(1|1) - \frac{q(1|0) q(1|1)}{q^*}$ is a bijection between $[q(1|0),q(1|1)]$ and itself. 

By arithmetic manipulation, we get:

\begin{align*}
\kappa_a(q(1|0) + q(1|1) - \frac{q(1|0) q(1|1)}{q^*})&=\frac{q(0|1)}{q(1|0)+q(0|1)}\frac{(\frac{q(1|0)q(1|1)}{q^*}-q(1|0))(q(1|1)-\frac{q(1|0)q(1|1)}{q^*})}{q(1|1) (q(0|0) - q(1|1) + \frac{q(1|0) q(1|1)}{q^*})}\\
&=\frac{q(0|1)}{q(1|0)+q(0|1)} \frac{q(1|0)(q(1|1)-q^*)(q^*-q(1|0))}{ q^*((q(1|1)-q(0|0))q^*-q(1|0)q(1|1)) }\\
&=\iota_a(q^*)\\
\end{align*}
\newline 
$When  I=b: $ \\
Similarly, we show $$ \kappa_b(\frac{q(0|0)q(1|0)}{1-q^*})=\iota_b(q^*) $$ then it is clear that $$ \max_{q^*}\kappa_b(q^*)=\max_{q^*}\iota_b(q^*) $$ when $ q(1|0)<q^*\leq q(0|0) $ since $ \frac{q(0|0)q(1|0)}{1-q^*} $ is a bijection function between $[q(1|0),q(0|0)]$ and itself.

\paragraph{Part (3) Show that for $ \kappa_I(q^*)$ from above, there exists a point $q^*_{stat,\kappa}$ such that for $q^*< q^*_{stat,\kappa}, \kappa_I(q^*)$ is increasing, and for $q^*> q^*_{stat,\kappa}, \kappa_I(q^*)$ is decreasing. Similarly for  $ \iota_I(q^*)$ and $q^*_{stat,\iota}$:\\}
We will show that $ \kappa_I(q^*) $ is first increasing and then decreasing. 

For $I=a$, we will use two facts to prove it:
\begin{enumerate} 
\item $ \kappa_a'(q(1|0))>0,\kappa_a'(q(1|1))<0 $
\item The derivative $ \kappa_a'(q^*) $ is a rational function where the denominator is positive and the numerator is a degree 2 polynomial. 
\end{enumerate}
The first fact implies that the derivative changes sign ($+\rightarrow -$ or $-\rightarrow +$ ) odd times. The second fact implies that the derivative changes sign at most two times. So they implies that the derivative changes sign only once at a critical point. We define this critical point as $q^*_{stat,\kappa}$. Then the result follows. 

Now it is left to show the two facts. By simple observation, We have $ \kappa_a(q^*) $ is positive when $q(1|0)<q^*\leq q(1|1)$ and $ \kappa_a(q(1|0))=0,\kappa_a(q(1|1))=0 $. Then the first fact follows. By taking derivatives, the second fact follows. 

Since $ q(1|0) + q(1|1) - \frac{q(1|0) q(1|1)}{q^*} $ is a increasing function,  $ \iota_a(q^*) $ is also first increasing and then decreasing when $ q(1|0)\leq q^* \leq q(1|1) $.\\

Similarly, when $ I=b $, we can also prove that both $ \kappa_b(q^*) $ and $ \iota_b(q^*) $ are first increasing and then decreasing when $ q(1|0)\leq q^*\leq q(0|0) $.\\

\paragraph{Part (4) Show that 2) and 3) imply that %if $q^*_{stat,\kappa} \neq q^*_{stat,\iota}$, then
there exists a $q_0^*$ between (or equal to) $q^*_{stat,\kappa}$ and $q^*_{stat,\iota}$ such that $\kappa_I(q^*_o)=\iota_I(q^*_o)$:\\}
We begin by proving that there exists a point $ q^*_o $ between $ q^*_{stat,\kappa} $ and $ q^*_{stat,\iota} $ such that $ \kappa_I(q^*_o)=\iota_I(q^*_o) $. 

If $ q^*_{stat,\kappa}=q^*_{stat,\iota} $, then $ q^*_o=q^*_{stat,\kappa}=q^*_{stat,\iota} $ and the results follow, so we assume $ q^*_{stat,\kappa}\neq q^*_{stat,\iota} $. Then we will show
\begin{align} \label{inequality:kappa:iota}
(\kappa_I-\iota_I)(q^*_{stat,\kappa})>0,(\kappa_I-\iota_I)(q^*_{stat,\iota})<0
\end{align} 
Once we show it, by the intermediate value theorem, we know there exists $ q^*_o $ between $ q^*_{stat,\kappa} $ and $ q^*_{stat,\iota} $ such that $ \kappa_I(q^*_o)=\iota_I(q^*_o) $. 

To show (\ref{inequality:kappa:iota}), we use the result in Part (2): $ \max_{q^*}\kappa_I(q^*)=\max_{q^*}\iota_I(q^*) $. 
\begin{align*}
\kappa_I(q^*_{stat,\kappa})=\max_{q^*}\kappa_I(q^*)=\max_{q^*}\iota_I(q^*)>\iota_I(q^*_{stat,\kappa})\\
\iota_I(q^*_{stat,\iota})=\max_{q^*}\iota_I(q^*)=\max_{q^*}\kappa_I(q^*)>\kappa_I(q^*_{stat,\iota})
\end{align*}

\paragraph{Part (5) Show that $ \xi(k^{sup}_I(q^*_o),q^*_o)=\max_{q^*}\xi(k^{sup}_I(q^*),q^*)$:\\}

Now we prove that $ \kappa_I(q^*_o)=\iota_I(q^*_o) $ maximizes $ \min\{\kappa_I(q^*),\iota_I(q^*)\} $. Without loss of generality, we assume $q^*_{stat,\kappa}\leq q^*_{stat,\iota}$. We prove it by contradiction:

We assume there exists $ q^*_\pi\neq q^*_o $ such that both $ \kappa_I(q^*_\pi) $ and $ \iota_I(q^*_\pi) $ are greater than $ \kappa_I(q^*_o)=\iota_I(q^*_o) $. Then we show the contradiction in two cases:\\
Case 1: $ q^*_\pi>q^*_o $\\
We have $q^*_\pi>q^*_o\geq q^*_{stat,\kappa}$. So $\kappa_I(q^*)$ is decreasing at $[q^*_o, q^*_\pi]$, which is a contradiction.\\
Case 2: $ q^*_\pi<q^*_o $\\
We have $q^*_\pi<q^*_o\leq q^*_{stat,\iota}$. So $\iota_I(q^*)$ is increasing at $[q^*_\pi, q^*_o]$, which is a contradiction.\\

\paragraph{Part (6) Show that $NG(k^{sup}_a(q^*_{a,o}),q^*_{a,o})= \mathcal{M}_1(Q);NG(k^{sup}_b(q^*_{b,o}),q^*_{b,o})= \mathcal{M}_2(Q);NG(k^{sup}_a(q(0|0)) + \epsilon, q(0|0) + \epsilon)= \mathcal{M}_3(Q, \epsilon)$ :\\}

%\smallskip\\
We first show $
NG(k^{sup}_I(q^*_{I,o}),q^*_{I,o})(0,1)=NG(k^{sup}_I(q^*_{I,o}),q^*_{I,o})(1,0),I=a,b $ to help calculate the mechanisms:

Notice that if $ \kappa_I(q^*_{I,o})=\iota_I(q^*_{I,o}) $ and $ q^*_{I,o}\neq q(1|0),q(1|1)  $, then by simplifying the equation, we get $$ (1+\frac{ k^{sup}_I(q^*_{I,o}) q(1|0)}{q(0|1)}) q^*_{I,o}=\frac{ k^{sup}_I(q^*_{I,o}) q(1|0)}{q(0|1)} $$
From (\ref{equation:NG}), we have
\begin{align*}
\Rightarrow& NG(k^{sup}_I(q^*_{I,o}),q^*_{I,o})(0,0)-NG(k^{sup}_I(q^*_{I,o}),q^*_{I,o})(0,1)\\
=& NG(k^{sup}_I(q^*_{I,o}),q^*_{I,o})(0,0)-NG(k^{sup}_I(q^*_{I,o}),q^*_{I,o})(1,0)
\end{align*}

\begin{align} \label{equation:symmetric}
\Rightarrow  NG(k^{sup}_I(q^*_{I,o}),q^*_{I,o})(0,1)=NG(k^{sup}_I(q^*_{I,o}),q^*_{I,o})(1,0) 
\end{align}

So the payoff function $NG(k^{sup}_I(q^*_{I,o}),q^*_{I,o})$ will be 
$\left( \begin{array}{ccc}
1 & & 0\\
 & & \\
0 & & *\\
\end{array} \right)$ or $ \left( \begin{array}{ccc}
* & & 0\\
 & & \\
0 & & 1\\
\end{array} \right)$, where $*$ will be calculated later. 

Now we will use this property to obtain $\mathcal{M}_1,\mathcal{M}_2$:

\paragraph{Show $NG(k^{sup}_a(q^*_{a,o}),q^*_{a,o})= \mathcal{M}_1(Q)$:\\}

Now we will calculate the coordinates of $q^*_{a,o}$.

We solve the equation $ \kappa_a(q^*)=\iota_a(q^*) $ to get a root between $q(1|0)$ and $q(1|1)$. We know that $\kappa_a(q(1|0))=\kappa_a(q(1|1))=\iota_a(q(1|0))=\iota_a(q(1|1))=0$. We also have two additional nontrivial roots $$root_1=\frac{q(1|0) q(1|1)-\sqrt{(q(1|0)-1) q(1|0) (q(1|1)-1) q(1|1)}}{q(1|0)+q(1|1)-1}$$ $$root_2=\frac{q(1|0) q(1|1)+\sqrt{(q(1|0)-1) q(1|0) (q(1|1)-1) q(1|1)}}{q(1|0)+q(1|1)-1}$$

But $root_2$ is greater than $q(1|1)$ since
\begin{align*}
root_2-q(1|1)&=\frac{\sqrt{(q(1|0)-1) q(1|0) (q(1|1)-1) q(1|1)}-q(1|1)^2+q(1|1)}{q(1|0)+q(1|1)-1}\\
&=\frac{\sqrt{q(0|0) q(1|0) q(0|1) q(1|1)}+q(1|1)q(0|1)}{q(1|0)+q(1|1)-1}\\
&=\frac{\sqrt{q(1|1)q(0|1)}(\sqrt{q(1|0)q(0|0)}+\sqrt{q(1|1)q(0|1)})}{q(1|1)-q(0|0)}>0.
\end{align*}
So $q^*_{a,o}=root_1 .$ 

Let $ \Delta^*_a=\sup_{q^*}\xi(k^{sup}_a(q^*),q^*) $. After substitutions, we have
\begin{align*}
&\sup_{q^*}\xi(k^{sup}_a(q^*),q^*)=\\
&\frac{q(0|1)(q(0|0)q(1|0)-\sqrt{q(0|0)q(1|0)q(0|1)q(1|1)})(\sqrt{q(0|0)q(1|0)q(0|1)q(1|1)}-q(0|1)q(1|1))}{(q(0|1)+q(1|0)q(1|1)(q(1|1)-q(0|0))(\sqrt{q(0|0)q(1|0)q(0|1)q(1|1)}+q(1|1)-q(0|0)-q(1|0)q(1|1))} \end{align*}

When $q^*=q^*_{a,o}$, by Claim~\ref{claim_l00}, we have $NG(k^{sup}_a(q^*_{a,o}),q^*_{a,o})(0,0)$ is the maximum which is $1$. We also have 
$ NG(k^{sup}_a(q^*_{a,o}),q^*_{a,o})(0,1)=NG(k^{sup}_a(q^*_{a,o}),q^*_{a,o})(1,0)=0 $. Since $q^*_{a,o}$ depends on the four values of $NG(k^{sup}_a(q^*_{a,o}),q^*_{a,o})$, given the three values and the coordinates of $q^*_{a,o}$, we can solve $NG(k^{sup}_a(q^*_{a,o}),q^*_{a,o})(1,1)=\sqrt{\frac{q(0|0) q(0|1)}{q(1|0) q(1|1)}}.$

So $$NG(k^{sup}_a(q^*_{a,o}),q^*_{a,o})=\left( \begin{array}{ccc}
\zeta(Q) & & 0\\
 & & \\
0 & & 1\\
\end{array} \right)$$
where $  \zeta(Q)=\sqrt{\frac{q(0|0) q(0|1)}{q(1|0) q(1|1)}} .  $

%Now we begin to calculate $M_3$:

\paragraph{Show $NG(k^{sup}_b(q^*_{b,o}),q^*_{b,o})= \mathcal{M}_2(Q)$:\\}

We solve the equation $ \kappa_b(q^*)=\iota_b(q^*) $. Besides $ q^*=q(0|0),q(1|1) $, we also have two additional nontrivial roots
$$root_1=\frac{-q(1|0)^2-\sqrt{(q(1|0)-1) q(1|0) (q(1|0)-q(1|1))
   (q(1|0)+q(1|1)-1)}+q(1|0)+q(1|1)-1}{q(1|1)-1} $$
$$root_2=\frac{-q(1|0)^2+\sqrt{(q(1|0)-1) q(1|0) (q(1|0)-q(1|1))
   (q(1|0)+q(1|1)-1)}+q(1|0)+q(1|1)-1}{q(1|1)-1}.$$
Notice that
\begin{align*}
root_2-q(1|0)&=-\frac{\sqrt{q(0|0) q(1|0) (q(1|1)-q(0|0)) (q(1|1)-q(1|0))}+q(0|0) (q(1|1)-q(0|0))}{q(0|1)}\\
&<0,
\end{align*}
so $q^*_{b,o}=root_1$. After substitutions, we have $$ \Delta^*_b=\frac{(q(1|1)-q(0|0))(q(0|0)q(1|0)(q(1|1)-q(0|1))-\sqrt{q(0|0)q(1|0)(q(1|1)-q(1|0))(q(1|1)-q(0|0))})}{(q(1|0)+q(0|1))q(1|0)q(1|1)q(0|0)} .$$ 
By a similar method to that used to obtain $\mathcal{M}_1$, we get

$$ NG(k^{sup}_b(q^*_{b,o}),q^*_{b,o})=\left( \begin{array}{ccc}
1 & & 0\\
 & & \\
0 & & \eta(Q)\\
\end{array} \right) $$
where $ \eta(Q)=\frac{1}{q(1|1)}(\sqrt{\frac{(q(1|1)-q(1|0))(q(1|1)-q(0|0))}{q(0|0)q(1|0)}}-q(0|1))   $

\paragraph{Show $NG(k^{sup}_a(q(0|0)) + \epsilon, q(0|0) + \epsilon)= \mathcal{M}_3(Q, \epsilon)$:\\}
To calculate $ NG(k^{sup}_a(q(0|0)+\epsilon),q(0|0)+\epsilon) $, we first prove that $NG(k^{sup}_a(q(0|0)+\epsilon),q(0|0)+\epsilon)(0,1)$ is the minimum $0$. Then combining with the fact that $NG(k^{sup}_a(q(0|0)+\epsilon),q(0|0)+\epsilon)(0,0)$ is the maximum $1$ according to Claim~\ref{claim_l00}, we only have two values left to calculate. We solve the rest of the values by $(k^{sup}_a(q(0|0)+\epsilon),q(0|0)+\epsilon)$ values and otain

$$NG(k^{sup}_a(q(0|0)+\epsilon),q(0|0)+\epsilon)=\left( \begin{array}{ccc}
\zeta(Q,\epsilon) & & \delta(Q,\epsilon)\\
 & & \\
0 & & 1\\
\end{array} \right) $$ where
$ \zeta(Q,\epsilon)=\frac{q(0|0) q(0|1)}{q(0|0) q(0|1) q(0|0)+\epsilon+q(1|0) (q(1|1)-q(1|1)
   q(0|0)+\epsilon)} $,\\
$ \delta(Q,\epsilon)=\frac{q(1|0) q(1|1) (q(0|0)+\epsilon-1)^2-q(0|0) q(0|1) q(0|0)+\epsilon^2}{q(0|0)+\epsilon (q(1|0)
      q(1|1) (q(0|0)+\epsilon-1)-q(0|0) q(0|1) q(0|0)+\epsilon)} $\\

%\newline

Now we begin to show that $NG(k^{sup}_a(q(0|0)+\epsilon),q(0|0)+\epsilon)(0,1)$ is the minimum $0$:

We first show that $ \kappa_{a}(q^*)>\iota_{a}(q^*) $ when $ q^*_{a,o} \leq q^* $. Once we show it, by a similar method to that used to prove (\ref{equation:symmetric}), we will get the result. It is now left to show  $ \kappa_{a}(q^*)>\iota_{a}(q^*) $:\\

Since there is only one intersection point $ q^*_{a,o} $ of $ \kappa_{a}(q^*) $ and $ \iota_{a}(q^*) $, so when $ q^*_{a,o} \leq q^* $, one of $ \kappa_{a}(q^*) $ and $ \iota_{a}(q^*) $ must be less than another. Notice that we have
$$\kappa_a(q(1|0) + q(1|1) - \frac{q(1|0) q(1|1)}{q^*})=\iota_a(q^*)$$
and
\begin{align*}
q(1|0) + q(1|1) - \frac{q(1|0) q(1|1)}{q^*}-q^*=\frac{1}{q^*}(q(1|1)-q^*)(q^*-q(1|0))>0.
\end{align*}
So $ \kappa_{a}(q^*)>\kappa_a(q(1|0) + q(1|1) - \frac{q(1|0) q(1|1)}{q^*})=\iota_{a}(q^*) $ when $ q^* $ is close to $ q(0|0) $ since $ \kappa_{a}(q^*) $ is decreasing now. Combining with the fact that when $ q^*_{a,o} \leq q^* $, one of $ \kappa_{a}(q^*) $ and $ \iota_{a}(q^*) $ must be less than another, we have $ \kappa_{a}(q^*)>\iota_{a}(q^*) $ when $ q^*_{a,o} \leq q^* $.

%So in $NG(k^{sup}_a(q^*),q^*)$, $$ \frac{1}{\alpha(k^{sup}_a(q^*),q^*)} = (1+\frac{ k^{sup}_a(q^*) q(1|0)}{q(0|1)}) q^* $$ when $ q^*_{a,o} \leq q^* $.  

\end{proof}

%We are going to show that the gap function only depends $ k $ and $ q^* $ of a payoff function $ PF $ if $ PF\in\mathcal{U} $:

%\begin{definition}
%We define $ G(\alpha,-k \alpha\frac{q(1|0)}{q(0|1)},q^*,\gamma)=G(\alpha,-\frac{k \alpha q(1|0)}{q(0|1)},q^*,\gamma) $
%\end{definition}

\else
All formal statements and proofs for this section appear in the full version. In this section, we first show that the gap $ \Delta_Q(PF) $ we want to optimize over payoff functions with payoffs in $ [0,1] $ region only depends on $ q^* $ and $ k $ where $ q^* $ is the ``break even" point and $ k $ is the slope of contours in $ R_{\tru} $ region in the mechanism under payoff function $ PF $. Then we fix $ q^* $ and tune $ k $ to optimize the gap.  Typically, the way that the gap is maximized corresponds to setting $k$ (the slope of contours in $ R_{\tru} $ region) to equalize the expected payments of the two equilibrium which correspond to adjacent extremal points to $\tru$ on the convex hull $\mathcal{H}$ of informative equilibrium translated into the $R_{\tru}$ quadrant of the best-response plot.  When $q^* > q(0|0)$ this this corresponds to making $\tru\one$ and $\tru\zer$ pay the same, and when $q^* < q(0|0)$ this this corresponds to making $\tru\one$ and $\fal$ pay the same.

After we finish that, the gap function will only be a function of $ q^* $ and at that time we will optimize the gap over $ q^* $.  The interesting part here is that there may be a discontinuity at $q^* = q(0|0)$.  If the optimum occurs when $q^* > q(0|0)$, then we are in $R_1$.  If the optimum occurs when $q^* < q(0|0)$, then we are in $R_2$.  If the optimum occurs at the boundary, then we are in $R_3$, and this accounts for the ``unattainable" region.  In the end, we use the best $ k,q^* $ to construct the mechanisms in Theorem~\ref{opt_thm}.

\fi

\section{Punishing all-$0$ and all-$1$ equilibria to complete the proof of the main theorem}\label{sec:punish}
Finally,  we would like to suitably punish the non-informative equilibria, so that, combining this with the mechanisms of Theorem~\ref{opt_thm}, truth-telling is focal with respect to all equilibria.
%In this section, we make $\tru$ focal when taking every equilibrium, including  $ \zer $ and $ \one$, into account.
%The idea is that we can eliminate the  $ \zer $ and $ \one $ equilibria by punishing all zero and all one reports with a large enough penalty.
By the following claim, if such a penalty only depends on other players' reports, then it does not affect the set of equilibria.  We therefore punish an agent if all the \emph{other} agents report the same signal.

\begin{claim}\label{shift}
Adding an arbitrary function of the other players' reports to the player payments does not
change the set of equilibria.
\end{claim}

\begin{proof}
Adding a term to agent $i$'s payoff that is only based on the actions of the other agents does not alter the set of equilibrium because the marginal gains/losses from unilateral deviation remain unchanged.
\end{proof}

There is a possible issue with this approach:  it might be the case that all players receive the same signal and that the truthful equilibrium will be penalized. Under certain conditions, this penalty will affect any informative equilibrium with low probability.  Thus the punishment's expected impact on informative equilibria will be overcome by the advantage of the truth-telling equilibrium's payments that we constructed in Theorem~\ref{opt_thm}. On the other hand, the $\zer$ and $\one$ equilibria will fully bear this punishment, and hence have lower payoff than truth-telling.

%We will use Claim~\ref{shift} to prove that we can make truthtelling focal among both informative and non-informative equilibria of the game by adding a function of the other players' reports to the payment of each player.

%\begin{theorem}(Main Theorem)\label{thm:main}
%Let $Q$ be a binary, symmetric, positively correlated and signal asymmetric prior, and let $\epsilon_Q$ be the maximum probability that a fixed set of $n-1$ agents receive the same signal (either all $\zer$ or all $\one$).  If $Q$ is attainable, let $PS^*$ be such that $\Delta_Q(PS^*) = \Delta^*(Q)$ and let $t = \nu^{PS^*}(\tru)$.  Otherwise let $\lim\sup_{\epsilon \rightarrow 0^+} \Delta_Q(PS^*(\epsilon)) = \Delta^*(Q)$  and let $t = \lim\sup_{\epsilon \rightarrow 0^+} \nu^{PS^*(\epsilon)}(\tru)$.If $\epsilon_Q < \frac{\Delta^*(Q)}{1 - t +\Delta^*(Q)}$  then there is a mechanism that makes truth-telling focal.\end{theorem}
{
\renewcommand{\thetheorem}{\ref{thm:focal_main}}

\begin{theorem}(Main Theorem (Restated))
Let $Q$ be a binary, symmetric, positively correlated and signal asymmetric prior, and let $\epsilon_Q$ be the maximum probability that a fixed set of $n-1$ agents receive the same signal (either all $\zer$ or all $\one$).  Then
\begin{enumerate}
\item $\Delta_Q(PPM(Q)) = \Delta^*(Q)$ when the prior $Q$ is attainable; $\lim\sup_{\epsilon \rightarrow 0^+} \Delta_Q(PPM(Q,\epsilon)) = \Delta^*(Q)$ when the prior $Q$ is unattainable. 
\item Let $t = \nu^{PPM(Q)}(\tru)$ when the prior $Q$ is attainable; let $t = \lim\sup_{\epsilon \rightarrow 0^+} \nu^{PPM(Q,\epsilon)}(\tru)$ when the prior is unattainable. If $\epsilon_Q < \frac{\Delta^*(Q)}{1 - t +\Delta^*(Q)}$, our $MPPM(Q)$ (or $MPPM(Q,\epsilon)$ when $Q$ is unattainable) makes truth-telling focal. 
\end{enumerate}
\end{theorem}
\addtocounter{theorem}{-1}
}

Note that once truth telling is focal, we can renormalise so that payments are between 0 and 1.  We also note that $\frac{\Delta^*(Q)}{1 - t +\Delta^*(Q)}$ only depends on $q(1|1)$ and $q(1|0)$\ifnum\fullversion=0---see full version for details\fi.

\begin{proof}
Part (1) has already been proved in Theorem~\ref{thm:optimization}. 

We start to prove part (2) now: 

Recall that we design the mechanism $MPPM(Q)$ (or $MPPM(Q,\epsilon)$)identical to $PPM(Q)$ (or $PPM(Q,\epsilon)$) except that we will issue a punishment of $ p = \frac{1-t}{2(1-\epsilon_Q)} + \frac{\Delta^*(Q)}{2 \epsilon_Q}$ to an agent if all the other agents play all $\zer$ or all $\one$.

By our assumption on $\epsilon_Q$, we know that $\frac{1-t}{1-\epsilon_Q} < \frac{\Delta^*(Q)}{ \epsilon_Q}$ .  Because $p$ is the average of these two values, we have $\frac{1-t}{1-\epsilon_Q} < p < \frac{\Delta^*(Q)}{ \epsilon_Q}$.

By Claim~\ref{shift} the equilibrium of $MPPM(Q)$ (or $MPPM(Q,\epsilon)$) are the same as those of $PPM(Q)$ (or $PPM(Q,\epsilon)$).

On the one hand, the expected payment of truth-telling has decreased by at most $\epsilon_Q p$ because in the truth-telling equilibrium, all agents report their true signals, and for any set of $n-1$ agents, all these signals are identical with probability at most $\epsilon_Q$.  However, $\epsilon_Q p < \Delta^*(Q)$.  If $Q$ is unattainable let $\epsilon  =  \frac{ \Delta^*(Q) - \epsilon_Q}{2}$.  Then the payment for truth-telling in $MPPM(Q)$ (or $MPPM(Q,\epsilon)$  is still greater than the payment for any other non-informative equilibrium (note that the payments for all equilibria only decreased).

On the other hand, the payment for the all zero or one equilibrium is now at most $1 - p = 1 - (p - \epsilon_Q p) - \epsilon_Q p  < 1 - (1- t) - \epsilon_Q p =  t - \epsilon_Q p$.   And so truth-telling now pays more than the all 0 or all 1 equilibria.
\end{proof}

\section{Future Directions}

%\yuqing{This extension of setting should be checked. Maybe move it to the end as a open question? Here is the open question section you wrote before (I put it here for your reference): Several open questions remain.  The first is to generalize this beyond binary signals.  One difficulty in handling this case is that the number of non-truthtelling equilibrium grows exponentially in the signal size.   Another question is whether these results can be extended to the Bayesian Truth Serum~\cite{prelec2004bayesian} or Robust Bayesian Truth Serum~\cite{witkowski2012robust} mechanisms (which have relaxed common prior assumptions).    }

Extending our mechanism to a more general setting (e.g. non-binary setting, asymmetric priors, and mechanisms where the prior is not known) are interesting and challenging directions for further work. We briefly discuss the challenges in these directions. In the non-binary setting, the number of equilibria increases exponentially (this can likely be handled via extensions of the current techniques).  However, the space of proper scoring rules also becomes more complicated.   We hope that with the correct generalization of our main technical tool, the best response plots, our results can extend to this case as well.  We are also hopeful that our results may extend to asymmetric priors (when the positive correlation requirement holds), but this requires additional technical work.  We do not (yet) see inherent barriers to this extension.  Certainly, removing the common prior assumption would make the mechanism more realistic. However, we note that if the prior were not known to the mechanism, then results as strong as ours would not be possible. This is because the agents can always permute their signals; if the prior were not known by the mechanism, such a strategy would always pay at least as well as truth-telling.

%\section{Open Questions}
%\katrina{should we remove this, given the section I'm suggestion in section 2?}
%\input{futurework}

%\section{Making truth-telling focal among all equilibria}

%\input{zero-one_elimination}

%\section{Extension to Higher Dimensions}	

%\input{higher_dimension}

%\section*{Acknowledgments}
%\begin{acks}
%We would like to thank David Parkes and Yiling Chen for helpful conversations and pointers, Juba Ziani for technical discussions, and also Aaron Roth and Arpita Ghosh who both contributed greatly to early discussions of this project.
%\end{acks}

% New style as of March 2012
%\bibliographystyle{ACM-Reference-Format-Journals}
%\bibliographystyle{abbrv}
\bibliographystyle{acmsmall}
\bibliography{peer,refs}

\begin{thebibliography}{}

\bibitem[\protect\citeauthoryear{Cr\'{e}mer and McLean}{Cr\'{e}mer and
  McLean}{1985}]{CremerM85}
{\sc Cr\'{e}mer, J.} {\sc and} {\sc McLean, R.~P.} 1985.
\newblock Optimal selling strategies under uncertainty for a discriminating
  monopolist when demands are interdependent.
\newblock {\em Econometrica\/}~{\em 53,\/}~2, 345--361.

\bibitem[\protect\citeauthoryear{Cr\'{e}mer and McLean}{Cr\'{e}mer and
  McLean}{1988}]{CremerM88}
{\sc Cr\'{e}mer, J.} {\sc and} {\sc McLean, R.~P.} 1988.
\newblock Full extraction of the surplus in bayesian and dominant strategy
  auctions.
\newblock {\em Econometrica\/}~{\em 56,\/}~6, 1247--1257.

\bibitem[\protect\citeauthoryear{Dasgupta and Ghosh}{Dasgupta and
  Ghosh}{2013}]{dasgupta2013crowdsourced}
{\sc Dasgupta, A.} {\sc and} {\sc Ghosh, A.} 2013.
\newblock Crowdsourced judgement elicitation with endogenous proficiency.
\newblock In {\em Proceedings of the 22nd international conference on World
  Wide Web}. International World Wide Web Conferences Steering Committee,
  319--330.

\bibitem[\protect\citeauthoryear{d'Aspremont and G{\'e}rard-Varet}{d'Aspremont
  and G{\'e}rard-Varet}{1982}]{dAspremontG1982}
{\sc d'Aspremont, C.} {\sc and} {\sc G{\'e}rard-Varet, L.-A.} 1982.
\newblock Bayesian incentive compatible beliefs.
\newblock {\em Journal of Mathematical Economics\/}~{\em 10,\/}~1, 83--103.

\bibitem[\protect\citeauthoryear{d'Aspremont and Gérard-Varet}{d'Aspremont and
  Gérard-Varet}{1979}]{dAspremontG1979}
{\sc d'Aspremont, C.} {\sc and} {\sc Gérard-Varet, L.-A.} 1979.
\newblock Incentives and incomplete information.
\newblock {\em Journal of Public Economics\/}~{\em 11,\/}~1, 25--45.

\bibitem[\protect\citeauthoryear{Faltings, Jurca, Pu, and Tran}{Faltings
  et~al\mbox{.}}{2014}]{faltings2014incentives}
{\sc Faltings, B.}, {\sc Jurca, R.}, {\sc Pu, P.}, {\sc and} {\sc Tran, B.~D.}
  2014.
\newblock Incentives to counter bias in human computation.
\newblock In {\em Second AAAI Conference on Human Computation and
  Crowdsourcing}.

\bibitem[\protect\citeauthoryear{Gao, Mao, Chen, and Adams}{Gao
  et~al\mbox{.}}{2014}]{GaoMCA2014}
{\sc Gao, X.~A.}, {\sc Mao, A.}, {\sc Chen, Y.}, {\sc and} {\sc Adams, R.~P.}
  2014.
\newblock Trick or treat: putting peer prediction to the test.
\newblock In {\em Proceedings of the fifteenth ACM conference on Economics and
  computation}. ACM, 507--524.

\bibitem[\protect\citeauthoryear{Ghosh, Ligett, Roth, and Schoenebeck}{Ghosh
  et~al\mbox{.}}{2014}]{ghosh2014buying}
{\sc Ghosh, A.}, {\sc Ligett, K.}, {\sc Roth, A.}, {\sc and} {\sc Schoenebeck,
  G.} 2014.
\newblock Buying private data without verification.
\newblock In {\em Proceedings of the fifteenth ACM conference on Economics and
  computation}. ACM, 931--948.

\bibitem[\protect\citeauthoryear{Gneiting and Raftery}{Gneiting and
  Raftery}{2007}]{GR07}
{\sc Gneiting, T.} {\sc and} {\sc Raftery, A.~E.} 2007.
\newblock Strictly proper scoring rules, prediction, and estimation.
\newblock {\em Journal of the American Statistical Association\/}~{\em
  102,\/}~477, 359--378.

\bibitem[\protect\citeauthoryear{Goel, Reeves, and Pennock}{Goel
  et~al\mbox{.}}{2009}]{goelrp09}
{\sc Goel, S.}, {\sc Reeves, D.~M.}, {\sc and} {\sc Pennock, D.~M.} 2009.
\newblock Collective revelation: A mechanism for self-verified, weighted, and
  truthful predictions.
\newblock In {\em Proceedings of the 10th ACM conference on Electronic commerce
  (EC 2009)}.

\bibitem[\protect\citeauthoryear{Jurca and Faltings}{Jurca and
  Faltings}{}]{jurca2011incentives}
{\sc Jurca, R.} {\sc and} {\sc Faltings, B.}
\newblock Incentives for answering hypothetical questions.
\newblock In {\em Proceedings of the 1st Workshop on Social Computing and User
  Generated Content (SC 2011)}. ACM.

\bibitem[\protect\citeauthoryear{Jurca and Faltings}{Jurca and
  Faltings}{2006}]{jurcafaltings06}
{\sc Jurca, R.} {\sc and} {\sc Faltings, B.} 2006.
\newblock Minimum payments that reward honest reputation feedback.
\newblock In {\em Proceedings of the 7th ACM conference on Electronic commerce
  (EC 2006)}.

\bibitem[\protect\citeauthoryear{Jurca and Faltings}{Jurca and
  Faltings}{2007}]{jurca2007collusion}
{\sc Jurca, R.} {\sc and} {\sc Faltings, B.} 2007.
\newblock Collusion-resistant, incentive-compatible feedback payments.
\newblock In {\em Proceedings of the 8th ACM conference on Electronic
  commerce}. ACM, 200--209.

\bibitem[\protect\citeauthoryear{Jurca and Faltings}{Jurca and
  Faltings}{2008}]{jurcafb08}
{\sc Jurca, R.} {\sc and} {\sc Faltings, B.} 2008.
\newblock Incentives for expressing opinions in online polls.
\newblock In {\em Proceedings of the 9th ACM conference on Electronic commerce
  (EC 2008)}.

\bibitem[\protect\citeauthoryear{Jurca and Faltings}{Jurca and
  Faltings}{2009}]{jurcafaltings09}
{\sc Jurca, R.} {\sc and} {\sc Faltings, B.} 2009.
\newblock Mechanisms for making crowds truthful.
\newblock {\em J. Artif. Int. Res.\/}~{\em 34,\/}~1.

\bibitem[\protect\citeauthoryear{Kamble, Shah, Marn, Parekh, and
  Ramachandran}{Kamble et~al\mbox{.}}{2015}]{kamble2015truth}
{\sc Kamble, V.}, {\sc Shah, N.}, {\sc Marn, D.}, {\sc Parekh, A.}, {\sc and}
  {\sc Ramachandran, K.} 2015.
\newblock Truth serums for massively crowdsourced evaluation tasks.
\newblock {\em arXiv preprint arXiv:1507.07045\/}.

\bibitem[\protect\citeauthoryear{Lambert and Shoham}{Lambert and
  Shoham}{2008}]{lambert2008truthful}
{\sc Lambert, N.} {\sc and} {\sc Shoham, Y.} 2008.
\newblock Truthful surveys.
\newblock {\em Proceedings of the 3rd International Workshop on Internet and
  Network Economics (WINE 2008)\/}.

\bibitem[\protect\citeauthoryear{Miller, Resnick, and Zeckhauser}{Miller
  et~al\mbox{.}}{2005}]{MRZ05}
{\sc Miller, N.}, {\sc Resnick, P.}, {\sc and} {\sc Zeckhauser, R.} 2005.
\newblock Eliciting informative feedback: The peer-prediction method.
\newblock {\em Management Science\/}, 1359--1373.

\bibitem[\protect\citeauthoryear{Prelec}{Prelec}{2004}]{prelec2004bayesian}
{\sc Prelec, D.} 2004.
\newblock A {B}ayesian {T}ruth {S}erum for subjective data.
\newblock {\em Science\/}~{\em 306,\/}~5695, 462--466.

\bibitem[\protect\citeauthoryear{Radanovic and Faltings}{Radanovic and
  Faltings}{2015}]{radanovic2015incentive}
{\sc Radanovic, G.} {\sc and} {\sc Faltings, B.} 2015.
\newblock Incentive schemes for participatory sensing.
\newblock In {\em Proceedings of the 2015 International Conference on
  Autonomous Agents and Multiagent Systems}. International Foundation for
  Autonomous Agents and Multiagent Systems, 1081--1089.

\bibitem[\protect\citeauthoryear{Riley}{Riley}{2014}]{riley2014minimum}
{\sc Riley, B.} 2014.
\newblock Minimum truth serums with optional predictions.
\newblock In {\em Proceedings of the 4th Workshop on Social Computing and User
  Generated Content (SC14)}.

\bibitem[\protect\citeauthoryear{Witkowski and Parkes}{Witkowski and
  Parkes}{2012}]{witkowski2012peer}
{\sc Witkowski, J.} {\sc and} {\sc Parkes, D.~C.} 2012.
\newblock Peer prediction without a common prior.
\newblock In {\em Proceedings of the 13th ACM Conference on Electronic
  Commerce}. ACM, 964--981.

\bibitem[\protect\citeauthoryear{Witkowski and Parkes}{Witkowski and
  Parkes}{2013}]{witkowski2013learning}
{\sc Witkowski, J.} {\sc and} {\sc Parkes, D.~C.} 2013.
\newblock Learning the prior in minimal peer prediction.
\newblock In {\em Proceedings of the 3rd Workshop on Social Computing and User
  Generated Content at the ACM Conference on Electronic Commerce}. Citeseer,
  14.

\end{thebibliography}

\ifnum\fullversion=1
\appendix

\section{Additional examples and observations related to scoring rules}\label{sec:appendix}

\begin{example}[Example of Proper Scoring Rule]\label{brier}
The Brier Scoring Rule for predicting a binary event is defined as follows. Let $I$ be the indicator random variable for the binary event to be predicted. Let $q$ be the predicted probability of the event occurring. Then:
$$B(I, q) = 2I\cdot q + 2(1-I)\cdot (1-q) - q^2 - (1-q)^2.$$
Note that if the event occurs with probability $p$, then the expected payoff of reporting a guess $q$ is (abusing notation slightly):
$$B(p,q) =2 p\cdot q + 2(1-p)\cdot (1-q) - q^2 - (1-q)^2 = 1 - 2(p - 2p\cdot q + q^2)$$
This is (uniquely) maximized when $p = q$, and so the Brier scoring rule is a strictly proper scoring rule. Note also that $B(p,q)$ is a linear function in $p$. Hence, if $p$ is drawn from a distribution, we have: $\E_p[B(p,q)] = B(\E[p],q)$, and so this is also maximized by reporting $q = \E[p]$.

A slight generalization of the Brier Scoring Rule is the ``Shifted Brier Scoring rule", which also takes a parameter $c \in \mathbb{R}$.  We write $B_c(p, q) = B(p - c, q - c)$, so that both of the inputs are ``shifted" before the scoring rule is evaluated. The Shifted Brier Scoring rule is also a strictly proper scoring rule.
%While we do not do so in this paper, the output of the Brier scoring rule can also be changed be scaled by any positive number, and shifted by any number without affecting the results.
\end{example}
%We observe that for the Shifted Brier Scoring rule,
%\[B_c(\hat{b}_j, q(1|b_i)) = B_{-c}(\neg \hat{b}_j, q(0|b_i)),\]
 %so that 0 and 1 are interchangeable as long as $c$ is negated.

%\jz{We want the scoring rule to be defined on $\mathbb{R}^2$ so that it can be shifted as much as we want. We should make this clear somewhere in the paper. Note that in particular, this gets rid of logarithmic scoring rules - Done}

We will consider two types of transformations on scoring rules: affine transformations of the outputs and affine transformations of the inputs.   The former are of the form $PS(\cdot, \cdot) \rightarrow \lambda \cdot PS(\cdot, \cdot) + \eta$ of the scoring rule itself; the latter are of the form $PS(p,q) \rightarrow PS(\lambda \cdot p + \eta,\lambda \cdot q + \eta)$ of the input to a scoring rule.
For this second type of transformation to be well-defined, we require that our scoring rules $PS(p,q)$ are defined $\forall (p,q) \in \mathbb{R}^2$. For example, the Brier scoring rule $B(p,q) =1 - 2(p - 2p\cdot q + q^2)$ can easily be defined on $\mathbb{R}^2$, allowing us to consider - for instance - $c$-shifts for any $c \in \mathbb{R}$.

Recall that any strictly proper scoring rule $PS(p,q)$ is affine in its first parameter $p$, and so we can write $PS(p,q)=f(q) \cdot p + g(q)$, for some functions $f(\cdot)$ and $g(\cdot)$. Based on this observation, the following lemma gives a useful characterization of scoring rules: %defined on $\mathbb{R}^2$:
\gs{I do not see why this is not stated as an iff statement}
\begin{claim}\label{PS-char}
Let $PS(p,q)=f(q) \cdot p + g(q)$ be a strictly proper scoring rule defined on %$\mathbb{R}^2$
$\mathbb{R} \times  \mathbb{R}$, and suppose $f,g \in \mathcal{C}^2[\mathbb{R}]$. Then $f$ is an increasing function on $\mathbb{R}$, and $g'(p)=-f'(p) \cdot p$ $\forall p \in \mathbb{R}$. Additionally, for any increasing function $f$ in $\mathcal{C}^2[\mathbb{R}]$, if  $g'(p)=-f'(p) \cdot p$ $\forall p \in \mathbb{R}$, $PS(p,q)=f(q) \cdot p + g(q)$ is a proper scoring rule.
\end{claim}

\begin{proof}[Proof of Claim \ref{PS-char}]\label{PS-char-pf}
Since we want $q \rightarrow PS(p,q)$ to be maximized in $p$ for all $ p \in \mathbb{R}$, we need that $\forall p \in \mathbb{R}$, $\frac{\partial PS}{\partial q}(p,p)=0$ and $\frac{\partial^2 PS}{\partial q^2}(p,p)<0$.
\begin{align*}
\frac{\partial PS}{\partial q}(p,q)=f'(q) \cdot p + g'(q)
\end{align*}
so we need $g'(p)=-f'(p) \cdot p$ $\forall p \in \mathbb{R}$. But now, this holds in particular for $q$ and
\begin{align*}
\frac{\partial PS}{\partial q}(p,q)=f'(q) \cdot p - f'(q) \cdot q.
\end{align*}
It follows that
\begin{align*}
\frac{\partial^2 PS}{\partial q^2}(p,q)=f''(q) \cdot p - f''(q) \cdot q - f'(q)
\end{align*}
but now
\begin{align*}
\frac{\partial^2 PS}{\partial q^2}(p,p)=- f'(p) < 0 \; \forall p \in \mathbb{R}
\end{align*}
Therefore, $f$ must be an increasing function. This proves the first part of the lemma.
\\
\\Now, the second part of the lemma simply comes from the fact that there exists a function $g$ such that $g'(p)=-f'(p) \cdot p$, as $-f'(p) \cdot p$ is continuous and therefore integrable.
\end{proof}

%It is easy to see that for a strictly proper scoring rule, $f$ is an increasing function (see proof later).\jz{Maybe I should move the proof here, or maybe it's well known and we don't really need to prove it?} \katrina{I would prefer that properties of proper scoring rules move to preliminaries, but I do not have the token on that section right now. Also, let's be careful and say strictly increasing.}\jz{Seems clearer now that this has been proved in the previous claim. Actually I think the statement is now useless}
%\jz{are we using increasing vs strictly increasing or non-decreasing vs increasing in the whole paper?}

We use the fact that $f$ is increasing to establish a useful rewriting of proper scoring rules.
\begin{lemma}\label{PS-rewrite}
For any strictly proper scoring rule $PS(p,q)=f(q) \cdot p + g(q)$ defined on $\mathbb{R} \times \mathbb{R}$, there exists a function $\delta$: $\mathbb{R} \times \mathbb{R} \rightarrow \mathbb{R}$ such that we can rewrite
\begin{align*}
PS(p,q(1|0))&=f(q(1|0))( p- q^*(1))+\delta(q(1|0),q(1|1))\\
PS(p,q(1|1))&=f(q(1|1))( p- q^*(1))+\delta(q(1|0),q(1|1))
\end{align*}
\end{lemma}
%The proof appears in Appendix~\ref{PS-rewrite-pf}.

\begin{proof}[Proof of Claim~\ref{PS-rewrite}]\label{PS-rewrite-pf}
Recall, for any proper scoring rule, there exists $q^*(1)$ such that $PS(q^*(1),q(1|0))=PS(q^*(1),q(1|1))$. Thus,
\begin{align*}
f(q(1|0)) \cdot q^*(1)+g(q(1|0)) &= PS(q^*(1),q(1|0))\\
&=PS(q^*(1),q(1|1))\\
&=f(q(1|1)) \cdot q^*(1)+g(q(1|1)),
\end{align*}
and therefore,
\[q^*(1)=\frac{g(q(1|1))-g(q(1|0))}{f(q(1|0))-f(q(1|1))}.\]
Note that $f(q(1|0)) \neq f(q(1|1))$, as $f$ is strictly increasing.

Now we have
\begin{align*}
PS(p,q(1|0))&=f(q(1|0)) \cdot p + g(q(1|0))\\
&=f(q(1|0))(p+\frac{g(q(1|0))}{f(q(1|0)})\\
&=f(q(1|0))(p-q^*(1)+\frac{g(q(1|0))}{f(q(1|0)}+\frac{g(q(1|1))-g(q(1|0))}{f(q(1|0))-f(q(1|1))})\\
&=f(q(1|0))(p-q^*(1)+\frac{f(q(1|0))  g(q(1|1))-f(q(1|1))  g(q(1|0))}{f(q(1|0))(f(q(1|0))-f(q(1|1)))})\\
&=f(q(1|0))(p-q^*(1))+\frac{f(q(1|0))  g(q(1|1))-f(q(1|1))  g(q(1|0))}{f(q(1|0))-f(q(1|1))}
\end{align*}
By symmetry,
\[PS(p,q(1|1))=f(q(1|1))(p-q^*(1))+\frac{f(q(1|1))  g(q(1|0))-f(q(1|0))  g(q(1|1))}{f(q(1|1))-f(q(1|0))}.\]

We may thus take
\begin{align*}
\delta(q(1|0),q(1|1))&=\frac{f(q(1|1))  g(q(1|0))-f(q(1|0))  g(q(1|1))}{f(q(1|1))-f(q(1|0))}\\
&=\frac{f(q(1|0))  g(q(1|1))-f(q(1|1))  g(q(1|0))}{f(q(1|0))-f(q(1|1))},
\end{align*}
to complete the proof.
\end{proof}

%We now introduce notation that highlights the role of three constants that will emerge as the key defining parameters of any proper scoring rule.
\begin{notation}
We now introduce the shorthand
\[\ell(x,b) :=PS(x,q(1|b)).\]
%which we will use throughout Section~\ref{focal-section}. %below
For a given strictly proper scoring rule, there exist constants $\alpha,\beta,\gamma$ such that we we can rewrite
\begin{align*}
\ell(x,0)&=\beta(x-q^*(1))+\gamma\\
 \ell(x,1)&=\alpha(x-q^*(1))+\gamma.
 \end{align*}
where $\beta=f(q(1|0))$ and $\alpha=f(q(1|1))$, according to lemma~\ref{PS-rewrite}.
\end{notation}

The following lemma shows that almost any proper scoring rule remains a proper scoring rule when shifted, and can be shifted to accommodate any desired ratio of the $\alpha$ and $\beta$ values above.
\begin{lemma}\label{scoring}%\jz{This is the cleaner version of our shifting lemma}
For a strictly proper scoring rule $PS(p,q)=f(q) \cdot p + g(q)$ defined on $\mathbb{R} \times \mathbb{R}$ with $f \in \mathcal{C}[\mathbb{R}]$, the two following statements are equivalent:
\\i) $\forall r \geq 0$, there exist $\alpha, \beta, \lambda, \eta \in \mathbb{R}$ with
%$\alpha \neq \beta$ and
$\lambda \neq 0$ such that $r=-\frac{\alpha}{\beta}$ and $\tilde{PS}(p,q)=PS(\lambda \cdot p + \eta, \lambda \cdot q + \eta)$ is a strictly proper scoring rule with $\tilde{\ell}(x,0)=\alpha(x-q^*(1))+\gamma $ and $ \tilde{\ell}(x,1)=\beta(x-q^*(1))+\gamma$.
\\ii) $\exists x \in \mathbb{R}$ such that $f(x)=0$.
\end{lemma}
%The proof appears in Appendix~\ref{scoring-pf}.

\begin{proof}[Proof of Lemma \ref{scoring}]\label{scoring-pf}
%Throughout the proof, we write $\Delta q=q(1|1)-q(1|0).$ $\Delta q >0$, because of Assumption~\ref{assn:signals-correlated}.
%We consider affine transformations on the parameters: $\tilde{PS}(p,q)=PS(\lambda \cdot p + \eta, \lambda \cdot q + \eta)$ for some $\lambda \neq 0, \eta \in \mathbb{R}$.
Note that if $PS$ is a strictly proper scoring rule, then $\tilde{PS}$ is a proper scoring rule with $\alpha=f(\lambda \cdot q(1|0) + \eta)$, $\beta=f(\lambda \cdot q(1|1) + \eta)$.
%as long as $[\eta,\lambda +\eta]   \subseteq \mathbb{R}$.
Therefore, i) is equivalent to the following statement: for any $r> 0$, there exists $\lambda \neq 0, \eta$ such that $f(\lambda \cdot q(1|0) + \eta)/f(\lambda \cdot q(1|1) + \eta)=-r$. This is equivalent to requiring that $\exists x,y \in \mathbb{R}$ such that $x \neq y$ and $f(x)/f(y)=-r$: indeed, since $q(1|0) \neq q(1|1)$, any $(x,y) \in \mathbb{R} \times \mathbb{R}$ with $x \neq y$ is injectively mapped to a $(\lambda,\eta)$ pair with $\lambda \neq 0$. Now we just need to prove that for $f$ continuous, $\exists x,y \in \mathbb{R}$ such that $x \neq y$ and $f(x)/f(y)=-r$, $\forall r >0$ iff $\exists x$ such that $f(x)=0$.

Suppose $f(x)=0$, let $y_1,y_2$ be such that $y_1 < x < y_2$ and $f(y_2) = - f(y_1) >0$. $z \rightarrow f(z)/f(y_2)$ is continuous on $[y_1,x]$, $f(y_1)/f(y_2)=-1$ and $f(x)/f(y_2)=0$, so any $r \in [0,1]$ can be attained. Similarly,  $z \rightarrow f(y_2)/f(z)$ is continuous on $[y_1,x[$ (f is strictly increasing and so can be $0$ at only one point, $x$), $f(y_2)/f(y_1)=-1$ and $f(y_2)/f(z) \rightarrow -\infty$ when $z \rightarrow x^-$, meaning any $r \in ]1,+\infty[$ can be attained; this concludes the backwards implication. The forward implication comes from the fact that if there is no $x$ such that $f(x) =0$, then $f$ cannot change sign over $\mathbb{R}$, by continuity. In particular, $\forall (x,y) \in \mathbb{R} \times \mathbb{R}$, $f(x)/f(y) >0$, and $f(x)/f(y)=-r$ with $r \geq 0$ is impossible.
\end{proof}

\gs{should be try around R.  Need to state that this applies to PS common rules}
Finally, we observe that the output of a proper scoring rule, when subjected to an affine shift, yields a proper scoring rule.
\begin{lemma}\label{outer}
%\jz{This is the ``outer" shifting lemma. It is implicitly used in section 4, we should state it }\jz{done}
Let $PS(p,q)$ be a proper scoring rule. Then $\forall \lambda >0, \eta$, $\tilde{PS}(p,q)=\lambda \cdot PS(p,q) + \eta$ is a proper scoring rule. Furthermore, writing $\tilde{PS}(p,q)=\tilde{f}(q) \cdot p + \tilde{g}(q)$, we have $\tilde{f}(q)=\lambda f(q)$,
%$\tilde{g}(q)=\lambda g(q)+\eta$
 and $\frac{f(q(1|0))}{f(q(1|1))}=\frac{\tilde{f}(q(1|0)}{\tilde{f}(q(1|1)}$.
\end{lemma}
\begin{proof}
As $\lambda >0$, $q \rightarrow PS(p,q)$ is maximized in $q=p$ iff $q \rightarrow \tilde{PS}(p,q)$ is maximized in $q=p$; this proves $\tilde{PS}$ is proper.
Now, $\tilde{PS}(p,q)=\lambda f(q) \cdot p+\lambda g(q)+\eta$ and $\tilde{f}=\lambda f(q)$, proving the lemma.
\end{proof}

\fi

\end{document}